\newif\ifreport\reporttrue

\documentclass[lettersize,journal]{IEEEtran}
\usepackage{setspace}
\usepackage{amsmath,amssymb,amsfonts}
\usepackage{amsthm}
\usepackage{graphicx}
\usepackage{textcomp}
\usepackage{xcolor}
\usepackage{caption}
\usepackage{subcaption}
\usepackage{mathtools}
\usepackage{algorithm}
\usepackage{algpseudocode}
\usepackage{bm,xstring}
\usepackage[multiple]{footmisc}
\usepackage{subcaption}
\usepackage{color}

\usepackage{tcolorbox}
\usepackage{fixltx2e}
\usepackage{tikz}
\usetikzlibrary{decorations.pathreplacing}
\usepackage[font=small,skip=0pt]{caption}

\DeclareMathOperator*{\argmax}{argmax}

\usepackage{color}

\newcommand{\ignore}[1]{}

\newtheorem{lemma}{Lemma}
\newtheorem{proposition}{Proposition}

\newtheorem{theorem}{Theorem}
\newtheorem{corollary}{Corollary}

\newtheorem{definition}{Definition}
\newtheorem{proof sketch}{Proof Sketch}

\newcommand{\EE}{\mathbb{E}}

\ifCLASSINFOpdf
\else
\fi

\begin{document}

\title{Remote Estimation of Gauss-Markov Processes over Multiple Channels: A Whittle Index Policy}

\author{Tasmeen Zaman Ornee, \IEEEmembership{Graduate Student Member,~IEEE,} and Yin Sun, \IEEEmembership{Senior Member,~IEEE}\\
\thanks{Tasmeen Zaman Ornee and Yin Sun are with Dept. of ECE,  Auburn University, Auburn, AL (emails: tzo0017@auburn.edu, yzs0078@auburn.edu). Part of this article was accepted by ACM MobiHoc 2023 \cite{Ornee2023AWhittle}.}
\thanks{This work was supported in part by the NSF under grant no. CNS-2239677 and by the ARO under grant no. W911NF-21-1-02444.}}



\maketitle

\begin{abstract}
In this paper, we study a sampling and transmission scheduling problem for multi-source remote estimation, where a scheduler determines when to take samples from multiple continuous-time Gauss-Markov processes and send the samples over multiple channels to remote estimators. The sample transmission times are \emph{i.i.d.} across samples and channels. The objective of the scheduler is to minimize the weighted sum of the time-average expected estimation errors of these Gauss-Markov sources. This problem is a continuous-time Restless Multi-armed Bandit (RMAB) problem with a continuous state space. We prove that the bandits are indexable and derive an exact expression of the Whittle index. To the extent of our knowledge, this is the first Whittle index policy for multi-source signal-aware remote estimation of Gauss-Markov processes. We further investigate signal-agnostic remote estimation and develop a Whittle index policy for multi-source Age of Information (AoI) minimization over parallel channels with \emph{i.i.d.} random transmission times. Our results unite two theoretical frameworks that were used for remote estimation and AoI minimization: threshold-based sampling and Whittle index-based scheduling. In the single-source, single-channel scenario, we demonstrate that the optimal solution to the sampling and scheduling problem can be equivalently expressed as both a threshold-based sampling strategy and a Whittle index-based scheduling policy. Notably, the Whittle index is equal to zero if and only if two conditions are satisfied: (i) the channel is idle, and (ii) the estimation error is precisely equal to the threshold in the threshold-based sampling strategy. Moreover, the methodology employed to derive threshold-based sampling strategies in the single-source, single-channel scenario plays a crucial role in establishing indexability and evaluating the Whittle index in the more intricate multi-source, multi-channel scenario.
Our numerical results show that the proposed policy achieves high performance gain over the existing policies when some of the Gauss-Markov processes are highly unstable.
\end{abstract}


\section{Introduction}
\IEEEPARstart{D}{ue} to the prevalence of networked control and cyber-physical systems, real-time estimation of the states of remote systems has become increasingly important for next-generation networks. For instance, a timely and accurate estimate of the trajectories of nearby vehicles and pedestrians is imperative in autonomous driving, and real-time knowledge about the movements of surgical robots is essential for remote surgery. In these examples, real-time system state estimation is of paramount importance to the performance of these networked systems. Other notable applications of remote state estimation include UAV navigation, factory automation, environment monitoring, and augmented/virtual reality.

To assess the freshness of system state information, one metric named \emph{Age of Information} (AoI) has drawn significant attention in recent years, e.g., \cite{KaulYatesGruteser-Infocom2012}, \cite{139341}. AoI is defined as the time difference between the current time and the generation time of the freshest received state sample. Besides AoI, nonlinear functions of the AoI have been introduced in \cite{Sun_2016_infocom}, \cite{sun2017update}, \cite{kosta2017age}, \cite{SunNonlinear2019} and illustrated to be useful as a metric of information freshness in sampling, estimation, and control \cite{yates2021age}, \cite{SunNonlinear2019}.

\begin{figure}
\vspace{0.1cm}
\centering
\includegraphics[width=8.5 cm]{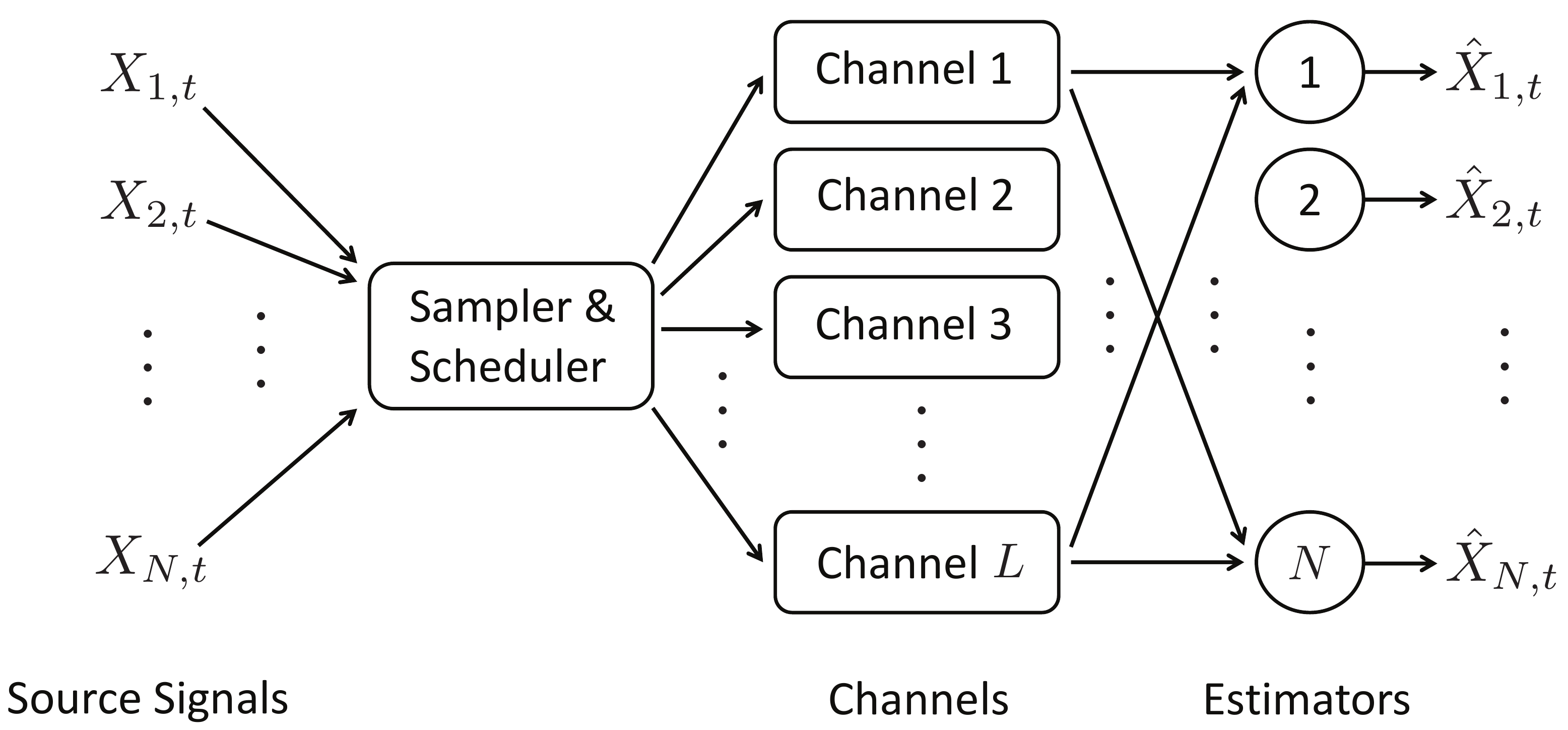}
\vspace{0.4cm}  
\caption{A multi-source, multi-channel remote estimation system.}\vspace{-0.0cm}
\label{fig_model}
\end{figure} 

In many applications, the system state of interest is in the form of a signal $X_t$, which may vary quickly at time $t$ and change slowly at a later time $t+\tau$ (even if the system state $X_t$ is Markovian and time-homogeneous). AoI, as a metric of the time difference, cannot precisely characterize how fast the signal $X_t$ varies at different time instants. To achieve more accurate system state estimation, it is important to consider \emph{signal-aware remote estimation}, where the signal sampling and transmission scheduling decisions are made using the historical \emph{realization} of the signal process $X_t$. Signal-aware remote estimation can achieve better performance than \emph{AoI-based, signal-agnostic remote estimation}, where the sampling and scheduling decisions are made using the probabilistic distribution of the signal process $X_t$, and the mean-squared estimation error can be expressed as a function of the AoI. The connection between signal-aware remote estimation and AoI minimization was first revealed in a problem of  sampling a Wiener process \cite{sun2017remote}. Subsequently, it was generalized to the case of (stable) Ornstein-Uhlenbeck (OU) process in \cite{Ornee_TON}. 

In many remote estimation and networked control systems, multiple sensors send their measurements (i.e., signal samples) to the destined estimators. For example, tire pressure, speed, and acceleration sensors in a self-driving vehicle send their data samples to the controller and nearby vehicles to make safe maneuvers \cite{KaulYatesGruteser-Infocom2012}. In this paper, we consider a remote estimation system with $N$ source-estimator pairs and $L$ channels, as illustrated in Figure \ref{fig_model}. Each source $n$ is a continuous-time Gauss-Markov process $X_{n,t}$, defined as the
solution of a Stochastic Differential Equation (SDE)
\begin{align} \label{SDE}
dX_{n,t} = \theta_n (\mu_n - X_{n,t})dt  + \sigma_n dW_{n,t},
\end{align}
where $\theta_n, \mu_n$, and $\sigma_n>0$ are the parameters of the Gauss-Markov process, and the $W_{n,t}$'s are independent Wiener processes. If $\theta_n > 0$, $X_{n,t}$ is  a stable Ornstein-Uhlenbeck (OU) process, which is the only nontrivial
continuous-time process that is stationary, Gaussian, and Markovian \cite{Doob1942}. If $\theta_n = 0$, then $X_{n,t} = \sigma_n W_{n,t}$ is a scaled Wiener process \cite{BMbook10}. If $\theta_n < 0$, we call $X_{n,t}$ an unstable Ornstein-Uhlenbeck (OU) process, because $\lim_{t \rightarrow \infty}\mathbb{E} [X_{n,t}^2] = \infty$ in this case. These Gauss-Markov processes can be used to model random walks \cite{nauta2019using}, interest rates \cite{nie2020sticky}, commodity prices \cite{Evans1994}, robotic swarms \cite{Kim2018}, 
biological processes \cite{bartoszek2017using}, control systems (e.g., the transfer of liquids or gases in and out of a tank) \cite{Nuno2011}, state exploration in deep reinforcement learning \cite{lillicrap2015continuous}, and etc.
A centralized sampler and scheduler decides when to take samples from the $N$  Gauss-Markov processes and send the samples over $L$ channels to remote estimators. At any time, at most $L$ sources can send samples over the channels. The samples experience \emph{i.i.d.} random transmission times over the channels due to interference, fading, etc. The $n$-th estimator uses causally received samples to reconstruct an estimate $\hat X_{n,t}$ of the real-time source value $X_{n,t}$.

Our objective is to find a sampling and transmission scheduling policy that minimizes the weighted sum of the time-average expected estimation errors of these Gauss-Markov sources. We develop a Whittle index policy to solve this problem. The technical contributions of this work are summarized as follows:
\begin{itemize}

\item We study the optimal sampling and transmission scheduling problem for the remote estimation of multiple continuous Gauss-Markov processes over parallel channels with \emph{i.i.d.} random transmission times. This problem
is a continuous-time Restless Multi-armed Bandit (RMAB) problem with a continuous state space, for which it is typically quite challenging to show indexability or to evaluate the Whittle index efficiently. We are able to prove indexability (see Theorem \ref{indexability}) and derive an exact expression for the Whittle index (Theorem \ref{Whittle index} and Lemma \ref{lemma1}). These results generalize prior studies on the remote estimation of a single Gauss-Markov process \cite{Wiener_TIT, Ornee_TON, Ornee_SPAWC} to the multi-source, multi-channel case. To the best of our knowledge, such results for multi-source remote estimation of Gauss-Markov processes were unknown before. Among the technical tools used to prove these results are Shiryaev’s free boundary method \cite{Peskir2006} for solving optimal stopping problems and Dynkin’s formula \cite{oksendal2013stochastic} for evaluating expectations involving stopping times.

\item We further investigate signal-agnostic remote estimation. In this context, the optimal sampling and scheduling problem becomes a multi-source AoI minimization problem over parallel channels with \emph{i.i.d.} random transmission times. We establish the indexability property and derive a precise expression of the Whittle index (Theorems \ref{aoi_indexability}-\ref{theorem3} and Lemma \ref{aoi_lemma}). Technically, these results carry forth and expand upon prior findings on Whittle index based AoI minimization \cite{tripathi2019whittle, hsu2018age, kadota2019scheduling}  in the following manner: In \cite{tripathi2019whittle, hsu2018age, kadota2019scheduling}, the transmission time remains constant, resulting in the optimality of the zero-wait sampling policy defined in \cite{yates2015lazy, sun2017update}. Consequently, the Whittle index derived in that case consistently maintains a non-negative value. In contrast, our results take into account scenarios involving \emph{i.i.d.} random transmission times. In such instances, the optimality of the zero-wait sampling policy is not guaranteed, leading to the possibility of both positive and negative values for the Whittle index.

\item Our results unite two important theoretical frameworks for remote estimation and AoI minimization: threshold-based sampling \cite{SunNonlinear2019, Wiener_TIT, Ornee_TON, Ornee_SPAWC} and Whittle index-based scheduling \cite{tripathi2019whittle, hsu2018age, kadota2019scheduling}. In the single-source, single-channel scenario, we demonstrate that the optimal solution to the sampling and scheduling problem can be expressed as both a threshold-based sampling strategy (\cite{Wiener_TIT, Ornee_TON, Ornee_SPAWC}) and a Whittle index-based scheduling policy (see Theorems \ref{single_theorem_whittle}, \ref{single_theorem_whittle_age}). Particularly noteworthy is that the Whittle index is equal to zero at time $t$ if and only if two conditions are satisfied: (i) the channel must be idle at time $t$, and (ii) the threshold condition is precisely met at time $t$. Moreover, the methodology used for deriving threshold-based sampling in the single-source, single-channel scenario plays a pivotal role in establishing indexability and evaluating the Whittle index in the more complex multi-source, multi-channel scenario.

\item Our numerical results show that the proposed policy
performs better than the signal-agnostic AoI-based Whittle index
policy and the Maximum-Age-First, Zero-Wait (MAF-ZW) policy.
The performance gain of the proposed policy is high when some of the Gauss-Markov processes are highly unstable.

\end{itemize}

\section{Related Work}

Remote state estimation has received considerable attention in numerous studies, e.g., see \cite{Rabi2012, Basar2014, Wiener_TIT, Ornee_TON, TsaiINFOCOM2020, Ornee_SPAWC, Wang2020, ahmad2009optimality}, and two recent surveys \cite{re_survey, vasconcelos2020survey}. Optimal sampling of one-dimensional and multi-dimensional Wiener processes with zero-delay, perfect channel was studied in \cite{Rabi2012,Basar2014}. A dynamic programming method was used in \cite{Rabi2012} to find the optimal sampling policy of the stable OU processes numerically for the case of zero-delay, perfect channel. A connection between remote estimation and AoI minimization was first reported in \cite{Wiener_TIT}, where optimal sampling strategies were obtained for the remote estimation of the Wiener process over a channel with \emph{i.i.d.} random transmission times. This study was further generalized to the case of the stable OU process in \cite{Ornee_TON}, where the optimal sampling strategy was derived analytically. 
In \cite{TsaiINFOCOM2020}, the authors considered remote estimation of the Wiener process with random two-way delay. 
When the system state follows a binary ON-OFF Markov process, Whittle index-based scheduling policies for remote estimation were developed in \cite{ahmad2009optimality}. Our study makes a contribution on the remote estimation of multiple Gauss-Markov processes (possibly with different distributions), by showing indexability and providing an analytical expression of the Whittle index.

Moreover, AoI-based scheduling for timely status updating has been studied extensively in, e.g., \cite{bedewy2020optimizing, he2017optimal, bedewy2021optimal, tripathi2019whittle, hsu2018age, tang2021whittle, yin2019only, kadota2018scheduling, kadota2019scheduling, shisher2022does, chen2022uncertainty}. A detailed survey on AoI was presented in \cite{yates2021age}. In \cite{he2017optimal}, the authors showed that under interference constraints, the scheduling problem for minimizing the age in wireless networks is NP-hard. In \cite{bedewy2020optimizing}, the authors minimized the 
weighted-sum peak AoI in a multi-source status updating system, subject to constraints on per-source battery lifetime. A joint sampling and scheduling problem for minimizing increasing AoI functions was considered in \cite{bedewy2021optimal}. AoI minimization in single-hop networks was considered in \cite{kadota2018scheduling}. AoI-based scheduling with timely throughput constraints was considered in \cite{kadota2019scheduling}. A Whittle index-based scheduling algorithm for minimizing AoI for stochastic arrivals was considered in \cite{hsu2018age}. In \cite{tang2021whittle}, \cite{tripathi2019whittle}, the Whittle index policy to minimize age functions for reliable and unreliable channels was proposed. 
A Whittle index policy for multiple source scheduling for binary Markov sources was studied in \cite{chen2022uncertainty}.
A Whittle index policy for signal-agnostic remote estimation was studied in \cite{Wang2020} for minimizing increasing AoI functions. In \cite{shisher2022does}, the authors proposed a Whittle index policy for minimizing non-monotonic AoI functions for both single and multi-actions scenarios. 
In the present paper, we propose a Whittle index policy for AoI-based, signal-agnostic remote estimation for \emph{i.i.d.} random transmission times. 

\section{Model and Formulation}
\subsection{System Model}
Consider a remote estimation system with $N$ source-estimator pairs and $L$ channels, which is shown in Figure \ref{fig_model}. Each source $n$ is a continuous-time Gauss-Markov process $X_{n,t}$, as defined in \eqref{SDE}. The sources are independent of each other and the parameters $\theta_n$, $\mu_n$, and $\sigma_n$ may vary across the sources. Hence, the $N$ sources could consist of scaled Wiener processes, stable OU processes, and unstable OU processes. A centralized sampler and transmission scheduler chooses when to take samples from the sources and transmit the samples over the channels to the associated remote estimators. At any given time, each source can be served by no more than one channel. In other words, if there are multiple samples from the same source waiting to be transmitted, only one of these samples can be transmitted over a single channel simultaneously.
Sample transmissions are \emph{non-preemptive}, i.e., once a channel starts to send a sample, it must finish transmitting the current sample before switching to serve another sample. Whenever a sample is delivered to the associated estimator, an acknowledgment (ACK) is immediately sent back to the scheduler. 

The operation of the system starts at time $t=0$. Let $S_{n,i}$ be the generation time of the $i$-th sample of source $n$, which satisfies $S_{n, i} \leq S_{n, i+1}$. This sample is submitted to a channel at time $G_{n,i}$, undergoes a random transmission time $Y_{n,i}$, and is delivered to the estimator $n$ at time $D_{n,i}$, where $S_{n,i} \leq G_{n,i}$, and $G_{n,i} + Y_{n,i} = D_{n,i}$.
Because (i) each source can be served by at most one channel at a time and (ii) the sample transmissions are non-preemptive, $D_{n,i} \leq G_{n, i+1}$.
The sample transmission times $Y_{n,i}$'s are \emph{i.i.d.} across samples and channels with mean $0 < \mathbb{E} [Y_{n,i}] < \infty$. In addition, we assume that the $Y_{n,i}$'s are independent of the Gauss-Markov processes $X_{n,t}$. The $i$-th sample packet $(S_{n,i}, X_{n, S_{n,i}})$ contains the sampling time $S_{n,i}$ and sample value $X_{n, S_{n,i}}$. Let $U_{n} (t) = \max_i \{S_{n,i} : D_{n,i} \leq t, i = 1,2, \ldots\}$ be the generation time of the freshest received sample from source $n$ at time $t$. The AoI of source $n$ at time $t$ is defined as \cite{KaulYatesGruteser-Infocom2012, 139341}
\begin{align}
\Delta_n (t) = t - U_{n} (t) = t - \max_i \{S_{n,i} : D_{n,i} \leq t, i=1,2,\ldots\}.
\end{align}
Because $D_{n, i} \leq D_{n, i+1}$, $\Delta_n (t)$ can also be expressed as
\begin{align} \label{age_eq}
\Delta_n (t) = t - S_{n, i}, \text{if} {\thinspace} t \in [D_{n, i}, D_{n, i+1}), {\thinspace} i = 0, 1, \ldots.
\end{align}
At time $t=0$, the initial state of the system satisfies $S_{n,0} = 0$, and  $D_{n,0} = Y_{n,0}$. 
The initial value of the Gauss-Markov process $X_{n,0}$ is finite. 

\subsection{MMSE Estimator}

At any time $t \geq S_{n,i}$, the Gauss-Markov process $X_{n,t}$ can be expressed as 
\begin{align}\label{eq_solution}
X_{n,t}  = 
\begin{cases}
& \begin{array}{l l}\!\!\!\!\!\!\!\!\! X_{n, S_{n,i}} e^{-\theta_n (t-S_{n,i})} + \mu_n \big[1-e^{-\theta_n (t-S_{n,i})} \big]\\ 
\!\!\!\!\!\!\!\!\! + \frac{\sigma_n}{\sqrt{2\theta_n}} {W_{n, 1-e^{-2 \theta_n (t-S_{n,i})}}},
& \!\!\!\!\!\!\!\!\!\!\!\!\!\!\!\!\!\!\text{ if }~ \theta_n > 0,\\
\!\!\!\!\!\!\!\!\!{\sigma_n} W_{n,t}, & \!\!\!\!\!\!\!\!\!\!\!\!\!\!\!\!\!\!\text{ if }~ \theta_n = 0,\\
\!\!\!\!\!\!\!\!\!X_{n, S_{n,i}} e^{-\theta_n (t-S_{n,i})} + \mu_n \big[1-e^{-\theta_n (t-S_{n,i})} \big] \\
\!\!\!\!\!\!\!\!\!+ \frac{\sigma_n}{\sqrt{-2\theta_n}} W_{n, e^{-2 \theta_n (t-S_{n,i})}-1}, & \!\!\!\!\!\!\!\!\!\!\!\!\!\!\!\!\!\!\text{ if }~ \theta_n < 0,
\end{array}
\end{cases}
\end{align}
where three expressions are provided for stable OU process ($\theta_n > 0$), scaled Wiener process ($\theta_n = 0$), and unstable OU process ($\theta_n < 0$), respectively. The first two expressions in \eqref{eq_solution} for the stable OU process and the scaled Wiener process were provided in \cite{Maller2009}. The third expression in \eqref{eq_solution} for the unstable OU process is proven in 
\ifreport
Appendix \ref{proof_unstable_OU_process}.
\else
the technical report \cite{Ornee2023} of the present paper.
\fi 

At time $t$, each estimator $n$ utilizes causally received samples to construct an estimate $\hat X_{n,t}$ of the signal value $X_{n,t}$. The information that is available at the estimator contains two parts: (i) $M_{n,t}\!=\!\{(S_{n,i}, X_{n, S_{n,i}}, G_{n,i}, D_{n,i}):D_{n,i} \leq t, i = 1, 2, \ldots\}$, which contains the sampling time $S_{n,i}$, sample value $X_{n, S_{n,i}}$, transmission starting time $G_{n,i}$, and the delivery time $D_{n,i}$ of the samples up to time $t$ and (ii) no sample has been received after the last  delivery time $\max_i \{D_{n,i}: D_{n,i}\leq t, i = 1, 2, \ldots\}$. Similar to \cite{Rabi2012, SOLEYMANI20161, Wiener_TIT, Ornee_TON}, we assume that the estimator neglects the second part of the information\footnote{This assumption can be removed by addressing the joint sampler and estimator design problem. In \cite{Hajek2008}, \cite{Nuno2011}, \cite{nayyar2013}, \cite{GAO201857}, \cite{ChakravortyTAC2020}, it was demonstrated that jointly optimizing the sampler and estimator in discrete-time systems yields the same optimal estimator expression, irrespective of the presence of the second part of information. This structural characteristic of the MMSE estimator, as highlighted in \cite[p. 619]{Hajek2008}, can also be established in continuous-time systems. 
The goal of this paper is to derive the closed-form expression for the optimal sampler while assuming this premise. To achieve the joint optimization of the sampler and estimator design, we can utilize majorization techniques previously developed in [4, 8, 11, 19, 23], as detailed in [10].}. 
If $t\in[D_{n,i},D_{n, i+1})$, the MMSE estimator is given by \cite{Ornee_SPAWC, Ornee_TON}
\begin{align}\label{mse}
\hat{X}_{n,t}\!\! 
= & \mathbb{E} [X_{n,t} | M_{n,t}] \nonumber\\
=\!& \left\{\!\!\!\!\! \begin{array}{l l} X_{n, S_{n,i}} \!e^{-\theta_n \!(t-S_{n,i})} \!+\! \mu_n\! \big[\!1\!\!-\!e^{-\theta_n\! (t-S_{n,i})} \!\big],
&\!\!\!\!\!\text{ if }~ \theta_n \neq 0,\\
{\sigma_n} W_{n, S_{n,i}}, &\!\!\!\!\!\text{ if }~ \theta_n = 0.
\end{array}\right. \!\!\!\!\!\!
\end{align}
The estimation error $\varepsilon_n (t)$ of source $n$ at time $t$ is given by
\begin{align}\label{est_err}
\varepsilon_n (t) = X_{n,t} - \hat{X}_{n,t}.
\end{align}
By substituting \eqref{eq_solution} and \eqref{mse} into \eqref{est_err}, if $t \in [D_{n,i} ,D_{n, i+1})$, then
\begin{align} \label{err_new}
\varepsilon_n (t) & =\left\{\!\! \begin{array}{l l} \frac{\sigma_n}{\sqrt{2\theta_n}} W_{n, 1-e^{-2 \theta_n (t-S_{n,i})}},
& \text{ if }~ \theta_n > 0,\\
{\sigma_n} (W_{n,t} - W_{n, S_{n,i}}), & \text{ if }~ \theta_n = 0,\\
\frac{\sigma_n}{\sqrt{-2\theta_n}} W_{n, e^{-2 \theta_n (t-S_{n,i})}-1}, & \text{ if }~ \theta_n < 0.
\end{array}\right. \!\!\!\!\!\!
\end{align}

\subsection{Problem Formulation}

Let $\pi = (\pi_n)_{n=1}^N$ denote a sampling and scheduling policy, where $\pi_n$= 
$((S_{n,1}, G_{n,1}), (S_{n,2}, G_{n,2}), \ldots)$ contains the sampling and transmission starting time instants of source $n$. Let $\pi_n$ denote a sub-sampling and scheduling policy for source $n$. In \emph{causal} sampling and scheduling policies, each sampling time $S_{n,i}$ is determined based on the up-to-date information that is available at the scheduler, without using any future information. Let $\Pi$ denote the set of all causal sampling and scheduling policies and let $\Pi_n$ denote the set of causal sub-sampling and scheduling policies for source $n$, both of which satisfy that (i) each source can be served by at most one channel at a time, and (ii) the sample transmissions are non-preemptive. At any time $t$, $c_n (t) \in \{0,1\}$ denotes the channel occupation status of source $n$. If source $n$ is being served by a channel at time $t$, then $c_n (t)=1$; otherwise, $c_n (t) =0$. Hence, if $t \in [G_{n,i}, D_{n,i})$, then $c_n (t) =1$. Because there are $L$ channels, $\sum_{n=1}^{N} c_n (t) \leq L$ is required to hold for all $t \geq 0$. 

Our objective is to find a causal sampling and scheduling policy for minimizing the weighted sum of the time-average expected estimation
errors of the $N$ Gauss-Markov sources.
This sampling and scheduling problem is formulated as
\begin{align} \label{problem}
& \inf_{\pi \in \Pi} \limsup_{T \to \infty}  \sum_{n=1}^{N} w_n \mathbb{E}_{\pi} \bigg[\frac{1}{T} \int_{0}^{T} \varepsilon^2 _{n} (t) dt\bigg] \\
& ~\text{s.t.} ~\sum_{n=1}^{N} \!c_n (t) \!\leq \!L, {\thinspace} c_n (t)\!\! \in \!\!\{0, 1\}, n=1,2,\ldots, N, t \in [0, \infty) \label{constraint},
\end{align} 
where $w_n > 0 $ is the weight of source $n$.
The sampling and scheduling policy $\pi$ can be simplified: 
\ifreport
In Appendix \ref{simplification_policy},
\else
In our technical report \cite{Ornee2023},
\fi
we prove that in the optimal policies to \eqref{problem}-\eqref{constraint}, the sampling time of the $i$-th sample $S_{n,i}$ and the transmission starting time of the $i$-th sample $G_{n,i}$ are equal to each other, i.e., $S_{n,i} = G_{n,i}$. Therefore, each sub-policy $\pi_n$ in $\pi$ can be simply denoted as $\pi_n =  (S_{n,1}, S_{n,2}, \ldots)$.


\section{Main Results} \label{main_results}

\subsection{Signal-aware Scheduling} \label{results_sig_aware}

Problem \eqref{problem}-\eqref{constraint} is a continuous-time Restless Multi-armed
Bandit (RMAB) with a continuous state space, where the estimation error $\varepsilon_n (t)$ of source $n$ is the state of the $n$-th restless bandit and each restless bandit is a Markov Decision Process (MDP) with two actions: active and passive. If a sample of source $n$ is taken and submitted to a channel at time $t$, we say that bandit $n$ takes an active action at time $t$; otherwise, bandit $n$ is made passive at time $t$. If a sample of source $n$ is in service, only the passive action is available for source $n$. 

An efficient approach for solving RMABs is to develop a low-complexity scheduling algorithm by leveraging the Whittle index theory \cite{weber1990index, whittle-restless}. If all the bandits are indexable
and certain technical conditions are satisfied, the Whittle index policy is asymptotically optimal as the number of bandits $N$ and the number of channels $L$ increases to infinity, keeping the ratio $L/N$ constant \cite{weber1990index}. 
In this section, we develop a Whittle index policy for solving problem (8)-(9) in three steps: (i) first, we relax the constraint \eqref{constraint} and utilize a Lagrangian dual approach to decompose the original problem into separated per-bandit problems; (ii) next, we prove that the per-bandit problems are indexable; and (iii) finally, we derive closed-form expressions for the Whittle index. Because the RMAB in \eqref{problem}-\eqref{constraint} has a continuous state space and requires continuous-time control, demonstrating indexability in Step (ii) and efficiently evaluating the Whittle index in Step (iii) are technically challenging. However, we are able to overcome these challenges.

\subsubsection{
Relaxation and Lagrangian Dual Decomposition} \label{decomposition}
{In standard restless multi-armed bandit problems, the channel resource constraint  needs to be satisfied with equality. In this paper, we consider a scenario where less than $L$ bandits can be activated at any time $t$, as indicated by the constraint \eqref{constraint}. Following \cite[Section 5.1.1]{verloop2016asymptotically}, we introduce $L$ additional \emph{dummy bandits} that will never change state and hence their estimation errors are 0. When a \emph{dummy bandit} is activated, it occupies one channel, but it does not incur any estimation error. Let $c_0 (t) \in\{0,1,2,\ldots,L\}$ denotes the number of \emph{dummy bandits} that are activated at time $t$. By considering \emph{dummy bandits}, the RMAB \eqref{problem}-\eqref{constraint} is equivalent to
\begin{align} \label{problem_dummy}
& \inf_{\pi \in \Pi} \limsup_{T \to \infty}  \sum_{n=1}^{N} w_n \mathbb{E}_{\pi} \bigg[\frac{1}{T} \int_{0}^{T} \varepsilon^2 _{n} (t) dt\bigg] \\
& ~\text{s.t.} ~\sum_{n=0}^{N} c_n (t) = L, c_0 (t) \in  \{0,1,\ldots, L\}, t \in [0, \infty), \nonumber\\ 
& \;\;\;\;\;\;\; c_n (t)\!\! \in \!\! \{0, 1\}, n=1,2,\ldots,N, t \in [0, \infty),  \label{constraint_dummy}
\end{align}
which is an RMAB with an equality constraint.

Following the standard relaxation and Lagrangian dual decomposition procedure in the Whittle index theory \cite{whittle-restless}, we relax the first constraint in \eqref{constraint_dummy} as 
\begin{align} \label{relatxed_constraint}
\limsup_{T \to \infty}  \sum_{n=0}^{N} \mathbb{E}_{\pi} \bigg[\frac{1}{T} \int_{0}^{T} c_n (t) dt\bigg] = L.
\end{align}
The relaxed constraint \eqref{relatxed_constraint} only needs to be satisfied on average, whereas \eqref{constraint_dummy} is required to hold at any time $t$.
Then, the RMAB \eqref{problem_dummy}-\eqref{constraint_dummy} is reformulated as
\begin{align} \label{problem_2}
& \inf_{\pi \in \Pi} \limsup_{T \to \infty}  \sum_{n=1}^{N} w_n \mathbb{E}_{\pi} \bigg[\frac{1}{T} \int_{0}^{T} \varepsilon_n ^2 (t) dt\bigg] \\
& ~\text{s.t.} ~\limsup_{T \to \infty} \sum_{n=0}^{N} \mathbb{E}_{\pi} \bigg[\frac{1}{T} \int_{0}^{T} c_n (t) dt\bigg] = L, \nonumber\\ \label{relaxed_ constraint_dummy_2}
& \;\;\;\;\;\;\; c_0 (t) \!\!\in \!\! \{0,1,\ldots, L\}\!, c_n (t)\!\! \in \!\! \{\!0, \!1\!\}\!, n\!\!=\!\!1,2,\ldots,N, t \!\!\in\!\! [0, \infty).
\end{align}} 

Next, we take the Lagrangian dual decomposition of the relaxed problem \eqref{problem_2}-\eqref{relaxed_ constraint_dummy_2}, which produces the following problem with a dual variable $\lambda \in \mathbb R$, also known as the activation cost \cite{whittle-restless}: 
 
\begin{align} \label{relaxed_problem}
\inf_{\pi \in \Pi} \limsup_{T \to \infty} & \mathbb{E}_{\pi}\! \bigg[\!\frac{1}{T}\!\! \int_{0}^{T} \!\!\sum_{n=1}^{N} w_n \varepsilon^2 _{n} (t) + \lambda \bigg(\sum_{n=0}^{N} c_n (t) - L\bigg) dt \bigg].
\end{align}
The term $\frac{1}{T} \int_{0}^{T} \sum_{n=0}^{N} \lambda L dt$ in \eqref{relaxed_problem} does not depend on policy $\pi$ and hence can be removed. Then, Problem \eqref{relaxed_problem} can be decomposed into {$(N+1)$ separated sub-problems. The sub-problem associated with source $n$ is 
\begin{align} \label{per_arm_problem}
\bar m_{n, \text{opt}} = \inf_{\pi_n \in \Pi_n} \limsup_{T \to \infty} \mathbb{E}_{\pi_n} \bigg[\frac{1}{T} \int_{0}^{T} w_n \varepsilon^2 _{n} (t) + \lambda c_n (t) dt\bigg],
\end{align}
where $\bar m_{n, \text{opt}}$ is the optimum value of \eqref{per_arm_problem} and $n=1,2,\ldots,N$.
On the other hand, the sub-problem associated with the \emph{dummy bandits} is given by
\begin{align} \label{dummy_problem}
\inf_{\pi_0 \in \Pi_0} \limsup_{T \to \infty} \mathbb{E}_{\pi_n} \bigg[\frac{1}{T} \int_{0}^{T} \lambda c_0 (t) dt\bigg],
\end{align}
where $\pi_0 = \{c_0 (t), t \in [0, \infty)\}$ and $\Pi_0$ is the set of all causal activation policies $\pi_0$.}

\subsubsection{Indexability} We now establish the indexability of the RMAB in \eqref{problem_dummy}-\eqref{constraint_dummy}. 
Let $\gamma_n (t) \in [0, \infty)$ denote the amount of time that has been used to send the current sample of source $n$ at time $t$. Here, if no sample from source $n$ is currently in service at time $t$, then $\gamma_n (t) =0$; if a sample from source $n$ is currently in service at time $t$, then $\gamma_n (t) > 0$. Consequently, if $\gamma_n (t) >0$, the active action is not available for source $n$ at time $t$.

Define $\Psi_n (\lambda)$ as the set of states $(\varepsilon, \gamma) \in \mathbb R \times [0,\infty)$ such that if $\varepsilon_n (t) = \varepsilon$ and $\gamma_n (t) = \gamma$, the optimal solution for \eqref{per_arm_problem} (or \eqref{dummy_problem} when $n=0$) is to take a passive action at time $t$.

\begin{definition} \textbf{(Indexability).} \label{def_1}
\cite{verloop2016asymptotically} Bandit $n$ is said to be \emph{indexable} if, as the activation cost $\lambda$ increases from $-\infty$ to $\infty$, the set $\Psi_n (\lambda)$ increases monotonically, i.e., $\lambda_1 \leq \lambda_2$ implies $\Psi_n (\lambda_1) \subseteq \Psi_n (\lambda_2)$.
The RMAB \eqref{problem_dummy}-\eqref{constraint_dummy} is said to be \emph{indexable} if all $(N+1)$ bandits are indexable.
\end{definition}

In general, establishing the indexability of an RMAB can be a challenging task. Because the per-bandit problem \eqref{per_arm_problem} is a continuous-time MDP with a continuous state space, determining the indexability of \eqref{per_arm_problem} appears to be quite formidable. In the sequel, we will utilize the techniques developed in our previous work \cite{Ornee_TON} to solve \eqref{per_arm_problem} precisely and analytically characterize the set $\Psi_n (\lambda)$. This analysis will allow us to demonstrate that \eqref{per_arm_problem} is indeed indexable.
Define
\begin{align} 
& Q(x) = \frac{\sqrt{\pi}}{2} \frac{e^{x^2}}{x} \text{erf} (x), \label{G}\\
& K(x) = \frac{\sqrt{\pi}}{2} \frac{e^{-x^2}}{x} \text{erfi} (x), \label{K}
\end{align}
where $\text{erf} (x)$ and $\text{erfi} (x)$ are the error function and imaginary error function, respectively, determined by \cite[Sec. 8.25]{Math_table_book}
\begin{align}
{\text{erf}}(x) = & \frac{2}{\sqrt \pi} \int_0^x e^{-t^2} dt, \\
{\text{erfi}}(x) = & \frac{2}{\sqrt \pi} \int_0^x e^{t^2} dt.
\end{align}
If $x = 0$, both $Q(x)$ and $K(x)$ are defined as their limits $Q(0) = \lim_{x \to 0} Q(x) =1$ and $K(0) = \lim_{x \to 0} K(x) =1$, respectively. Both $Q(\cdot)$ and $K(\cdot)$ are even functions. The function $Q(x)$ is strictly increasing on $x \in [0, \infty)$ and $Q(0) = 1$ \cite{Ornee_TON}. On the other hand, $K(x)$ is strictly decreasing on $x \in [0, \infty)$ and $K(0) = 1$ \cite{Ornee_SPAWC}. Hence, the inverse functions of $Q(x)$ and $K(x)$ are well defined on $x \in [0,\infty)$. The relation between these two functions is given by \cite{Ornee_SPAWC}
\begin{align}\label{KG}
K(x) = Q(jx),
\end{align}
where $j = \sqrt{-1}$ is the unit imaginary number. 


\begin{proposition} \label{opt_sampler_theorem}
If the $Y_{n,i}$'s are i.i.d. with $0<\mathbb{E}[Y_{n,i}] < \infty$, then $(S_{n,1} (\beta_n),S_{n,2} (\beta_n),\ldots)$ with a parameter $\beta_n$ is an optimal solution to \eqref{per_arm_problem}, 
where 
\begin{align}\label{eq_opt_solution}
S_{n, i+1} (\beta_n)= \inf_t \left\{ t \geq D_{n,i} (\beta_n):\! |\varepsilon_n (t)| \!\geq\! {v_n}(\beta_n)\right\},
\end{align}
$D_{n,i} (\beta_n)= S_{n,i} (\beta_n)+ Y_{n,i}$, ${v_n}(\beta_n)$ is defined by 
\begin{align} \label{threhsold}
v_n (\beta_n) =
\begin{cases}
& \begin{array}{l l}\!\!\!\!\!\!\!\!\! \frac{\sigma_n}{\sqrt{\theta_n}} Q^{-1} \bigg(\frac{w_n \frac{\sigma_n ^2}{2 \theta_n} \mathbb{E} [e^{-2 \theta_n Y_{n,i}}]}{w_n \frac{\sigma_n ^2}{2 \theta_n} - \beta_n}\bigg),
& \text{ if }~ \theta_n > 0,\\
\!\!\!\!\!\!\!\!\! \frac{1}{\sqrt{w_n}} \sqrt{3 (\beta_n - w_n \sigma_n ^2 \mathbb{E} [Y_{n,i}])}, & \text{ if }~ \theta_n = 0,\\
\!\!\!\!\!\!\!\!\!\frac{\sigma_n}{\sqrt{-\theta_n}} K^{-1} \bigg(\frac{w_n \frac{\sigma_n ^2}{2 \theta_n} \mathbb{E} [e^{-2 \theta_n Y_{n,i}}]}{w_n \frac{\sigma_n ^2}{2 \theta_n} - \beta_n}\bigg), & \text{ if }~ \theta_n < 0,
\end{array}
\end{cases}
\end{align}
$Q^{-1}(\cdot)$ and $K^{-1}(\cdot)$ are the inverse functions of $Q(x)$ in \eqref{G} and $K(x)$ in \eqref{K}, respectively, defined in the region of $x \in [0, \infty)$, and $\beta_n$ is the unique root of
\begin{align}\label{thm1_eq22}
\!\!\! & \mathbb{E}\left[\int_{D_{n,i} (\beta_n)}^{D_{n, i+1}(\beta_n)} w_n \varepsilon^2 _n (t) dt\right]\! - \!{\beta_n} {\mathbb{E}[D_{n, i+1}(\beta_n)\!-\!D_{n,i} (\beta_n)]} \nonumber\\
& + \lambda \mathbb{E} [Y_{n,i+1}] \!=\! 0.\!\!\! 
\end{align}
The optimal objective value to \eqref{per_arm_problem} is given by
\emph{\begin{align}\label{thm1_eq23}
{\bar m}_{n, \text{opt}}  = \frac{\mathbb{E}\left[\int_{D_{n,i} (\beta_n)}^{D_{n, i+1}(\beta_n)}\! w_n \varepsilon^2 _n (t) dt\right] + \lambda \mathbb{E} [Y_{n,i+1}]}{\mathbb{E}[D_{n, i+1}(\beta_n)\!-\!D_{n,i} (\beta_n)]}. {\noindent}
\end{align}}
{\!\!Furthermore}, $\beta_n$ is exactly the optimal objective value of \eqref{per_arm_problem}, i.e., \emph{${\beta_n} = {{\bar m}}_{n, \text{opt}}$}.
\end{proposition}

 \ifreport
\begin{proof}
See Appendix \ref{threshold_proof}.
\end{proof}
\else
The proof is provided in our technical report \cite{Ornee2023}.
\fi

Proposition \ref{opt_sampler_theorem} complements earlier optimal sampling results for the remote estimation of the Wiener process (i.e., the case of $\theta_n =0$ and $\lambda=0$) \cite{Wiener_TIT} and stable OU process (i.e., $\theta_n >0$ and $\lambda=0$) \cite{Ornee_TON}, by (i) adding a third case on unstable OU process (i.e., $\theta_n <0$) and (ii) incorporating an activation cost $\lambda \in \mathbb{R}$. 


{By using Proposition \ref{opt_sampler_theorem}, we can analytically characterize the set $\Psi_n (\lambda)$. To that end, we first show that the threshold $v_n (\beta_n)$ in \eqref{eq_opt_solution} is a function of the activation cost $\lambda$. For any given $\lambda$, $\beta_n$ is the unique root of equation \eqref{thm1_eq22}. Hence, $\beta_n$ can be expressed as an implicit function $\beta_n(\lambda)$ of $\lambda$, defined by equation \eqref{thm1_eq22}. Moreover, the threshold $v_n(\beta_n)$ can be rewritten as a function $v_n(\beta_n(\lambda))$ of the activation cost $\lambda$. 
According to \eqref{eq_opt_solution} and the definition of set $\Psi_n (\lambda)$, a point $(\varepsilon_n (t), \gamma_n(t)) \in \Psi_n (\lambda)$ if either (i) $\gamma_n (t) > 0$ such that a sample from source $n$ is currently in service at time $t$, or (ii) $|\varepsilon_n (t)| < v_n (\beta_n (\lambda))$ such that the threshold condition in \eqref{eq_opt_solution} for taking a new sample is not satisfied. By this, an analytical expression of set $\Psi_n (\lambda)$ is derived as 
\begin{align}  \label{index_set}
\Psi_n (\lambda) \!\!=\!\! \{\!(\varepsilon, \gamma)\! \in \!\mathbb{R} \!\times\! [0, \infty)\! : \!\gamma\!>\!0 {\thinspace} \text{or} {\thinspace} |\varepsilon| \!<\! v_n (\beta_n(\lambda))\!\}.
\end{align}

Using \eqref{index_set}, we can prove the first key result of the present paper: 


\begin{theorem} \label{indexability}
The RMAB problem \eqref{problem_dummy}-\eqref{constraint_dummy} is indexable.
\end{theorem}

{\it Proof sketch.} According to Proposition \ref{opt_sampler_theorem}, for any $\lambda$, the optimal solution to \eqref{per_arm_problem} is a threshold policy. Using this, we can show that the unique root $\beta_n (\lambda)$ of \eqref{thm1_eq22} is a strictly increasing function of $\lambda$. In addition, $v_n(\beta_n)$ in \eqref{threhsold} is a strictly increasing function of $\beta_n$. Hence, $v_n(\beta_n(\lambda))$ is a strictly increasing function of $\lambda$. Substituting this into \eqref{index_set}, if $\lambda_1 \leq \lambda_2$, then $\Psi_n (\lambda_1) \subseteq \Psi_n (\lambda_2)$. For the \emph{dummy bandits}, it is optimal in \eqref{dummy_problem} to activate a \emph{dummy bandit} when $\lambda < 0$. Hence, \emph{dummy bandits} are always indexable.
The details are provided in
\ifreport
Appendix \ref{proof_of_indexability}.
\else
our technical report \cite{Ornee2023}.
\fi 
$\square$


\subsubsection{Whittle Index Policy} \label{Whittle}

Next, we introduce the definition of the Whittle index.
 
\begin{definition} \label{def_2}
\cite{whittle-restless} If bandit $n$ is indexable, then the Whittle index $\alpha_n (\varepsilon, \gamma)$ of bandit $n$ at state $(\varepsilon, \gamma)$ is defined by
\begin{align} \label{Whittle_index_general_definition}
\alpha_n (\varepsilon, \gamma) = \inf_{\lambda \in \mathbb R} \{\lambda \in \mathbb{R} : (\varepsilon, \gamma) \in \Psi_n (\lambda)\},
\end{align}
which is the infimum of the activation cost $\lambda$ for which it is better not to activate bandit $n$.
\end{definition}



\begin{theorem} \label{Whittle index}
The following assertions are true for the Whittle index $\alpha_n (\varepsilon, \gamma)$ of problem \eqref{per_arm_problem} at state $(\varepsilon, \gamma)$:

(a)  If $\gamma = 0$, then the Whittle index $\alpha_n (\varepsilon, \gamma)$ is derived in the following three cases:

(i) \emph{Case 1:} If $\theta_n > 0$ (i.e., $X_{n,t}$ is a stable OU process), then
\begin{align} \label{Whittle_index}
& \alpha_n (\varepsilon, 0) = \nonumber\\
& \frac{w_n}{\mathbb{E} [Y_{n,i}]}\bigg\{\mathbb{E} [D_{n, i+1} (\varepsilon) - D_{n,i} (\varepsilon)] \frac{\sigma_n^2}{2 \theta_n} \bigg(1 - \frac{\mathbb{E} [e^{-2 \theta_n Y_{n,i}}]}{Q \big(\frac{\sqrt{\theta_n}}{\sigma_n}  \varepsilon\big)}  \bigg) \nonumber\\
& ~~~~~~~~~~- \mathbb{E} \bigg[\int_{D_{n,i} (\varepsilon)}^{D_{n, i+1} (\varepsilon)} \varepsilon_{n} ^2 (s) ds\bigg]\bigg\}, 
\end{align}

(ii) \emph{Case 2:} If $\theta_n = 0$ (i.e., $X_{n,t}$ is a scaled Wiener process), then
\begin{align}
& \alpha_n (\varepsilon, 0) = \nonumber\\
& \frac{w_n}{\mathbb{E} [Y_{n,i}]} \bigg\{\mathbb{E} [D_{n, i+1} (\varepsilon) - D_{n,i} (\varepsilon)] \bigg(\frac{\varepsilon ^2}{3} + \sigma_n ^2 \mathbb{E} [Y_{n,i}]\bigg)\nonumber\\
& ~~~~~~~~~~- \mathbb{E} \bigg[\int_{D_{n,i} (\varepsilon)}^{D_{n, i+1} (\varepsilon)} \varepsilon_{n} ^2 (s) ds\bigg]\bigg\}, \label{wiener_index}
\end{align}

(iii) \emph{Case 3:} If $\theta_n < 0$ (i.e., $X_{n,t}$ is an unstable OU process), then
\begin{align} 
& \alpha_n (\varepsilon, 0) = \nonumber\\
& \frac{w_n}{\mathbb{E} [Y_{n,i}]}\bigg\{ \mathbb{E} [D_{n, i+1} (\varepsilon) - D_{n,i} (\varepsilon)] \frac{\sigma_n^2}{2 \theta_n} \bigg(1 - \frac{\mathbb{E} [e^{-2 \theta_n Y_{n,i}}]}{K \big(\frac{\sqrt{-\theta_n}}{\sigma_n}  \varepsilon\big)}  \bigg) \nonumber\\
&~~~~~~~~~~- \mathbb{E} \bigg[\int_{D_{n,i} (\varepsilon)}^{D_{n, i+1} (\varepsilon)} \varepsilon_{n} ^2 (s) ds\bigg]\bigg\}, \label{Whitt_index_unstable}
\end{align}
where $Q(\cdot)$ and $K(\cdot)$ are defined in \eqref{G} and \eqref{K}, respectively. 

(b)  If $\gamma > 0$, then 
\begin{align} \label{Whitt_infty}
\alpha_n (\varepsilon, \gamma) = -\infty.
\end{align}
\end{theorem}

\vspace{-5pt}

{\it Proof sketch.} {When $\gamma =0$, by \eqref{index_set}, \eqref{Whittle_index_general_definition}, and the monotonicity of $v_n (\cdot)$ and $\beta_n (\cdot)$, 
the Whittle index $\alpha_n(\varepsilon, 0)$ is equal to the unique root $\lambda$ of equation 
\begin{align} \label{thm3_eq}
|\varepsilon| = v_n (\beta_n(\lambda)). 
\end{align}
Hence, $\alpha_n (\varepsilon, 0) = \beta_n ^{-1} (v_n ^{-1} (|\varepsilon|)$.
By substituting \eqref{threhsold} and \eqref{thm1_eq22} into \eqref{thm3_eq} and using the fact that $Q(\cdot)$ and $K(\cdot)$ are even functions, statement (a) in Theorem \ref{Whittle index} is proven. When $\gamma > 0$, $(\varepsilon, \gamma)$ is always in the set $\Psi_n (\lambda)$ for any $\lambda \in \mathbb R$. Hence, by using \eqref{Whittle_index_general_definition}, $\alpha_n (\varepsilon, \lambda) = -\infty$. By this, statement (b) in Theorem \ref{Whittle index} is proven.} The details are provided in
\ifreport
Appendix \ref{proof_of_Whittle_index}.
\else
our technical report \cite{Ornee2023}.
\fi 
$\square$

In Theorem \ref{Whittle index}, the delivery time $D_{n,i} (\varepsilon)$ is expressed as a function of $\varepsilon$ for the following reason: 
in the optimal solution to \eqref{per_arm_problem}, the sample delivery time is a function of the activation cost $\lambda$. If the activation cost $\lambda$ in \eqref{per_arm_problem} is chosen as $\lambda = \alpha_n(\varepsilon, \gamma)$, then the sample delivery time in the optimal solution to \eqref{per_arm_problem} is a function of $\varepsilon$. We use the notation $D_{n,i} (\varepsilon)$ to remind us that  the expectations $\mathbb{E} [D_{n, i+1} (\varepsilon) - D_{n,i} (\varepsilon)]$ and $\mathbb{E} [\int_{D_{n,i} (\varepsilon)}^{D_{n, i+1} (\varepsilon)} \varepsilon_{n} ^2 (s) ds]$ in \eqref{Whittle_index}-\eqref{Whitt_index_unstable} change as $\varepsilon$ varies. 

In order to compute the Whittle index $\alpha_n (\varepsilon, \gamma)$, we need to calculate the expectations $\mathbb{E} [D_{n, i+1} (\varepsilon) - D_{n,i} (\varepsilon)]$ and $\mathbb{E} [\int_{D_{n,i} (\varepsilon)}^{D_{n, i+1} (\varepsilon)} \varepsilon_{n} ^2 (s) ds]$ in \eqref{Whittle_index}-\eqref{Whitt_index_unstable}. Because $S_{n,i} (\varepsilon)$ and $D_{n,i} (\varepsilon)$ are stopping times of the process $X_{n,t}$, numerically evaluating these two expectations is nontrivial. This challenge can be addressed by resorting to Lemma \ref{lemma1} provided below, which is obtained by using Dynkin’s formula \cite[Theorem 7.4.1]{Bernt2000} to simplify expectations involving stopping times. 

To that end, let us introduce a Gauss-Markov process $O_{n,t}$ with a zero initial condition $O_{n,0} = 0 $ and parameter $\mu_n = 0$, which is expressed as 
\begin{align} \label{OU_shifted}
O_{n,t}  = 
& \left\{\!\! \begin{array}{l l}  \frac{\sigma_n}{\sqrt{2\theta_n}} W_{n, 1-e^{-2 \theta_n t}},
& \text{ if }~ \theta_n > 0,\\
{\sigma_n} W_{n,t}, & \text{ if }~ \theta_n = 0,\\
\frac{\sigma_n}{\sqrt{-2\theta_n}} W_{n, e^{-2 \theta_n t}-1}, & \text{ if }~ \theta_n < 0.
\end{array}\right. \!\!\!\!\!\!
\end{align}
By comparing \eqref{err_new} with \eqref{OU_shifted}, the estimation error process $\varepsilon_n(t)$ has the same distribution with the time-shifted Gauss-Markov process $O_{n, t- S_{n,i} (\varepsilon)}$, when $t \in [D_{n,i} (\varepsilon), D_{n,i+1}(\varepsilon))$. 

Then, we have the following lemma for computing the expectations in  \eqref{Whittle_index}, \eqref{wiener_index}, and \eqref{Whitt_index_unstable}. 
\begin{lemma} \label{lemma1}
In Theorem \ref{Whittle index}, it holds that

\begin{align}\label{eq_expectation1}
\mathbb{E}[D_{n, i+1} (\varepsilon) - D_{n,i} (\varepsilon)] = \mathbb{E} \big[R_{n,1} \big(\max\big\{|\varepsilon|, |O_{n, Y_{n,i}}|\big\}\big)\big],
\end{align}
\begin{align}\label{eq_expectation2}
\!\!\!\!\!\! & {\mathbb{E}\left[\int_{D_{n,i} (\varepsilon)}^{D_{n, i+1} (\varepsilon)}\! \varepsilon_n ^2 (s) ds\right]} \nonumber\\
=& \mathbb{E}\! \big[R_{n,2} \!\big(\!\max\!\big\{\!|\varepsilon|, \!|O_{n, Y_{n,i}}|\!\big\} \!+\! O_{n, Y_{n,i+1}}\big)\!\big] \!\!-\!\! \mathbb{E}\! \big[R_{n,2} \big(\!O_{n, Y_{n,i}}\!\big)\!\big],
\end{align}
where if $\theta_n \neq 0$, then
\begin{align}
& R_{n,1} (\varepsilon) = {\frac{\varepsilon^2}{\sigma_n^2}} \thinspace {}_2F_2\left(1,1;\frac{3}{2},2;\frac{\theta_n}{\sigma_n^2} \varepsilon^2\right), \label{eq_R_1}\\
& R_{n,2} (\varepsilon) = -\frac{\varepsilon^2}{2\theta_n} + {\frac{\varepsilon^2}{2\theta_n}} \thinspace {}_2F_2\left(1,1;\frac{3}{2},2;\frac{\theta_n}{\sigma_n^2}\varepsilon^2\right); \label{eq_R_2}
\end{align}
if $\theta_n = 0$, then
\begin{align} \label{lemma1_case21}
& R_{n,1} (\varepsilon) = \frac{\varepsilon^2}{\sigma_n ^2}, \\
& R_{n,2} (\varepsilon) = \frac{\varepsilon^4}{6\sigma_n ^2}.
\end{align}
\end{lemma}

\ifreport
\begin{proof}
See Appendix \ref{proof_lemma1}.
\end{proof}
\else
{}
\fi

In \eqref{eq_R_1} and \eqref{eq_R_2}, we have used the generalized hypergeometric function,  which is defined by \cite[Eq. 16.2.1]{olver2010nist}
\begin{align}
& {}{_pF_q} (a_1, a_2, \cdots, a_p; b_1, b_2, \cdots b_q; z)\nonumber\\
=& {\sum_{n=0}^{\infty}} {\frac{{(a_1)_n} {(a_2)_n} \cdots{(a_p)_n}}{{(b_1)_n}{(b_2)_n} \cdots {(b_p)_n}}}{\frac{z^n}{n!}},
\end{align}
where
\begin{align}
& (a)_{0}=1,\\
& (a)_{n}=a(a+1)(a+2) {\cdots} (a+n-1),~ n \geq 1.
\end{align} 

Lemma \ref{lemma1} of the present paper is more general than Lemma 1 in \cite{Ornee_TON}, because Lemma \ref{lemma1} in this paper holds for all three cases of the Gauss-Markov processes, i.e., $\theta_n>0$, $\theta_n=0$, and $\theta_n<0$, whereas Lemma 1 in \cite{Ornee_TON} was only shown for $\theta_n>0$. Moreover, \eqref{eq_expectation1}-\eqref{eq_expectation2} in Lemma \ref{lemma1} are neater than (22)-(23) in Lemma 1 of \cite{Ornee_TON}. 
\ifreport
{}
\else
Due to space limitation the proof of Lemma \ref{lemma1} is relegated to our technical report \cite{Ornee2023}.
\fi

The expectations in \eqref{eq_expectation1} and \eqref{eq_expectation2} can be evaluated by Monte-Carlo simulations of scalar random variables $O_{n, Y_{n, i}}$ and $O_{n, Y_{n,i+1}}$ which is much easier than directly simulating the entire process $\{\varepsilon_n (t), t \geq 0\}$. 

The Whittle index of the \emph{dummy bandits} is derived in the following lemma.
\begin{lemma} \label{dummy_lemma}
The Whittle index of the dummy bandits is 0, i.e., $\alpha_0 (\varepsilon, \gamma) = 0$.
\end{lemma}

\ifreport
\begin{proof}
See Appendix \ref{proof_dummy}.
\end{proof}
\else
The proof of Lemma \ref{dummy_lemma} and the Whittle index policy for the RMAB \eqref{problem_dummy}-\eqref{constraint_dummy} is provided in our technical report \cite{Ornee2023}.
\fi
\ifreport
\begin{algorithm}[t] 
\caption{Whittle Index Policy for Signal-aware Remote Estimation Problem \eqref{problem_dummy}-\eqref{constraint_dummy}} \label{alg1_dummy}
\begin{algorithmic}[1]
\For{all time $t$}
      \State Update $X_{n,t}$ using \eqref{eq_solution} for all $n$.
      \State Update $\hat X_{n,t}$ using \eqref{mse} for all $n$.
      \State Update $\varepsilon_n (t)$, $\gamma_n (t)$, and the Whittle index $\alpha_n (\varepsilon_n (t), \gamma_n (t))$ for the $N$ regular bandits using \eqref{est_err}, \eqref{Whittle_index}-\eqref{Whitt_infty}, \eqref{eq_expectation1}-\eqref{eq_expectation2}.
      \State Update the Whittle index $\alpha_0 (\varepsilon_0 (t), \gamma_0 (t)) = 0$ for the $L$ dummy bandits using Lemma \ref{dummy_lemma}.
      \For{all $l = 1, 2, \ldots, L$}
      \If{channel $l$ is idle}
         \State Take a sample from the bandit with the highest Whittle index among the $N$ regular bandits and $L$ dummy bandits, and send it on channel $l$. 
      \EndIf
      \EndFor
\EndFor         
\end{algorithmic}
\end{algorithm}
\else
{}
\fi

\ifreport
The Algorithm for solving \eqref{problem_dummy}-\eqref{constraint_dummy} is provided in Algorithm \ref{alg1_dummy} which activates the $L$ bandits with the highest Whittle index at any given time $t$. As stated in Lemma \ref{dummy_lemma}, each dummy bandit has a Whittle index of $\alpha_0 (\varepsilon_0 (t), \gamma_0 (t)) = 0$. Consequently, if a bandit $n$ (for $n = 1, 2,\ldots, N$) possesses a negative Whittle index, denoted as $\alpha_n (\varepsilon_n (t), \gamma_n(t)) < 0$, it will remain inactive.
Furthermore, if source $n$  is being served by a channel at time $t$ such that $\gamma_n(t) > 0$, then $\alpha_n (\varepsilon_n (t), \gamma_n (t)) = -\infty$ and no more channel will be scheduled to serve source $n$.
\else
{}
\fi
 
\ifreport
Now, we return to the original RMAB \eqref{problem}-\eqref{constraint}.
\else
{}
\fi 
The Whittle index scheduling policy to solve \eqref{problem}-\eqref{constraint} is illustrated  in Algorithm \ref{alg1}. Initially, the set $A$ of unserved bandits is set as $A=\{1,2,\ldots, N\}$. 
If channel $l$ is idle and $\max_{n \in A}\!\! \alpha_n(\varepsilon_n (t), \gamma_n (t)) \geq 0$, then one sample is taken from bandit $n$ having the highest non-negative Whittle index and sent over the channel $l$; meanwhile, bandit $n$ is removed from the set $A$ of unserved bandits. 
Both Algorithms \ref{alg1_dummy} and \ref{alg1} can be either used as an event-driven algorithm, or be executed on discretized time slots $t=0, T_s, 2T_s,\ldots$. When $T_s$ is sufficiently small, the performance degradation caused by time discretization can be omitted.
Because RMAB \eqref{problem}-\eqref{constraint} and the RMAB \eqref{problem_dummy}-\eqref{constraint_dummy} are equivalent to each other, the Whittle index policy 
\ifreport
in Algorithm \ref{alg1_dummy}
\else
for RMAB \eqref{problem_dummy}-\eqref{constraint_dummy}
(provided in \cite{Ornee2023}) 
\fi
and the Whittle index policy in Algorithm \ref{alg1} are equivalent. Specifically, at any time $t$, $L$ bandits having the highest non-negative Whittle index $\alpha_n (\varepsilon, \gamma)$ will be activated. 
\begin{algorithm}[t] 
\caption{Whittle Index Policy for Signal-aware Remote Estimation Problem \eqref{problem}-\eqref{constraint}} \label{alg1}
\begin{algorithmic}[1]
\State Initialize the set of unserved bandits $A=\{1, 2, \ldots, N\}$.
\For{all time $t$}
      \State Update $X_{n,t}$ and $\hat X_{n,t}$ for all $n=1,2, \ldots, N$ using \eqref{eq_solution} and \eqref{mse}, respectively.
      \State Update $\varepsilon_n (t)$, $\gamma_n (t)$, and the Whittle index $\alpha_n (\varepsilon_n (t), \gamma_n (t))$ for all $n = 1, 2, \ldots, N$ using \eqref{est_err}, \eqref{Whittle_index}, \eqref{wiener_index}, \eqref{Whitt_index_unstable}, \eqref{Whitt_infty}, \eqref{eq_expectation1}, and \eqref{eq_expectation2}.
      \State Update $A=\{n \in \{1, 2, \ldots, N\}: \gamma_n (t) =0\}$.
      \For{all $l = 1, 2, \ldots, L$}
      \If{channel $l$ is idle and $\max\limits_{n \in A} \alpha_n(\varepsilon_n (t), \gamma_n (t)) \geq 0$}
         \State $n = \argmax_{n \in A} \alpha_n (\varepsilon_n (t), \gamma_n (t))$.
         \State Take a sample of bandit $n$ and send it on channel $l$.
         \State $A \gets A - \{n\}$.
      \EndIf
      \EndFor
\EndFor         
\end{algorithmic}
\end{algorithm}

\subsubsection{Unity of Whittle Index-based Scheduling and Threshold-based Sampling}
Let consider the special case $N = L = 1$, where the system has a single source and a single channel. Let $w_1 =1$, then problem \eqref{problem}-\eqref{constraint} reduces to
\begin{align} \label{single_problem}
\inf_{\pi \in \Pi} \limsup_{T \to \infty} \mathbb{E}_{\pi} \bigg[\frac{1}{T} \int_{0}^{T} \varepsilon_1 ^2 (t) dt \bigg].
\end{align}
The single-source, single-channel sampling and scheduling problem \eqref{single_problem} is a special case of Proposition \ref{opt_sampler_theorem} with $n=1$ and $\lambda =0$. A threshold-based optimal solution to \eqref{single_problem} is provided by the following corollary of Proposition \ref{opt_sampler_theorem}. 

\begin{corollary} \label{single_theorem}
If the $Y_{1,i}$'s are i.i.d. with $0<\mathbb{E}[Y_{1,i}] < \infty$, then $(S_{1,1} (\beta_1),S_{1,2} (\beta_1),\ldots)$  with a parameter $\beta_1$ is an optimal solution to \eqref{single_problem}, 
where 
\begin{align}\label{single_solution}
S_{1,i+1} (\beta_1)= \inf_t \left\{ t \geq D_{1,i} (\beta_1) :\! |\varepsilon_1 (t)| \!\geq\! v_1 (\beta_1)\right\},
\end{align}
$D_{1, i} (\beta_1)= S_{1, i} (\beta_1)+ Y_{1,i}$, ${v}_1 (\beta_1)$ is defined by 
\begin{align} \label{single_threshold}
v_1(\beta_1) =
\begin{cases}
& \begin{array}{l l}\!\!\!\!\!\!\!\!\! \frac{\sigma_1}{\sqrt{\theta_1}} Q^{-1} \bigg(\frac{\frac{\sigma_1 ^2}{2 \theta_1} \mathbb{E} [e^{-2 \theta_1 Y_{1,i}}]}{\frac{\sigma_1 ^2}{2 \theta_1} - \beta_1}\bigg),
& \text{ if }~ \theta_1 > 0,\\
\!\!\!\!\!\!\!\!\!  \sqrt{3 (\beta_1 - \sigma_1 ^2 \mathbb{E} [Y_{1,i}])}, & \text{ if }~ \theta_1 = 0,\\
\!\!\!\!\!\!\!\!\!\frac{\sigma_1}{\sqrt{-\theta_1}} K^{-1} \bigg(\frac{\frac{\sigma_1^2}{2 \theta_1} \mathbb{E} [e^{-2 \theta_1 Y_{1,i}}]}{\frac{\sigma_1^2}{2 \theta_1} - \beta_1}\bigg), & \text{ if }~ \theta_1 < 0,
\end{array}
\end{cases}
\end{align}
$Q^{-1}(\cdot)$ and $K^{-1}(\cdot)$ are the inverse functions of $Q(x)$ in \eqref{G} and $K(x)$ in \eqref{K}, respectively, for the region $x \in [0,\infty)$, and $\beta_1$ is the unique root of
\begin{align}\label{single_beta}
\!\!\! \mathbb{E}\left[\int_{D_{1,i} (\beta_1)}^{D_{1,i+1}(\beta_1)} \varepsilon_1 ^2 (t) dt\right]\! - \!{\beta_1} {\mathbb{E}[D_{1,i+1}(\beta_1)\!-\!D_{1,i} (\beta_1)]} \!=\! 0.
\end{align}
The optimal objective value to \eqref{single_problem} is given by
\emph{\begin{align}\label{single_obj}
{\bar m}_{1, \text{opt}}  = \frac{\mathbb{E}\left[\int_{D_{1,i} (\beta_1)}^{D_{1,i+1}(\beta_1)}\! \varepsilon_1 ^2 (t) dt\right]}{\mathbb{E}[D_{1,i+1}(\beta_1)\!-\!D_{1,i} (\beta_1)]}. {\noindent}
\end{align}}
{\!\!Furthermore}, $\beta_1$ is exactly the optimal objective value of \eqref{single_problem}, i.e., \emph{$\beta_1 = {\bar m}_{1,\text{opt}}$}.
\end{corollary}
Corollary \ref{single_theorem} follows directly from Proposition \ref{opt_sampler_theorem}. For the cases of the Wiener process ($\theta_1=0$) and stable OU process ($\theta_1>0$), the threshold-based policy in Corollary \ref{opt_sampler_theorem} were earlier reported in \cite{Ornee_TON}. The case of unstable OU process ($\theta_1 < 0$) is new.

It is important to note that the threshold-based policy in Corollary \ref{single_theorem} and the Whittle index policy in the following theorem are equivalent.

\begin{theorem} \label{single_theorem_whittle}
If the $Y_{1,i}$'s are i.i.d. with $0<\mathbb{E}[Y_{1,i}] < \infty$, then $(S_{1,1} ,S_{1,2},\ldots)$ is an optimal solution to \eqref{single_problem}, 
where 
\begin{align}\label{single_solution_whittle}
S_{1,i+1} = \inf_t \left\{ t \geq S_{1,i} :\! \alpha_1 (\varepsilon_1 (t), \gamma_1 (t)) \!\geq\! 0\right\},
\end{align}
and $\alpha_1 (\varepsilon_1 (t), \gamma_1 (t))$ is the Whittle index of source $1$, defined by \eqref{Whittle_index}, \eqref{wiener_index}, \eqref{Whitt_index_unstable}, and \eqref{Whitt_infty} for $n=1$.
\end{theorem} 
{\it Proof sketch.} Because (i) Corollary \ref{single_theorem} provides an optimal solution to \eqref{single_problem} and (ii) \eqref{single_solution_whittle} is equivalent to the solution in Corollary \ref{single_theorem}, \eqref{single_solution_whittle} is also an optimal solution to \eqref{single_problem}. The details are provided in
\ifreport
Appendix \ref{proof_of_single_whittle}.
\else
our technical report \cite{Ornee2023}. 
\fi
$\square$
\begin{figure}
\vspace{-0.3cm}
\centering
\includegraphics[width=6cm]{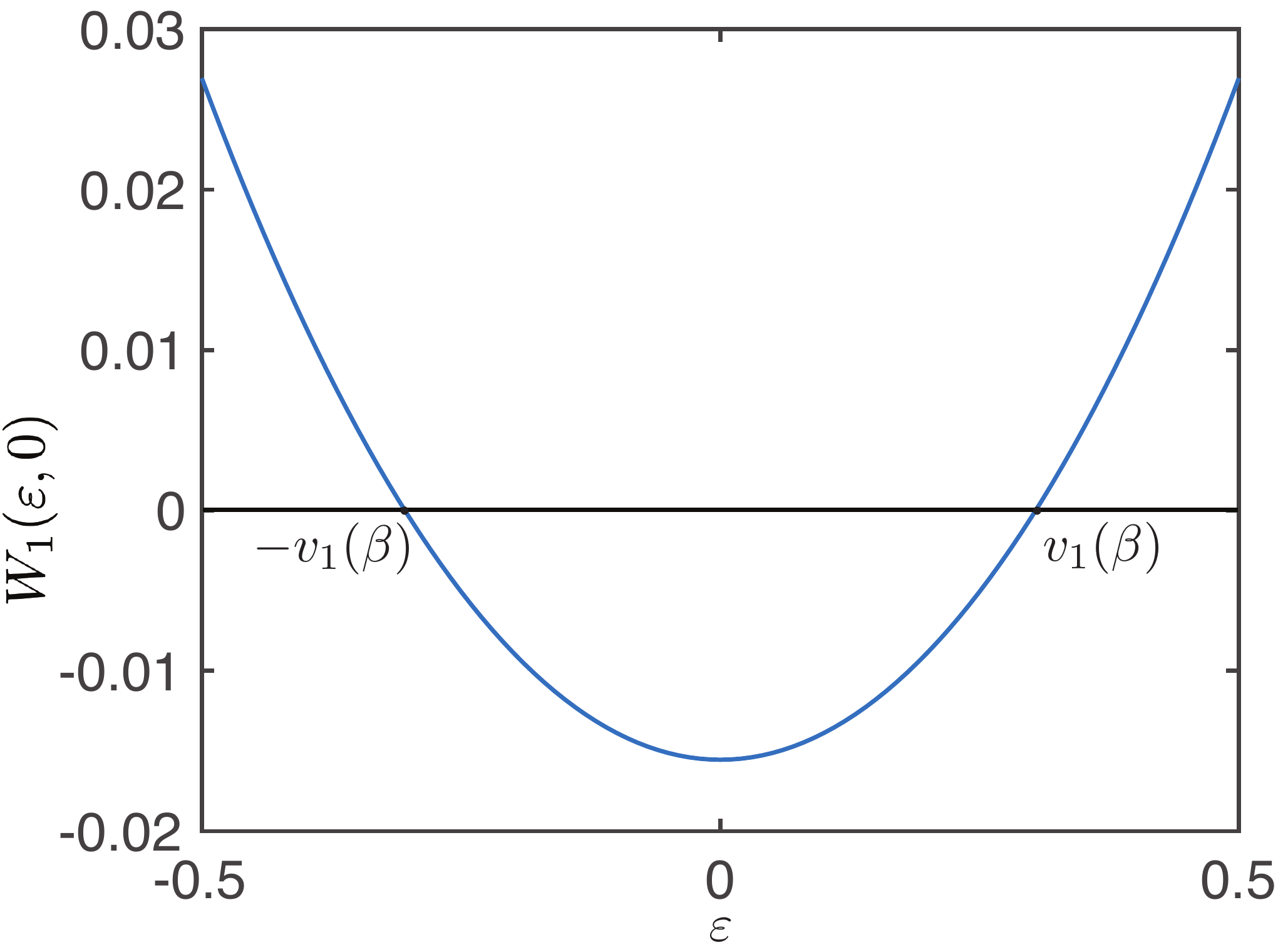}
\caption{\small Illustration of the Whittle index $\alpha_1(\varepsilon, \gamma)$ and the optimal threshold $v_1(\beta_1)$, where the parameters of the Gauss-Markov process are $\sigma_1 =1$ and $\theta_1 =0.1$ and the \emph{i.i.d.} transmission times follow an  exponential distribution with mean $\mathbb{E} [Y_{1,i}] = 2$.}
\label{fig_discussion}
\end{figure}

Corollary \ref{single_theorem} and Theorem \ref{single_theorem_whittle} reveal a unification of threshold-based sampling and scheduling policy developed in \cite{Ornee_TON} and the Whittle index policy developed in the present paper. 
In particular, if the Whittle index $\alpha_1(\varepsilon_1(t), \gamma_1 (t))=0$, then (i) the channel is idle at time $t$ and (ii) the instantaneous estimation error $|\varepsilon_1(t)|$ exactly crosses the optimal threshold $v_1(\beta_1)$ at time $t$. As illustrated in Figure \ref{fig_discussion}, $\varepsilon = \pm v_1 (\beta_1)$ are the roots of equation $\alpha_1 (\varepsilon, 0) = 0$.

The threshold-based sampling and scheduling results outlined in Corollary \ref{single_theorem} and \cite{Ornee_TON} are applicable specifically to the single-source, single-channel scenario. Nevertheless, our exploration in Sections \ref{decomposition}-\ref{Whittle} illustrates the methodology for utilizing these findings to establish indexability and evaluate the Whittle index in the multi-source, multi-channel scenario. 

\subsection{Signal-agnostic Scheduling}
A scheduling policy $\pi \in \Pi$ is called \emph{signal-agnostic} if the policy $\pi$ is independent of the observed process $\{X_{n,t}, t \geq 0\}_{n=1}^N$. Let $\Pi_{\text{agnostic}} \in \Pi$ denote the set of signal-agnostic, causal policies, defined by
\begin{align}
\Pi_{\text{agnostic}} = \{\pi \in \Pi: \pi ~\text{is independent of}~ \{X_{n,t}, t \geq 0\}_{n=1}^N\}. 
\end{align}

In a signal-agnostic policy, the mean-squared estimation error of the process $X_{n,t}$ at time $t$ is \cite{Wiener_TIT}, \cite{Ornee_TON}
\begin{align} \label{aoi_penalty}
\mathbb{E} [\varepsilon_n^2 (t)] =\!\! p_n (\Delta_n (t))\!\!=
\begin{cases}
& \begin{array}{l l}\!\!\!\!\!\!\!\!\! \frac{\sigma_n ^2}{2 \theta_n} (1- e^{-2 \theta_n \Delta_n (t)}),
& \text{ if }~ \theta_n \neq 0,\\
\!\!\!\!\!\!\!\!\! \sigma_n ^2 \Delta_n (t), & \text{ if }~ \theta_n = 0,
\end{array}
\end{cases}
\end{align}
where $\Delta_n(t)$ is the AoI and $p_n(\cdot)$ is an increasing function of the AoI defined in \eqref{aoi_penalty}. By using \eqref{aoi_penalty}, for any policy $\pi \in \Pi_{\text{agnostic}}$ 
\begin{align} \label{aoi_penalty_1}
\mathbb{E} \bigg[\int_{0}^{T} \varepsilon^2 _n (t) dt\bigg] = \mathbb{E} \bigg[\int_{0}^{T}  p_n (\Delta_n (t)) dt\bigg].
\end{align}
Hence, the signal-agnostic sampling and scheduling problem can be formulated as
\begin{align} \label{problem_aoi}
\inf_{\pi \in \Pi_{\text{agnostic}}} &\limsup_{T \to \infty}  \sum_{n=1}^{N} w_n \mathbb{E}_{\pi} \bigg[\frac{1}{T} \int_{0}^{T} p_n (\Delta_n (t)) dt\bigg] \\
\text{s.t.}\;\;\;\;&\sum_{n=1}^{N}  c_n (t) \leq L, c_n (t) \in \{0, 1\}, t \in [0, \infty). \label{aoi_constraint}
\end{align}

Problem \eqref{problem_aoi}-\eqref{aoi_constraint} is a continuous-time Restless 
Multi-armed Bandit (RMAB) with a continuous state space, where $\Delta_n(t )$ of source $n$ is modeled as the state of the restless bandit.
\ifreport
Following the procedure developed in Section \ref{decomposition}, we consider $L$ additional \emph{dummy bandits} where $c_0 (t) \in\{0,1,2,\ldots,L\}$ denotes the number of \emph{dummy bandits} that are activated at time $t$ and reformulate \eqref{problem_aoi}-\eqref{aoi_constraint} as
\begin{align} \label{problem_aoi_dummy}
& \inf_{\pi \in \Pi_{\text{agnostic}}} \limsup_{T \to \infty}  \sum_{n=1}^{N} w_n \mathbb{E}_{\pi} \bigg[\frac{1}{T} \int_{0}^{T} p_n (\Delta_n (t)) dt\bigg] \\
& ~~~~\text{s.t.} ~~~ \sum_{n=0}^{N} c_n (t) = L, c_0 (t) \!\!\in \!\! \{0,1,\ldots, L\}, t \in [0, \infty),\nonumber\\ 
& ~~~~~~~~~~~c_n (t)\!\! \in \!\! \{0, 1\}, n=1,2,\ldots,N, t \!\in\! [0, \infty), \label{constraint_aoi_dummy}
\end{align}
which is an RMAB with an equality constraint.
%
By relaxing constraint \eqref{constraint_aoi_dummy}, the RMAB \eqref{problem_aoi_dummy}-\eqref{constraint_aoi_dummy} is reformulated as
\begin{align} \label{problem_2_aoi}
& \inf_{\pi \in \Pi_{\text{agnostic}}} \limsup_{T \to \infty}  \sum_{n=0}^{N} w_n \mathbb{E}_{\pi} \bigg[\frac{1}{T} \int_{0}^{T} p_n (\Delta_n(t)) dt\bigg] \\
& ~~~~\text{s.t.} ~~~\limsup_{T \to \infty} \sum_{n=0}^{N} \mathbb{E}_{\pi} \bigg[\frac{1}{T} \int_{0}^{T} c_n (t) dt\bigg] = L, \nonumber\\
& ~~~~~~~~~~~c_0 (t) \!\!\in \!\! \{0,1,\ldots, L\}, t \in [0, \infty) \nonumber\\
& ~~~~~~~~~~~c_n (t)\!\! \in \!\! \{0, 1\}, n=1,2,\ldots,N, t \!\in\! [0, \infty).
 \label{relaxed_ constraint_dummy_2_aoi}
\end{align} 
%
Next, we take the Lagrangian dual decomposition of the relaxed problem \eqref{problem_2_aoi}-\eqref{relaxed_ constraint_dummy_2_aoi}, which produces the following problem with a dual variable $\lambda \in \mathbb R$:
\begin{align} \label{relaxed_problem_age}
\inf_{\pi \in \Pi_{\text{agnostic}}} \limsup_{T \to \infty} \mathbb{E}_{\pi} \bigg[\frac{1}{T} \int_{0}^{T} & \sum_{n=1}^{N} w_n p_n (\Delta_n (t)) \nonumber\\
& + \lambda \bigg(\sum_{n=0}^{N} c_n (t) - L\bigg) dt \bigg].
\end{align}
Then, problem \eqref{relaxed_problem_age} can be decomposed into $(N+1)$ separated sub-problems. The sub-problem associated with source $n$ is
\else 
By following the standard relaxation and Lagrangian dual decomposition procedure as explained in Section \ref{decomposition}, we obtain the following sub-problem associated with bandit $n$:
\fi 
\begin{align} \label{per_arm_problem_age}
& \bar m_{n, \text{age-opt}}= \nonumber\\
& \inf_{\pi_{n} \in \Pi_{n,\text{agnostic}}} \!\!\!\limsup_{T \to \infty} \mathbb{E}_{\pi_n} \bigg[\frac{1}{T} \int_{0}^{T} \!\!\!\!w_n p_n (\Delta_n (t)) + \lambda c_n (t) dt\bigg], \nonumber\\
& ~~~~~n = 1,2, \ldots, N.
\end{align}
where $\bar m_{n, \text{age-opt}}$ is the optimum value of \eqref{per_arm_problem_age}, $\pi_{n}= (S_{n,1}, S_{n,2}, \ldots)$ denotes a sub-scheduling policy for source $n$, and $\Pi_{n,\text{agnostic}}$ is the set of all causal sub-scheduling policies of source $n$.

\ifreport
An optimal solution to problem \eqref{per_arm_problem_age} is provided in the following proposition.

\begin{proposition} \label{opt_sampler_theorem_age}
For signal-agnostic scheduling, if the $Y_{n,i}$'s are i.i.d. with $0<\mathbb{E}[Y_{n,i}] < \infty$, then $(S_{n,1} (\beta_{n, \text{age}}),S_{n,2} (\beta_{n, \text{age}}),\ldots)$ with a parameter $\beta_{n, \text{age}}$ is an optimal solution to \eqref{per_arm_problem_age}, 
where 
\begin{align}\label{eq_opt_solution_age}
& S_{n, i+1} (\beta_{n, \text{age}})= \nonumber\\
& \inf \left\{ t \!\geq\! D_{n,i} (\beta_{n, \text{age}})\!\!:\! \mathbb{E} [p_n \!(\!\Delta_n (t+ Y_{n, i+1}))] \geq \beta_{n, \text{age}}\right\},
\end{align}
$D_{n,i} (\beta_{n, \text{age}})= S_{n,i} (\beta_{n, \text{age}})+ Y_{n,i}$, and $\beta_{n, \text{age}}$ is the unique root of
\begin{align}\label{thm1_eq22_age}
& \mathbb{E} \left[\int_{D_{n,i} (\beta_{n, \text{age}})}^{D_{n, i+1}(\beta_{n, \text{age}})} w_n p_n (\Delta_n (t)) dt\right] + \lambda \mathbb{E} [Y_{n,i+1}]  \nonumber\\
& - {\beta_{n, \text{age}}} {\mathbb{E}[D_{n, i+1}(\beta_{n, \text{age}})-D_{n,i} (\beta_{n, \text{age}})]}=0. 
\end{align}
The optimal objective value to \eqref{per_arm_problem} is given by
\emph{\begin{align}\label{thm1_eq23_age}
& {\bar m}_{n, \text{age-opt}}  \nonumber\\
=& \frac{\mathbb{E}\left[\int_{D_{n,i} (\beta_{n, \text{age}})}^{D_{n, i+1}(\beta_{n, \text{age}})}\! w_n p_n (\Delta_n (t)) dt\right] + \lambda \mathbb{E} [Y_{n,i+1}]}{\mathbb{E}[D_{n, i+1}(\beta_{n, \text{age}})\!-\!D_{n,i} (\beta_{n, \text{age}})]}. {\noindent}
\end{align}
{\!\!Furthermore}, $\beta_{n, \text{age}}$ is exactly the optimal objective value of \eqref{per_arm_problem}, i.e., $\beta_{n, \text{age}} = {\bar m}_{n, \text{age-opt}}$.}
\end{proposition}}
\else
{}
\fi

Proposition \ref{opt_sampler_theorem_age} follows from \eqref{aoi_penalty}, \eqref{aoi_penalty_1}, and Corollary 3 of \cite{SunNonlinear2019}.

Define $\Psi_{n, \text{age}} (\lambda)$ as the set of states $(\delta, \gamma) \in [0,\infty) \times [0,\infty)$ such that if $\Delta_n (t) = \delta$ and $\gamma_n (t) = \gamma$, the optimal solution for \eqref{per_arm_problem_age} is to take a passive action at time $t$.

\begin{definition} \textbf{(Indexability).} \label{def_2}
\cite{verloop2016asymptotically}
Bandit $n$ is said to be indexable if, as the activation cost $\lambda$ increases from $-\infty$ to $\infty$, the set \emph{$\Psi_{n, \text{age}} (\lambda)$} increases monotonically, i.e., $\lambda_1 \leq \lambda_2$ implies \emph{$\Psi_{n, \text{age}} (\lambda_1) \subseteq \Psi_{n, \text{age}} (\lambda_2)$}. The RMAB \eqref{problem_aoi}-\eqref{aoi_constraint} is said to be indexable if all $(N + 1)$ bandits are indexable.
\end{definition}

\ifreport
By using Proposition \ref{opt_sampler_theorem_age}
\else
An optimal solution to problem \eqref{per_arm_problem_age} is provided in of our technical report \cite[Proposition 2]{Ornee2023}, where we show that a parameter $\beta_{n, \text{age}}$ is equal to the optimum value of \eqref{per_arm_problem_age}, i.e., $\beta_{n, \text{age}} = {{\bar m}}_{n, \text{age-opt}}$ and $\beta_{n, \text{age}}$ is a function of $\lambda$. By using the solution of \eqref{per_arm_problem_age}, 
\fi
the set $\Psi_{n, \text{age}} (\lambda)$ in Definition \ref{def_2} can be simplified as
\begin{align} \label{def_3}
& \Psi_{n, \text{age}} (\lambda) = \nonumber\\
& \{\!(\!\delta, \gamma\!)\!\! \in \!\![0,\infty)\!\! \times \!\![0,\infty) \!\!:\!\! \gamma \!>\! 0 {\thinspace} \text{or} {\thinspace} \mathbb{E} [p_n (\delta + Y_{n, i+1})] < \beta_{n, \text{age}} (\lambda)\}.
\end{align}

Following the techniques developed in Section \ref{results_sig_aware}, we can obtain

\begin{theorem} \label{aoi_indexability}
If $p_n(\delta)$ is a strictly increasing function of $\delta$, the RMAB problem \eqref{problem_aoi_dummy}-\eqref{constraint_aoi_dummy} is indexable.
\end{theorem}

\ifreport
\begin{proof}
See Appendix \ref{proof_of_aoi_indexability}.
\end{proof}
\else
{}
\fi

{\begin{theorem} \label{theorem3}
If $p_n(\delta)$ is a strictly increasing function of $\delta$ and the $Y_{n,i}$'s are i.i.d. with $0 < \mathbb{E} [Y_{n,i}] < \infty$, then the following assertions are true for the Whittle index of problem \eqref{per_arm_problem_age} at state $(\delta, \gamma)$:

(a) If $\gamma = 0$, then
\emph{\begin{align} \label{age-based_index}
& \alpha_{n, \text{age}} (\delta, 0) = \nonumber\\
& \frac{w_n}{\mathbb{E} [Y_{n,i}]} \bigg\{\mathbb{E}[D_{n,i+1} (\delta) - D_{n,i} (\delta)] \mathbb{E} [p_n (\delta+Y_{n, i+1})] \nonumber\\
& ~~~~~~~~~~~- \mathbb{E} \bigg[\int_{D_{n,i} (\delta)}^{D_{n, i+1} (\delta)} p_n (s) ds\bigg]\bigg\},
\end{align}}
where 
$D_{n,i}(\delta) = S_{n,i}(\delta) + Y_{n,i}$ and 
\begin{align}
S_{n,i+1}(\delta) = D_{n,i}(\delta) + \max\{\delta- Y_{n,i}, 0\}.
\end{align}

(b) If $\gamma > 0$, then
\emph{\begin{align} \label{Whitt_infty_age}
\alpha_{n, \text{age}} (\delta, \gamma) = -\infty.
\end{align}}
\end{theorem}

\ifreport
\begin{proof}
See Appendix \ref{proof of age-based index}.
\end{proof}
\else
{}
\fi

The expectations in \eqref{age-based_index} can be easily evaluated using the following lemma: 
\begin{lemma} \label{aoi_lemma}
In Theorem \ref{theorem3}, it holds that 
\begin{align} \label{age_exp1}
\mathbb{E}[D_{n,i+1} (\delta) - D_{n,i} (\delta)] = \mathbb{E} [\max\{\delta, Y_{n,i}\}],
\end{align}
\begin{align} \label{age_exp2}
\mathbb{E} \bigg[\int_{D_{n,i} (\delta)}^{D_{n, i+1} (\delta)} p_n (s) ds\bigg] = &\mathbb{E} [R_{n,3} (\max\{\delta, Y_{n,i}\} + Y_{n, i+1})] \nonumber\\
& - \mathbb{E} [R_{n,3} (Y_{n,i})],
\end{align}
where 
\begin{align}
R_{n,3} (\delta) = \int_0^{\delta} p_n (s) ds.
\end{align}
\end{lemma} 

\ifreport
\begin{proof}
See Appendix \ref{proof_aoi_lemma}.
\end{proof}
\else
{}
\fi

Theorems \ref{aoi_indexability}-\ref{theorem3} and Lemma \ref{aoi_lemma} hold for all increasing functions $p_n(\delta)$ of the AoI $\delta$, not necessarily the mean-square estimation error function in \eqref{aoi_penalty}.
\ifreport
{}
\else
Due to space limitation, the proofs of Theorems \ref{aoi_indexability}-\ref{theorem3} and Lemma \ref{aoi_lemma} are relegated to our technical report \cite{Ornee2023}.
\fi 


\ifreport
\begin{algorithm}[t] 
\caption{Whittle Index Policy for Signal-aware Remote Estimation to solve \eqref{problem_aoi_dummy}-\eqref{constraint_aoi_dummy}} \label{alg1_dummy_age}
\begin{algorithmic}[1]
\For{all time $t$}
      \State Update $\Delta_n (t)$, $\gamma_n (t)$, and the Whittle index $\alpha_{n, \text{age}} (\Delta_n (t), \gamma_n (t))$ for all $n=1,2, \ldots, N$ using \eqref{age_eq}, \eqref{age-based_index}, \eqref{Whitt_infty_age}, \eqref{age_exp1}, and \eqref{age_exp2}.
      \State Update $\Delta_0 (t)$, $\gamma_0 (t)$, and the Whittle index $\alpha_{0, \text{age}} (\Delta_0 (t), \gamma_n (t))$ for $L$ dummy bandits.
      \For{all $l = 1, 2, \ldots, L$}
      \If{channel $l$ is idle}
         \State Choose bandit $n$ with the highest non-negative Whittle index $\alpha_{n, \text{age}} (\Delta_n (t), \gamma_n (t))$ for $n=0, 1, \ldots, N$. 
      \EndIf
      \EndFor
\EndFor         
\end{algorithmic}
\end{algorithm}
\else
{}
\fi

\ifreport
The Algorithms for solving RMAB \eqref{problem_aoi_dummy}-\eqref{constraint_aoi_dummy} is provided in Algorithm \ref{alg1_dummy_age} and for solving the original RMAB \eqref{problem_aoi}-\eqref{aoi_constraint} is provided in  Algorithm \ref{alg2}. 
\else
{}
\fi
 
\ifreport
\else 
The Whittle index scheduling policy for solving the sampling and scheduling problem \eqref{problem_aoi}-\eqref{aoi_constraint} is illustrated  in Algorithm \ref{alg2}. 
\fi 
Theorems \ref{aoi_indexability}-\ref{theorem3}, Lemma \ref{aoi_lemma}, and
\ifreport
Algorithms \ref{alg1_dummy_age}-\ref{alg2}
\else
Algorithm \ref{alg2}
\fi
generalize prior studies on AoI-based Whittle index policies, e.g.,  \cite{tripathi2019whittle, hsu2018age, kadota2019scheduling}. More specifically, the Whittle index policies detailed in \cite{tripathi2019whittle, hsu2018age, kadota2019scheduling} were derived for the scenario of constant transmission times where the zero-wait sampling policy \cite{yates2015lazy, sun2017update} is an optimal solution for the sub-problem \eqref{per_arm_problem_age}, and the resulting Whittle index always maintains a non-negative value. In contrast, our current study accommodates scenarios involving \emph{i.i.d.} random transmission times. In such cases, the optimality of zero-wait sampling is not assured for sub-problem \eqref{per_arm_problem_age}, resulting in the potential for both positive and negative values for the Whittle index derived in Theorem \ref{theorem3}.

\subsubsection{Unity of Whittle Index-based Scheduling and Threshold-based Sampling}
For single-source, single-channel special case with $w_1 =1$, problem \eqref{problem_aoi}-\eqref{aoi_constraint} reduces to
\begin{align} \label{problem_aoi_single}
\inf_{\pi \in \Pi_{\text{agnostic}}} &\limsup_{T \to \infty}   \mathbb{E}_{\pi} \bigg[\frac{1}{T} \int_{0}^{T} p_1 (\Delta_1 (t)) dt\bigg].
\end{align}

The single-source, single-channel sampling and scheduling problem \eqref{problem_aoi_single} is a special case of Proposition \ref{opt_sampler_theorem_age} with $n=1$ and $\lambda =0$. A threshold-based optimal solution to \eqref{problem_aoi_single} is provided by the following Corollary of Proposition \ref{opt_sampler_theorem_age}. 

\begin{corollary} \label{single_theorem_age}
If the $Y_{1,i}$'s are i.i.d. with $0<\mathbb{E}[Y_{1,i}] < \infty$ and $p_1 (\cdot)$ in \eqref{aoi_penalty} is strictly increasing, then \emph{$(S_{1,1} (\beta_{1, \text{age}}),S_{1,2} (\beta_{1, \text{age}}),\ldots)$}  with a parameter \emph{$\beta_{1, \text{age}}$} is an optimal solution to \eqref{problem_aoi_single}, 
where 
\emph{\begin{align}\label{single_solution_age}
S_{1,i+1} \!(\beta_{1, \text{age}}\!)\!\!=\! \inf_t\!\! \left\{ t\!\! \geq \!\!D_{1,i} (\beta_{1, \text{age}}\!) \!:\! \mathbb{E} [p_1 \!(\!\Delta_1 \!(t \!+\! Y_{1, i+1})\!)\!] \!\!\geq\!\! \beta_{1, \text{age}}\!\right\}\!\!,
\end{align}}
\emph{$D_{1, i} (\beta_{1, \text{age}})= S_{1, i} (\beta_{1, \text{age}})+ Y_{1,i}$}, and \emph{$\beta_{1, \text{age}}$} is the unique root of
\emph{\begin{align}\label{single_beta_age}
& \mathbb{E}\left[\int_{D_{1,i} (\beta_{1, \text{age}})}^{D_{1,i+1}(\beta_{1, \text{age}})} p_1 (\Delta_1 (t)) dt\right]\! \nonumber\\
& - \!{\beta_{1, \text{age}}} {\mathbb{E}[D_{1,i+1}(\beta_{1, \text{age}})-D_{1,i} (\beta_{1, \text{age}})]} = 0.
\end{align}}
The optimal objective value to \eqref{problem_aoi_single} is given by
\emph{\begin{align}\label{single_obj_age}
{\bar m}_{1, \text{age-opt}}  = \frac{\mathbb{E}\left[\int_{D_{1,i} (\beta_{1, \text{age}})}^{D_{1,i+1}(\beta_{1, \text{age}})}\! p_1 (\Delta_1 (t)) dt\right]}{\mathbb{E}[D_{1,i+1}(\beta_{1, \text{age}})\!-\!D_{1,i} (\beta_{1, \text{age}})]}. {\noindent}
\end{align}}
{\!\!Furthermore}, \emph{$\beta_{1, \text{age}}$} is exactly the optimal objective value of \eqref{problem_aoi_single}, i.e., \emph{$\beta_{1, \text{age}} = {\bar m}_{1,\text{age-opt}}$}.
\end{corollary}
Corollary \ref{single_theorem_age} follows directly from Proposition \ref{opt_sampler_theorem_age}. This result was reported earlier in \cite[Theorem 1]{SunNonlinear2019}.
The threshold-based policy in Corollary \ref{single_theorem_age} and the Whittle index policy in the following theorem are equivalent.

\begin{algorithm}[t] 
\caption{Whittle Index Policy for Signal-agnostic Remote Estimation} \label{alg2}
\begin{algorithmic}[1]
\State Initialize the set $A$ of unserved bandits $A=\{1, 2, \ldots, N\}$.
\For{all time $t$}
      \State Update $\Delta_n (t)$, $\gamma_n (t)$, and the Whittle index $\alpha_{n, \text{age}} (\Delta_n (t), \gamma_n (t))$ for all $n = 1, 2, \ldots, N$ using \eqref{age_eq}, \eqref{age-based_index}, \eqref{Whitt_infty_age}, \eqref{age_exp1}, and \eqref{age_exp2}.
      \State Update $A=\{n \in \{1,2,\ldots, N \}: \gamma_n (t) =0\}$.
      \For{all $l = 1, 2, \ldots, L$}
      \If{channel $l$ is idle and $\max\limits_{n \in A} \alpha_{n, \text{age}} (\Delta_n (t), \gamma_n (t))\!\! \geq \!\!0$}
         \State $n = \argmax_{n \in A} \alpha_{n, \text{age}} (\Delta_n (t), \gamma_n (t))$.
         \State Take a sample of bandit $n$ and send it on channel $l$.
         \State $A \gets A - \{n\}$.
      \EndIf
      \EndFor
\EndFor         
\end{algorithmic}
\end{algorithm}

\begin{theorem} \label{single_theorem_whittle_age}
If $p_1 (\delta)$ is a strictly increasing function of $\delta$ and the $Y_{1,i}$'s are i.i.d. with $0 < \mathbb{E} [Y_{1,i}] < \infty$, then $(S_{1, 1}, S_{1, 2}, \ldots)$ is an optimal solution to \eqref{problem_aoi_single}, where
\emph{\begin{align}
S_{1, i+1} = \inf_t \{t \geq S_{1, i} : \alpha_{1, \text{age}} (\Delta_1 (t), \gamma_1 (t)) \geq 0\}, 
\end{align}}
where \emph{$\alpha_{1, \text{age}} (\Delta_1 (t), \gamma_1 (t))$} is the Whittle index of source $1$, defined by \eqref{age-based_index} and \eqref{Whitt_infty_age} for $n=1$.
\end{theorem}

Theorem \ref{single_theorem_whittle_age} follows from Theorem \ref{theorem3} for $n=1$ and $\lambda=0$ and Corollary \ref{single_theorem_age}. 

In the AoI literature, threshold-based scheduling and Whittle index have been two distinct approaches for AoI minimization. Our study unifies the two approaches: for AoI minimization of a single source, the threshold policy in Corollary \ref{single_theorem_age} and the Whittle index policy based in Theorem \ref{single_theorem_whittle_age} are equivalent. Specifically, if the Whittle index $\alpha_{1, \text{age}} (\Delta_1 (t), \gamma_1 (t)) = 0$, then (i) the channel is idle at time $t$ and (ii) the expected age-penalty function surpasses the threshold in Corollary \ref{single_theorem_age} at time $t$.

\section{Numerical Results}

In this section, we compare the following three scheduling policies for multi-source remote estimation:

\begin{itemize}

\item Maximum Age First, Zero-Wait (MAF-ZW) policy: 
Whenever one channel $l$ becomes free, the MAF-ZW policy will take a sample from the source with the highest AoI among the sources that are currently not served by any channel, and send the sample over channel $l$. 

\item Signal-agnostic, Whittle Index policy: The policy that we proposed in Algorithm \ref{alg2}.

\item Signal-aware, Whittle Index policy: The policy that we proposed in Algorithm \ref{alg1}.

\end{itemize}

Figure \ref{figure1} depicts the total time-average mean-squared estimation error versus the parameter $\sigma_1$ of the Gauss-Markov source 1, where the number of sources is $N=4$ and the number of channels is $L=2$. The other parameters of the Gauss-Markov processes are $\sigma_2 = 0.8, \sigma_3 = 0.9, \sigma_4=1$, and $\theta_1=-0.1, \theta_2 = \theta_3= \theta_4=0.1$. The transmission times are \emph{i.i.d.} and follow a normalized log-normal distribution, where $Y_{n,i} = e^{\rho Q_{n,i}}/\mathbb{E} [e^{\rho Q_{n,i}}]$, $\rho >0$ is the scale parameter of the log-normal distribution, and $(Q_{n,1}, Q_{n,2}, \ldots)$ are \emph{i.i.d.} Gaussian random variables with zero mean and unit variance. In our simulation, $\rho = 1.5$. All sources are given the same weight $w_1 = w_2 = w_3 = w_4 = 1$. In Figure \ref{figure1}, the signal-aware Whittle index policy has a smaller total MSE than the signal-agnostic Whittle index policy and the MAF-ZW policy. The total MSE of the signal-aware Whittle index policy achieves up to 1.58 times performance gain over the signal-agnostic Whittle index policy, and up to 1.65 times over the MAF-ZW policy. 

\begin{figure}
\vspace{-0.0cm}
\centering
\includegraphics[width=6cm]{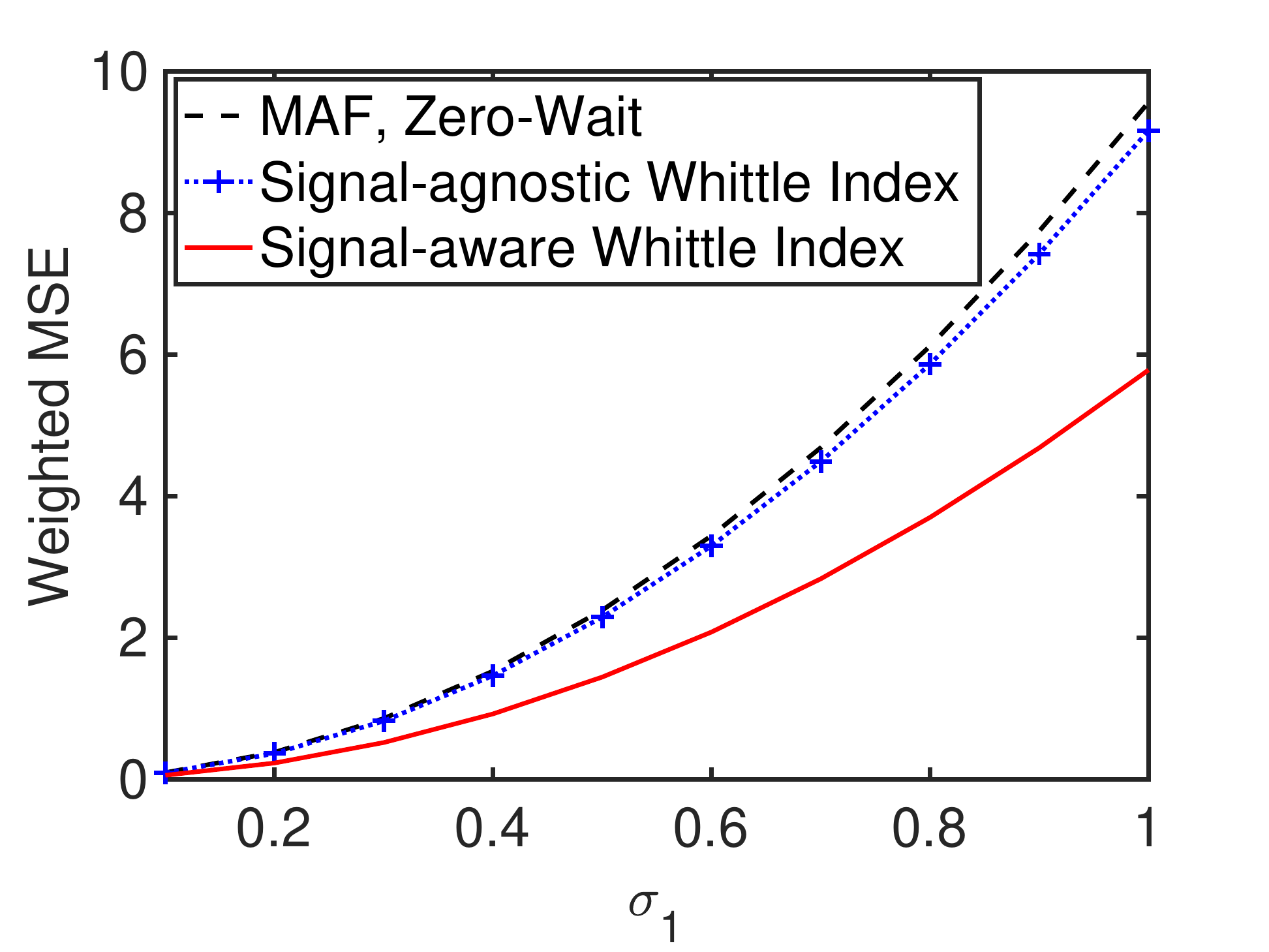}  
\caption{\small Total time-average MSE vs the parameter $\sigma_1$ of the Gauss-Markov source 1, where the number of sources is $N=4$ and the number of channels is $L=2$. The transmission times are \emph{i.i.d}, following a normalized log-normal distribution with parameter $\rho=1.5$, and $\mathbb{E} [Y_{n,i}] =1$. The other parameters of the Gauss-Markov sources are $\sigma_2 = 0.8, \sigma_3 = 0.9, \sigma_4=1$, and $\theta_1=-0.1, \theta_2 = \theta_3= \theta_4=0.1$.}\vspace{-0.0cm}
\label{figure1}
\end{figure}

Figure \ref{figure2} illustrates the total time-average mean-squared estimation error versus the parameter $\theta_1$ of the Gauss-Markov source 1, where the number of sources is $N=4$, and the number of channels is $L=2$. The other parameters of the Gauss-Markov processes are $\theta_2 = 0.2, \theta_3 = 0.3, \theta_4= 0.1$, and $\sigma_1= \sigma_2= \sigma_3 = \sigma_4=1$. The transmission time distribution and the weights of the sources are the same as in Figure \ref{figure1}. In Figure \ref{figure2}, the total MSE of the signal-aware Whittle index policy achieves up to 8.6 times performance gain over the MAF-ZW policy and up to 1.32 times over the signal-agnostic Whittle index policy. 
When $\theta_1<0$, the performance gain of the signal-aware Whittle index policy is much higher than that in the case of $\theta_1>0$.  This suggests a high performance gain can be achieved if the Gauss-Markov sources are highly unstable. For all three policies, the total MSE decreases,
as $\theta_1$ increases. 

\begin{figure}
\vspace{-0.0cm}
\centering
\includegraphics[width=6cm]{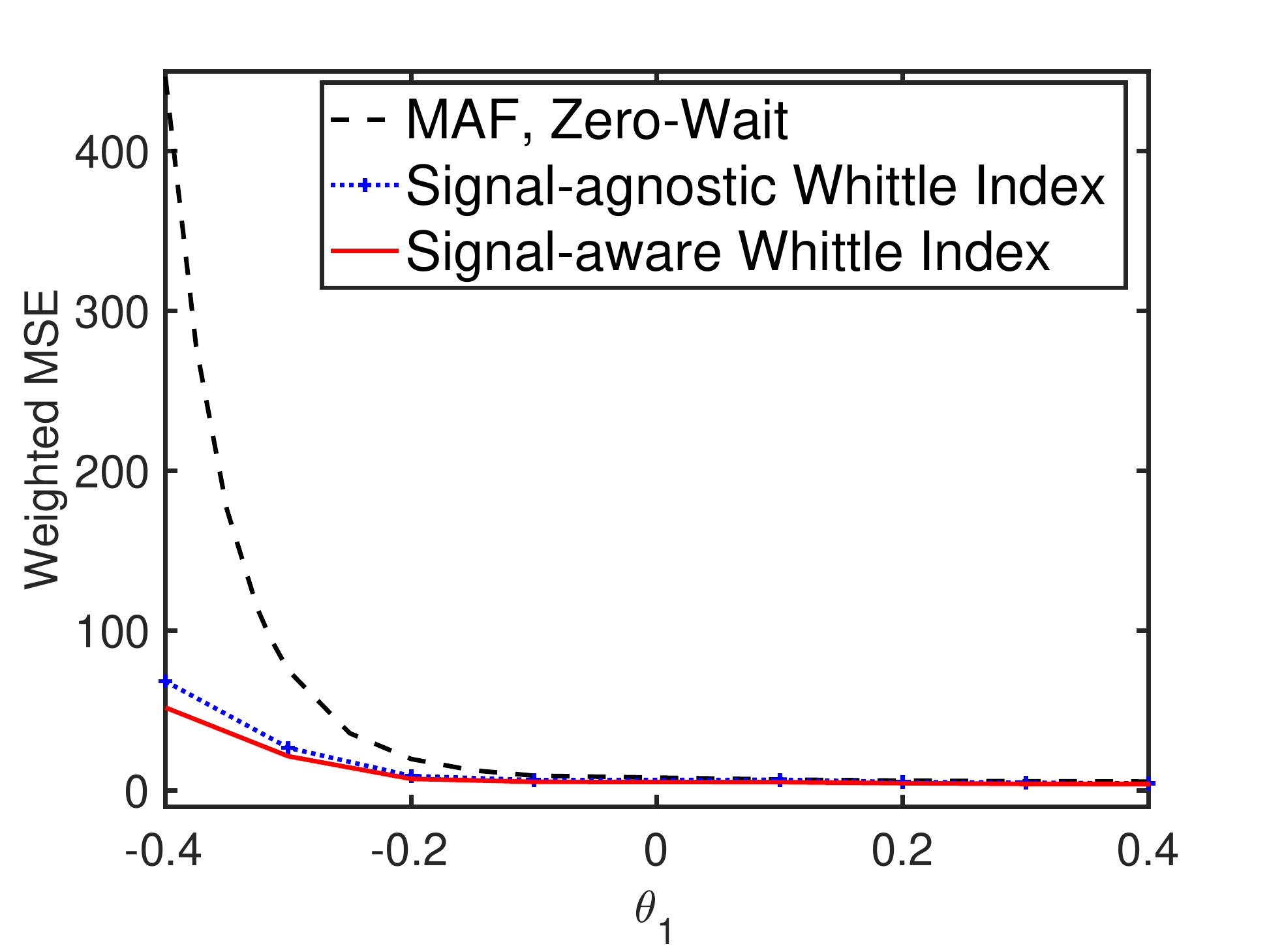}  
\caption{ \small Total time-average MSE vs the parameter $\theta_1$ of the Gauss-Markov source 1, where the number of sources is $N=4$ and the number of channels is $L=2$. The transmission times are \emph{i.i.d.}, following a normalized log-normal distribution with parameter $\rho=1.5$, and $\mathbb{E} [Y_{n,i}] =1$. The other parameters for the Gauss-Markov sources are $\sigma_1 = \sigma_2 = \sigma_3= \sigma_4=1$, and $\theta_2=0.2, \theta_3= 0.3, \theta_4= 0.1$.}\vspace{-0.0cm}
\label{figure2}
\end{figure}
\section{conclusion}


In this paper, we have studied a sampling and scheduling problem in which samples of multiple Gauss-Markov sources are sent to remote estimators that need to monitor the sources in real-time. The formulated sampling and scheduling problem is a restless multi-armed bandit problem, where each bandit process has a continuous state space and requires continuous-time control. We have proved that the problem is indexable and proposed a Whittle index policy. Analytical expressions of the Whittle index have been obtained. For single-source, single-channel scheduling, we have showed that it is optimal to take a sample at the earliest time 
when the Whittle index is no less than zero. This result provides a new interpretation of earlier studies on threshold-based sampling policies for the Wiener and Ornstein-Uhlenbeck processes. 

%

\bibliographystyle{IEEEtran}
\bibliography{ref,ref1,ref_2,sueh}

\newpage

\appendices
\section{Proof of (\ref{eq_solution}) for $\theta_n <0$} \label{proof_unstable_OU_process}
A solution to \eqref{SDE} for initial state $X_{n,0} = 0$ and parameter $\mu_n =0$ can be written in terms of a stochastic integral as follows
\begin{align}
X_{n,t} = \sigma_n e^{-\theta_n (t-S_{n,i})} \int_{0}^{t-S_{n,i}} e^{\theta_n s} dW_s,
\end{align}
which holds for any value of $\theta_n$.
To derive an alternative formula for $\theta_n <0$, let consider the following well-known lemma:
\begin{lemma} \label{lemma_Y_t}
Let $Y_{n,t}$ be a Gaussian process with $Y_{n,0}=0$, $\mathbb{E} [Y_{n,t}] = 0$, and it has independent increments. Then the distribution of $Y_{n,t}$ can be completely determined by its variance function $\mathbb{E} [Y^2 _{n,t}]$.
\end{lemma}

Define
\begin{align} \label{Y_t}
Y_{n,t} = e^{\theta_n (t- S_{n,i})} X_{n,t} =  \sigma_n  \int_{0}^{t-S_{n,i}} e^{\theta_n s} dW_s.
\end{align}
Lemma \ref{lemma_Y_t} implies that $Y_{n,t}$ in \eqref{Y_t} is a Gaussian process and its variance function is given by
\begin{align}
\mathbb{E} [{Y^2 _{n,t}}] = \frac{\sigma_n ^2}{2 \theta_n} (e^{2 \theta_n (t- S_{n,i})} - 1). 
\end{align}
Consider the following process
\begin{align} \label{Z_t}
Z_{n,t} = \frac{\sigma_n}{\sqrt{2 \theta_n}} W_{n, e^{2 \theta_n (t-S_{n,i})} - 1}. 
\end{align}
Because Brownian motion $W_{n,t}$ is a Gaussian process with variance function $t$, $Z_{n,t}$ in \eqref{Z_t} is a Gaussian process with variance function $\frac{\sigma_n ^2}{2 \theta_n} (e^{2 \theta_n (t-S_{n,i})} - 1)$ which is the same as the variance function of $Y_{n,t}$. From Lemma \ref{lemma_Y_t}, both the processes $Y_{n,t}$ and $Z_{n,t}$ are equal in distribution. 

When $\theta_n < 0$, consider $\rho_n = -\theta_n > 0$. In this setting, the variance function of $Y_{n,t}$ can be written as
\begin{align} \label{stoc_proof_1}
\mathbb{E} [{Y^2 _{n,t}}] = \frac{\sigma_n ^2}{2 \rho_n} (1-e^{-2 \rho_n (t-S_{n,i})}). 
\end{align}
Let
\begin{align} \label{stoc_proof_2}
Z'_{n,t} = \frac{\sigma_n}{\sqrt{2 \rho_n}} W_{n, 1 - e^{-2 \rho_n (t- S_{n,i})}},
\end{align}
which implies $Z'_{n,t}$ has the same variance function as $Y_{n,t}$. Hence, $Z'_{n,t}$ and $Y_{n,t}$ are equal in distribution. By using \eqref{stoc_proof_1} and \eqref{stoc_proof_2} in \eqref{Y_t}, we get that
\begin{align}
X_{n,t} =& \frac{\sigma_n ^2}{2 \rho_n} e^{\rho_n (t- S_{n,i})} W_{n, 1 - e^{-2 \rho_n (t- S_{n,i})}},
\end{align}
from which \eqref{eq_solution} for $\theta_n < 0$ follows. This completes the proof.

\section{Proof of the simplification of policy $\pi_n$} \label{simplification_policy}

The sampling and scheduling policy $\pi_n = ((S_{n,1}, G_{n,1}), (S_{n,2}, G_{n,2}), \ldots)$ consists of the sampling time $S_{n, i}$ and the transmission starting time $G_{n, i}$ for each sample $i$. In policy $\pi_n$, sample $i$ can be generated when the server is busy sending another sample, and hence sample $i$  
needs to wait for some time before being submitted to the server, i.e., $S_{n,i} <G_{n,i}$. 
Consider a sampling and scheduling policy $\pi_n ' = \{(S_{n,1}, G_{n,1}), \ldots, (S_{n, i-1}, G_{n, i-1}), G_{n,i}, (S_{n, i+1}, G_{n, i+1}),\\ \ldots\}$ such that the generation time and transmission starting time of sample $i$ are equal to each other, i.e., $S_{n,i} = G_{n, i}$. 
We will show that the MSE of the sampling policy $\pi_n '$ is smaller than that of the sampling policy $\pi_n$.

Note that $\{X_{n,t}: t\in[0, \infty)\}$ does not change according to the sampling policy, and the sample delivery times $\{D_{n,1},D_{n,2},\ldots\}$ remain the same in policy $\pi_n$ and policy $\pi_n '$. 
Hence, the only difference between policies $\pi_n$ and $\pi_n '$ is that \emph{the generation time of sample $i$}. 
The MMSE estimator under policy $\pi_n$ is given by \eqref{mse} and the  MMSE estimator under policy $\pi_n '$ is given by
\begin{align}\label{eq_esti_pi1}
& \hat X_{n,t} ' \nonumber\\
=& \mathbb{E}[X_{n,t}|(\!S_{n,j}, \!X_{n, S_{n,j}}\!, \!G_{n,j}\!, \!D_{n,j})_{j\leq i-1}, (\!G_{n,i}, \!X_{n, G_{n,i}}\!, \!D_{n,i})]\nonumber\\
 = &\left\{\begin{array}{l l} 
\mathbb{E}[X_{n,t}|G_{n,i}, X_{n, G_{n,i}}],& t\in[D_{n,i},D_{n, i+1}); \\
\mathbb{E}[X_{{n,t}}|S_{n,j}, X_{n, S_{n,j}}],& t\in[D_{n,j},D_{n, j+1}),~j\neq i. \\
\end{array}\right.
\end{align}

Next, we consider another sampling and scheduling policy $\pi_n ''$ in which the samples $(G_{n,i}, X_{n, G_{n,i}})$ and  $(S_{n,i}, X_{n, S_{n,i}})$ are both delivered to the estimator at the same time $D_{n,i}$. Clearly, the estimator under policy $\pi_n ''$ has more information than those under policies $\pi_n$ and $\pi_n '$. One can also show that the MMSE estimator under policy $\pi_n ''$ is
\begin{align}\label{eq_esti_pi2}
& \hat X_{n,t} '' \nonumber\\
=& \mathbb{E}[X_{n,t}|(\!S_{n,j}, \!X_{n, S_{n,j}}, \!G_{n,j}, \!D_{n,j})_{j\leq i}, (G_{n,i}, X_{n, G_{n,i}}, D_{n,i})] \nonumber\\
=&\left\{\begin{array}{l l} 
\mathbb{E}[X_{n,t}|G_{n,i}, X_{n, G_{n,i}}],& t\in[D_{n,i},D_{n,i+1}); \\
\mathbb{E}[X_{n,t}|S_{n,j}, X_{n, S_{n,j}}],& t\in[D_{n,j},D_{n, j+1}),~j\neq i. \\
\end{array}\right.
\end{align}
Notice that, because of the strong Markov property of OU process, the estimator under policy $\pi'' _n$ uses the fresher sample $(\!G_{n,i}, \!X_{n, G_{n,i}}\!)$, instead of the stale sample $(S_{n,i}, X_{n, S_{n,i}})$, to construct $\hat X_{n,t} ''$ during $[D_{n,i},D_{n, i+1})$. Because the estimator under policy $\pi_n ''$ has more information than that of under policy $\pi_n$, one can imagine that policy $\pi_n ''$ has a smaller estimation error than policy $\pi_n$, i.e., 
\begin{align}\label{eq_small_error}
& \mathbb{E}\left\{\int_{D_{n,i}} ^{D_{n, i+1}} (X_{n,t}-\hat X_{n,t})^2dt\right\}\geq \nonumber\\ 
& \mathbb{E}\left\{\int_{D_{n,i}} ^{D_{n, i+1}} (X_{n,t}-\hat X_{n,t} '')^2  dt\right\}.
\end{align}
To prove \eqref{eq_small_error}, we invoke the orthogonality principle of the MMSE estimator \cite[Prop. V.C.2]{poor1998introduction} under policy $\pi_n ''$ and obtain 
\begin{align}\label{eq_orthogonality}
\mathbb{E}\left\{\int_{D_{n,i}} ^{D_{n, i+1}}2 (X_{n,t}-\hat X_{n,t} '') (\hat X_{n,t} ''- \hat X_{n,t}) dt\right\}=0,
\end{align}
where we have used the fact that  $(G_{n,i}, X_{n, G_{n,i}})$ and  $(S_{n,i}, X_{n, S_{n,i}})$ are available by the MMSE estimator under policy $\pi_n ''$.
Next, from \eqref{eq_orthogonality}, we can get
\begin{align}
&\mathbb{E}\left\{\int_{D_{n,i}} ^{D_{n, i+1}} (X_{n,t}-\hat X_{n,t})^2 dt\right\} \nonumber\\
= & \mathbb{E}\left\{\int_{D_{n,i}} ^{D_{n, i+1}} (X_{n,t}-\hat X_{n,t} '')^2 +(\hat X_{n,t} ''- \hat X_{n,t})^2 dt\right\} \nonumber\\
& + \mathbb{E}\left\{\int_{D_{n,i}} ^{D_{n, i+1}}2 (X_{n,t}-\hat X_{n,t} '') (\hat X_{n,t} ''- \hat X_{n,t}) dt\right\} \nonumber 
\end{align}
\begin{align}
= & \mathbb{E}\left\{\int_{D_{n,i}} ^{D_{n, i+1}} (X_{n,t}-\hat X_{n,t} '')^2 +(\hat X_{n,t} ''- \hat X_{n,t})^2 dt\right\} \nonumber\\
\geq &\mathbb{E}\left\{\int_{D_{n,i}} ^{D_{n, i+1}} (X_{n,t}-\hat X_{n,t} '')^2  dt\right\}. 
\end{align}
In other words, the estimation error of policy $\pi_n ''$ is no greater than that of policy $\pi_n$. Furthermore, by comparing \eqref{eq_esti_pi1} and \eqref{eq_esti_pi2}, we can see  that the MMSE estimators under policies $\pi_n ''$ and $\pi_n '$ are exactly the same. Therefore, the estimation error of policy $\pi_n '$ is no greater than that of policy $\pi_n$.

By repeating the above arguments for all samples $i$ satisfying $S_{n,i} <G_{n,i}$, one can show that the sampling policy $\{G_{n,1}, G_{n,2}, \ldots\}$ is better than the sampling policy $\pi \!\!=\!\!\{(S_{n,1}, G_{n,1}), (S_{n,2}, G_{n,2}), \ldots\}$.
This completes the proof.


\section{Proof of Proposition \ref{opt_sampler_theorem}} \label{threshold_proof}

In this section, we present the proof of Proposition \ref{opt_sampler_theorem} for unstable OU process, i.e., for $\theta_n < 0$. The proofs for stable OU process (i.e., $\theta_n > 0$) and Wiener process (i.e., $\theta_n = 0$) follow the similar steps.

Define the $\sigma$-field
\begin{align}
\mathcal{N}_{n,t} = \sigma(X_{n,s} : 0 \leq s \leq t),
\end{align}
which is the set of events whose occurrence are determined by the realization of the process $\{X_{n,s} : 0 \leq s \leq t\}$. The right continuous filtration $\{\mathcal{N}_{n,t} ^{+}, t \geq 0\}$ is defined by
\begin{align}
\mathcal{N}_{n,t} ^{+} = \cup_{s > t} \mathcal{N}_{n,s}.
\end{align}
In causal sampling policies, each sampling time is a stopping time with respect to the filtration $\{\mathcal{N}_{n,t} ^{+}, t \geq 0\}$, i.e., \cite{Durrettbook10}
\begin{align} \label{stopping_time}
\{S_{n,i} \leq t\} \in \mathcal{N}_{n,t} ^{+}, \forall t \geq 0.
\end{align}

Let the sampling and scheduling policy $\pi_n = (S_{n,1}, S_{n,2}, \ldots)$ in \eqref{per_arm_problem} satisfy two conditions: (i) Each sampling policy $\pi_n \in \Pi_n$ satisfies \eqref{stopping_time} for all $i$. (ii) The sequence of inter-sampling times $\{T_{n,i} = S_{n, i+1}-S_{n,i}, i=0,1,\ldots\}$ forms a \emph{regenerative process} \cite[Section 6.1]{Haas2002}: There exists an increasing sequence  $0\leq {k_1}<k_2< \ldots$ of almost surely finite random integers such that the post-${k_j}$ process $\{T_{n, k_j+i}, i=0,1,\ldots\}$ has the same distribution as the post-${k_0}$ process $\{T_{n, k_0+i}, i=0,1,\ldots\}$ and is independent of the pre-$k_j$ process $\{T_{n, i}, i=0,1,\ldots, k_j-1\}$; further, we assume that $\mathbb{E}[{k_{j+1}}-{k_j}]<\infty$, $\mathbb{E}[S_{n, k_{1}}]<\infty$, and $0<\mathbb{E}[S_{n, k_{j+1}}-S_{n, k_j}]<\infty, ~j=1,2,\ldots$.

We will prove Proposition \ref{opt_sampler_theorem} in three steps: First , we show that it is better not to sample when no channel is free. Second, we decompose the MDP in \eqref{per_arm_problem} into a series of mutually independent per-sample MDPs. Finally, we solve the per-sample MDP analytically.


In Appendix \ref{simplification_policy}, by using the strong Markov property of the Gauss-Markov process and the orthogonality principle of MMSE estimation, we have shown that it is better not to take a sample before the previous sample is delivered. Hence, the sampling time and the transmission starting time are equal to each other. 
By this, let us consider a sub-class of sampling and scheduling policies $\Pi_{n,1} \subset \Pi_n$ such that each  sample is generated and sent out after all previous samples are delivered, i.e., 
\begin{align}
\Pi_{n,1} = \{\pi_n \in \Pi_n : S_{n,i} = G_{n,i} \geq D_{n,i-1} \text{ for all $i$}\}. \nonumber
\end{align}
For any policy $\pi_n \in \Pi_{n,1}$, the \emph{information} used for determining $S_{n,i}$ includes: (i) the history of signal values $(X_{n,t}: t\in[0, S_{n,i}])$ and (ii) the service times  $(Y_{n,1},\ldots, Y_{n,i-1})$ of previous samples. Let us define the $\sigma$-fields $\mathcal{F}_{n,t} = \sigma(X_{n,s}: s\in[0, t])$ and $\mathcal{F}_{n,t}^+ = \cap_{r>t}\mathcal{F}_{n,r}$. 
Then, $\{\mathcal{F}_{n,t}^+,t\geq0\}$ is the {filtration} (i.e., a non-decreasing and right-continuous family of $\sigma$-fields) of the Gauss-Markov process $X_{n,t}$.
Given the service times   $(Y_{n,1},\ldots, Y_{n,i-1})$ of previous samples, $S_{n,i} $ is a \emph{stopping time} with respect to the filtration $\{\mathcal{F}_{n,t}^+,t\geq0\}$ of the Gauss-Markov process $X_{n,t}$, that is
\begin{align}
[\{S_{n,i}\leq t\} | Y_{n,1},\ldots, Y_{n,i-1}] \in \mathcal{F}_{n,t} ^+.\label{eq_stopping}
\end{align}  
Hence, the policy space $\Pi_{n,1}$ can be expressed as
\begin{align}\label{eq_policyspace}
\!\!\!\!\Pi_{n,1} = & \{S_{n,i} : [\{S_{n,i}\leq t\} | Y_{n,1},\ldots, Y_{n,i-1}] \in \mathcal{F}_{n,t}^+, \nonumber\\
&~~~~~~~\text{$T_{n,i}$ is a regenerative process} \}.\!\!
\end{align}  
Let $Z_{n,i} = S_{n,i+1} - D_{n,i}\geq0$ represent the \emph{waiting time} between the delivery time $D_{n,i}$ of the $i$-th sample and the generation time $S_{n,i+1}$ of the $(i+1)$-th sample. Then, 
\begin{align} 
& S_{n,i} =  \sum_{j=0}^{i-1} (Y_{n,j} + Z_{n,j}), \label{sample_time}\\
& D_{n,i} = \sum_{j=0}^{i-1} (Y_{n, j} + Z_{n,j}) + Y_{n,i} \label{deliver_time}
\end{align}
for each $i=1,2,\ldots$. 
Given $(Y_{n,0},Y_{n,1},\ldots)$, $(S_{n,1},S_{n,2},\ldots)$ is uniquely determined by $(Z_{n,0},Z_{n,1},\ldots)$. Hence, one can also use $\pi = (Z_{n,0}, Z_{n,1},\ldots)$ to represent a sampling and scheduling policy.

By using \eqref{est_err}, \eqref{sample_time}, \eqref{deliver_time}, and the assumption that the inter-sampling times follow a regenerative process, the MDP in \eqref{per_arm_problem} can be transformed as the following.
\begin{align} \label{relaxed_MDP}
& \bar m_{n, \text{opt}} = \nonumber\\
& \inf_{\pi_n \!\in\! \Pi_{n,1}} \!\lim_{t \to \infty} \!\!\!\frac{\sum_{i=1}^{t} \!\!\mathbb{E}\! \bigg[\!\! \int_{Y_{n,i}}^{Y_{n,i} + Z_{n,i} + Y_{n, i+1}} \!\!w_n \varepsilon_n ^2 \!(s)\! ds\!\bigg] \!\!+\!\! \lambda \mathbb{E} [Y_{n,i+1}]}{\sum_{i=1}^{t} \mathbb{E} [Y_{n,i} + Z_{n,i}]}.
\end{align} 
In order to solve \eqref{relaxed_MDP}, let consider the following MDP with parameter $k \geq 0$:
\begin{align} \label{relaxed_MDP_1}
h(k) = & \inf_{\pi_n \in \Pi_{n,1}} \lim_{t \to \infty} \sum_{i=1}^{t} \mathbb{E} \bigg[ \int_{Y_{n,i}}^{Y_{n,i} + Z_{n,i} + Y_{n, i+1}} \!\!\!\!\!\!\!\!\!\!w_n \varepsilon_n ^2 (s) ds  \nonumber\\
& + \lambda Y_{n,i+1} - k (Y_{n,i} + Z_{n,i})\bigg],
\end{align}
where $h(k)$ is the optimum value of \eqref{relaxed_MDP_1}. Similar to the Dinkelbach's method \cite{dinkelbach1967nonlinear} for non-linear fractional programming, the following lemma holds for the MDP in \eqref{relaxed_MDP}:
\begin{lemma} \label{lem_ratio_to_minus}
\cite{Ornee_TON}, \cite{Wiener_TIT} 
The following assertions are true: 
\begin{itemize}
\vspace{0.5em}
\item[(a).] \emph{$\bar m_{n,\text{opt}} \gtreqqless k$} if and only if $h(k)\gtreqqless 0$. 
\vspace{0.5em}
\item[(b).] If $h(k)=0$, the solutions to \eqref{relaxed_MDP}
and \eqref{relaxed_MDP_1} are identical. 
\end{itemize}
\end{lemma}

Hence, the solution to \eqref{relaxed_MDP} can be obtained by solving \eqref{relaxed_MDP_1} and finding $k = \bar m_{n, \text{opt}}$ for which $h(\bar m_{n, \text{opt}}) = 0$.


Define 
\begin{align} \label{beta}
\beta_n = \bar m_{n, \text{opt}}.
\end{align}
In this sequel, we need to introduce the following lemma.
\begin{lemma} \label{preliminary_lemma}
For any $\beta_n \geq 0$, it holds that
\begin{align} \label{decomposed_MDP}
& \mathbb{E} \bigg[\!\!\int_{Y_{n,i}}^{Y_{n,i} + Z_{n,i} + Y_{n, i+1}} \!\!\!\!\!\!\!\!\!\!\!w_n \varepsilon_n ^2 \!(s) ds\! +\! \lambda Y_{n,i+1} - \beta_n (Y_{n,i} + Z_{n,i}) \bigg] \nonumber\\
=& \mathbb{E} \bigg[\!\int_{Y_{n,i}}^{Y_{n,i} + Z_{n,i} + Y_{n, i+1}} \!\!\!\!\!\!\!\!\!\!(w_n \varepsilon_n ^2 (s) - \beta_n) ds + w_n \gamma_n O^2 _{n, Y_{n,i} + Z_{n,i}}\!\bigg] \nonumber\\
& + w_n \frac{\sigma_n ^2}{2 \theta_n}(\mathbb{E} [Y_{n,i}] - \gamma_n) + \lambda \mathbb{E} [Y_{n,i+1}] - \beta_n \mathbb{E} [Y_{n,i}],
\end{align}
where $O_{n, Y_{n,i} + Z_{n,i}}$ can be obtained from \eqref{OU_shifted}, and $\gamma_n$ is a constant defined as
\begin{align} \label{gamma}
{\gamma}_n = \frac{1}{2 \theta_n} \mathbb{E} [1 - e^{-2 \theta_n Y_{n,i}}].
\end{align}
\end{lemma}

\begin{proof}
We can write \eqref{decomposed_MDP} as
\begin{align} \label{eq_2}
& \mathbb{E} \bigg[\!\!\int_{Y_{n,i}}^{Y_{n,i} + Z_{n,i} + Y_{n, i+1}} \!\!\!\!\!\!\!\!\!\!w_n \varepsilon_n ^2 (s) ds + \lambda Y_{n,i+1} - \beta_n (Y_{n,i} + Z_{n,i}) \bigg] \nonumber\\
=& \mathbb{E} \bigg[\!\int_{Y_{n,i}}^{Y_{n,i} + Z_{n,i}} \!\!\!\!\!\!\!\!\!\!w_n \varepsilon_n ^2 (s) ds\!\bigg] \!\!+\!\! \mathbb{E} \bigg[\int_{Y_{n,i} + Z_{n,i}}^{Y_{n,i} + Z_{n,i} + Y_{n, i+1}} w_n \varepsilon_n ^2 (s) ds\bigg] \nonumber\\
& + \lambda \mathbb{E} [Y_{n,i+1}] - \beta_n \mathbb{E} (Y_{n,i} + Z_{n,i}).
\end{align}
In order to prove Lemma \ref{preliminary_lemma}, we need to compute the second term in \eqref{eq_2}. The Gauss-Markov process $O_{n,t}$ in \eqref{OU_shifted} is the solution to the following SDE
\begin{align}
dO_{n,t} = - \theta_n O_{n,t} dt + \sigma_n dW_{n,t}. 
\end{align}
In addition, the infinitesimal generator of $O_{n,t}$ is \cite[Eq. A1.22]{Borodin1996}
\begin{align}\label{eq_generator}
  \mathcal{G}= -\theta_n u \frac{\partial} {\partial u} + \frac{\sigma_n ^2}{2}\frac{\partial^2} {\partial u^2}.
\end{align}

Now, let us introduce the following lemma which is more general than Lemma 5 in \cite{Ornee_TON} and works for any OU process irrespective of the signal structure, i.e., the value of parameter $\theta_n$. By using Dynkin's formula and optional stopping theorem, we get the following useful lemma.

\begin{lemma}\label{lem_stop}
\cite{Ornee_TON} Let $\tau \geq0$ be a stopping time of the OU process $O_{n,t}$ with $ \EE\left[ O_\tau ^2 \right]<\infty$, 
then
\begin{align}
\mathbb{E}\left[\int_0^\tau O_{n,t} ^2 dt\right] & = \mathbb{E}\left[\frac{\sigma_n ^2}{2\theta_n}\tau - \frac{1}{2\theta_n}O_{n, \tau} ^2\right].\label{eq_stop}
\end{align}
If, in addition, $\tau$ is the first exit time of a bounded set, then
\begin{align}\label{eq_stop11}
&\mathbb{E}\left[\tau\right] = \mathbb{E} [R_{n,1} (O_{n, \tau})] ,\\
&\mathbb{E}\left[\int_0^\tau O_{n,t} ^2 dt\right]  = \mathbb{E} [R_{n,2} (O_{n, \tau})],\label{eq_stop12}
\end{align}
where $R_{n,1} (\cdot)$ and $R_{n,2} (\cdot)$ are defined in \eqref{eq_R_1} and \eqref{eq_R_2}, respectively. 
\end{lemma}

\begin{proof}
We first prove \eqref{eq_stop}. It is known that the OU process $O_{n,t}$ is a Feller process \cite[Section 5.5]{liggett2010continuous}. By using a property of Feller process in   \cite[Theorem 3.32]{liggett2010continuous}, we get that
\begin{align}
&O_{n,t} ^2 -  \int_0^t \mathcal{G} (O_{n,s} ^2 ) ds\nonumber\\
=&O_{n,t} ^2 - \int_0^t (- \theta_n O_{n,s} 2 O_{n,s} + \sigma_n ^2) ds\nonumber\\
=& O_{n,t} ^2 - \sigma_n ^2 t +  2 \theta_n \int_0^t  O_{n,s} ^2 ds
\end{align}
is a martingale.
According to \cite{Durrettbook10}, the minimum of two stopping times is a stopping time and constant times are stopping times. Hence, $t{\wedge}{\tau}$ is a bounded stopping time for every $t\in[{0,\infty})$, where $x{\wedge}y=\min\{x, y\}$. Then, by \cite[Theorem 8.5.1]{Durrettbook10}, for all $t\in[{0,\infty})$
\begin{align}\label{eq_app_lem4_1}
&\mathbb{E}\left[\int_{0}^{t{\wedge}{\tau}}{O_{n,s} ^2}ds\right] =\mathbb{E}\bigg[\frac{{\sigma_n}^2}{2\theta_n}(t{\wedge}{\tau})\bigg]-\mathbb{E}\bigg[\frac{1}{2\theta_n}O_{n, t{\wedge}{\tau}}^2\bigg].
\end{align}
Because $\mathbb{E}\left[\int_{0}^{t{\wedge}{\tau}}O_{n,s} ^{2}ds\right]$ and $\mathbb{E}[t{\wedge}{\tau}]$ are positive and increasing with respect to $t$, by using the monotone convergence theorem \cite[Theorem 1.5.5]{Durrettbook10}, we get
\begin{align}\label{eq_app_lem4_4}
\lim_{t\to\infty} \mathbb{E}\left[\int_{0}^{t{\wedge}{\tau}}O_{n,s} ^2ds\right]&=\mathbb{E}\left[\int_{0}^{\tau}O_{n,s} ^2ds\right],\\
\lim_{t\to\infty}\mathbb{E}[(t{\wedge}{\tau})]&=\mathbb{E}[{\tau}] . \label{eq_app_lem4_3}
\end{align}
In addition, according to Doob's maximal inequality \cite{Durrettbook10}, we get that
\begin{align} \label{eq_doob}
\EE \left[\sup_{0\leq s\leq \tau}O_{n,s} ^2\right] \leq 4 \mathbb{E} [O_{n, \tau} ^2] < \infty.
\end{align}
Because $0 \leq O_{n, t\wedge \tau}^2 \leq \sup_{0\leq s\leq \tau} O_{n, s}^2$ for all $t$ and \eqref{eq_doob} implies that $\sup_{0\leq s\leq \tau} O_{n, s}^2$ is integratable, by invoking the dominated convergence theorem \cite[Theorem 1.5.6]{Durrettbook10}, we have
\begin{align}\label{eq_goal}
&\lim_{t\to\infty}\mathbb{E}\left[O_{n, t{\wedge}{\tau}}^2\right]=\mathbb{E}\left[{O_{n, \tau}^2}\right].
\end{align}
Combining \eqref{eq_app_lem4_4}-\eqref{eq_goal}, \eqref{eq_stop} is proven. 

Next, we prove \eqref{eq_stop11} and \eqref{eq_stop12}. By using the solution of the ODE in \eqref{eq_free1}, one can show that $R_{n,1} (\varepsilon)$ in \eqref{eq_R_1} is the solution to the following ODE
\begin{align}
\frac{\sigma_n ^2}{2} R''_{n,1} (\varepsilon) - {\theta_n} \varepsilon R'_{n,1} (\varepsilon) =1,
\end{align}
and $R_{n,2} (\varepsilon)$ in \eqref{eq_R_2} is the solution to the following ODE
\begin{align}
\frac{\sigma_n ^2}{2} R''_{n,2} (\varepsilon) - {\theta_n} \varepsilon R'_{n,2} (\varepsilon) =\varepsilon^2.
\end{align}
In addition, $R_{n,1} (\varepsilon)$ and $R_{n,2} (\varepsilon)$ are twice continuously differentiable.
According to Dynkin's formula in \cite[Theorem 7.4.1 and the remark afterwards]{Bernt2000}, because the initial value of $O_{n,t}$ is $O_{n,0}=0$, if $\tau$ is the first exit time of a bounded set, then
\begin{align}
\!\!\!\EE_0 [R_{n,1} (O_{n,\tau})] &\!\!=\!\! R_{n,1} (0) \!+\! \EE_0 \left[\int_0^\tau 1 ds\right] = R_{n,1} (0) + \EE_0 [\tau],\!\!\!\\
\!\!\!\EE_0 [R_{n,2} (O_{n,\tau})] &\!\!=\!\! R_{n,2} (0) + \EE_0 \left[\int_0^\tau O_{n,s}^2 ds\right].\!\!\!
\end{align}
Because $R_{n,1} (0)=R_{n,2} (0)  = 0$, \eqref{eq_stop11} and \eqref{eq_stop12} follow. 
This completes the proof.
\end{proof}

By using Lemma \ref{lem_stop}, we can write
\begin{align} \label{eq_12}
& \mathbb{E} \bigg[\int_{Y_{n,i} + Z_{n,i}}^{Y_{n,i} + Z_{n,i} + Y_{n, i+1}} w_n \varepsilon_n ^2 (s) ds\bigg] \nonumber\\
& =\!\! w_n \frac{\sigma_n^2}{2 \theta_n} \mathbb{E} [Y_{n, i+1}]\! -\!  w_n \frac{1}{2 
\theta_n} \mathbb{E} \big[O^2_{n, Y_{n,i} + Z_{n,i} + Y_{n, i+1}} \! \nonumber\\
& ~~~~~~~~~~~~~~~~~~~~~~~~~~~~~~~~~~~~- \! O^2_{n, Y_{n,i} + Z_{n,i}}\big],
\end{align}
where
\begin{align}
&\EE \left[O_{n, Y_{n,i}+Z_{n,i}+Y_{n,i+1}}^2 - O_{n, Y_{n,i}+Z_{n,i}}^2\right]\nonumber\\
=&\EE \!\!\left[\!\! \left(\!\!O_{n, Y_{n,i}\!+\!Z_{n,i}} e^{-\theta_n\! Y_{n,i+1} } \!\!+\!\! \frac{\sigma_n}{\sqrt{\!-2\theta_n}} \!e^{-{\theta_n}Y_{n,i+1}} \!W_{\!\!1\!-e^{2{\theta_n}Y_{n,i+1}}}\!\!\right)^2\right.\nonumber\\
&~~~-O_{n, Y_{n,i}+Z_{n,i}}^2\bigg],\!\!\label{eq_lemma5_1}
\end{align}
\begin{align}
=& \EE\!\! \left[\!O_{\!n, Y_{n,i}\!+Z_{n,i}}^2 \!(\!e^{\!-2{\theta_n}\!Y_{n,i+1}} \!\!-\!\! 1\!) \!\!-\!\! \frac{\sigma_n ^2}{{2\theta_n}} \!e^{\!-2{\theta_n}\!Y_{n,i+1}} \!W_{\!\!n, \!1\!- e^{2{\theta_n}\!Y_{n,i+1}}}^2 \!\right] \nonumber\\
& \!+\!\!\EE \!\!\left[\!2O_{\!n, Y_{n,i}+Z_{n,i}}\!e^{\!-{\theta_n}\!Y_{n,i+1}} \!\frac{\sigma_n}{\sqrt{\!-\!2\theta_n}} \!e^{\!-{\theta_n}\!Y_{n,i+1}} \!W_{\!\!n, 1\!-e^{2{\theta_n}\!Y_{n,i+1}}} \!\!\right].
\end{align}
Because $Y_{n,i+1}$ is independent of $O_{n, Y_{n,i}+Z_{n,i}}$ and $W_{n,t}$, we have 
\begin{align}
 &\EE \left[O_{n, Y_{n,i}+Z_{n,i}}^2 (e^{-2{\theta_n}Y_{n,i+1}} - 1)\right] \nonumber\\
 =& \EE \left[O_{n, Y_{n,i}+Z_{n,i}}^2\right] \EE \left[e^{-2{\theta_n}Y_{n,i+1}} - 1\right], 
\end{align}
 and
\begin{align}
& \EE \!\!\left[\!2O_{\!n, Y_{n,i}\!+Z_{n,i}}\!e^{\!-{\theta_n}\!Y_{\!n,i+1}} \!\frac{\sigma_n}{\sqrt{\!-\!2\theta_n}} e^{\!-{\theta_n}\!Y_{\!n,i+1}} \!W_{\!\!n, 1\!-\!e^{2{\theta_n}\!Y_{n,i+1}}} \!\right] \nonumber\\
=& \EE \!\!\left[\!2O_{n, \!Y_{n,i}\!+\!Z_{n,i}}\!\right]\!\!\EE\!\!\left[\!e^{\!-{\theta_n}\!Y_{n,i+1}} \!\frac{\sigma_n}{\sqrt{\!-\!2\theta_n}} e^{\!-{\theta_n}\!Y_{n,i+1}} \!W_{\!\!n, 1\!-\!e^{2{\theta_n}Y_{n,i+1}}} \!\right]\nonumber\\
\overset{(a)}{=}& \EE \left[2O_{n, Y_{n,i}+Z_{n,i}}\right] \nonumber\\
& \!\EE\!\!\left[\!\EE \!\left[\!e^{-{\theta_n}Y_{n,i+1}} \!\frac{\sigma_n}{\sqrt{-2\theta_n}} e^{-{\theta}Y_{n,i+1}} W_{n, 1-e^{2{\theta_n}Y_{n,i+1}}} \bigg| Y_{n,i+1}\right]\right]\!\!,
\end{align}
where Step (a) is  due to the law of iterated expectations.
Because $\EE[W_{n,t}]=0 $ for all constant $t\geq 0$, it holds for all realizations of $Y_{n,i+1}$ that
\begin{align}
\EE \!\left[\!e^{-{\theta_n}Y_{n,i+1}} \frac{\sigma_n}{\sqrt{-2\theta_n}} e^{-{\theta_n}Y_{n,i+1}} W_{n, 1-e^{2{\theta_n}Y_{n,i+1}}} \bigg| Y_{n,i+1}\!\right] \!=\!0.
\end{align}
Hence, 
\begin{align}
\EE \!\left[\!2O_{n, Y_{n,i}+Z_{n,i}}e^{-{\theta_n}\!Y_{n,i+1}} \!\frac{\sigma_n}{\sqrt{\!-\!2\theta_n}}\! e^{\!-{\theta_n}\!Y_{n,i+1}} \!W_{\!\!n, 1-e^{2{\theta_n}\!Y_{n,i+1}}} \!\!\right]  \!\!\!=\!\! 0.
\end{align}
In addition,
\begin{align}
 &\EE \left[\frac{\sigma_n^2}{{2\theta_n}} e^{-2{\theta_n}Y_{n,i+1}} W_{n, 1-e^{2{\theta_n}Y_{n,i+1}}}^2 \right] \nonumber\\
\overset{(a)}{=} &\frac{\sigma_n ^2}{{2\theta_n}} \EE \left[\EE \left[e^{-2{\theta_n}Y_{n,i+1}} W_{n, 1-e^{2{\theta_n}Y_{n,i+1}}}^2 \bigg| Y_{n,i+1}\right]\right] \nonumber\\
\overset{(b)}{=}&  \frac{\sigma_n ^2}{{2\theta_n}}\EE \left[e^{-2{\theta_n}Y_{n,i+1}}-1 \right],
\end{align}
where Step (a) is due to the law of iterated expectations and Step (b) is due to $\EE [W_{n,t} ^2] = t$ for all constant $t\geq 0$. Hence,
\begin{align} \label{eq_11}
& \mathbb{E} \big[O^2_{n, Y_{n,i} + Z_{n,i} + Y_{n, i+1}} - O^2_{n, Y_{n,i} + Z_{n,i}}\big] =\nonumber\\
& \mathbb{E}\! \big[O^2_{n, Y_{n,i} + Z_{n,i}}\big] \mathbb{E} [e^{-2 \theta_n Y_{n, i+1}} -1] + \frac{\sigma_n^2}{2 \theta_n} \mathbb{E} [1-e^{-2 \theta_n Y_{n, i+1}}].
\end{align}
By using \eqref{eq_11} in \eqref{eq_12}, we get that
\begin{align} \label{eq_13}
& \mathbb{E} \bigg[\int_{Y_{n,i} + Z_{n,i}}^{Y_{n,i} + Z_{n,i} + Y_{n, i+1}} w_n \varepsilon_n ^2 (s) ds\bigg] \nonumber\\
= & w_n\! \frac{\sigma_n^2}{2 \theta_n} \mathbb{E} [Y_{n, i+1}] \!-\! w_n \!\frac{1}{2 \theta_n}\!  \mathbb{E}\! \big[O^2_{Y_{n,i} + Z_{n,i}}\big] \mathbb{E} [e^{-2 \theta_n \!Y_{n, i+1}} -1\!] \nonumber\\
& - w_n\! \frac{\sigma_n ^2}{4 \theta_n ^2} \mathbb{E} [1-e^{-2 \theta_n Y_{n, i+1}}], \nonumber\\
= & w_n \frac{\sigma_n^2}{2 \theta_n} \big\{ \mathbb{E} [Y_{n, i+1}] - \gamma_n\big\} + w_n \gamma_n \mathbb{E} \left[O^2_{Y_{n,i} + Z_{n,i}}\right],
\end{align}
where $\gamma_n$ is defined in \eqref{gamma}.
Substituting \eqref{eq_13} into \eqref{eq_2} yields
\begin{align} \label{first}
& \mathbb{E} \bigg[\int_{Y_{n,i}}^{Y_{n,i} + Z_{n,i} + Y_{n, i+1}} \!\!\!\!\!w_n \varepsilon_n ^2 (s) ds + \lambda Y_{n,i+1} - \beta_n [Y_{n,i} + Z_{n,i}] \bigg] \nonumber\\
=& \mathbb{E} \bigg[\int_{Y_{n,i}}^{Y_{n,i} + Z_{n,i}} w_n \varepsilon_n ^2 (s) ds\bigg] + w_n \frac{\sigma_n^2}{2 \theta_n} \big\{ \mathbb{E} [Y_{n, i+1}] - \gamma_n\big\} \nonumber\\
& + w_n \gamma_n \mathbb{E} \left[O^2_{Y_{n,i} + Z_{n,i}}\right] + \lambda \mathbb{E} [Y_{n,i+1}] - \beta_n \mathbb{E} [Y_{n,i} + Z_{n,i}],
\end{align}
\begin{align}
=& \mathbb{E} \bigg[\int_{Y_{n,i}}^{Y_{n,i} + Z_{n,i}} (w_n \varepsilon_n ^2 (s) - \beta_n) ds + w_n \gamma_n O^2_{Y_{n,i} + Z_{n,i}}\bigg] \nonumber\\
& +  w_n \frac{\sigma_n^2}{2 \theta_n} \big\{ \mathbb{E} [Y_{n, i+1}] - \gamma_n\big\} - \beta_n \mathbb{E} [Y_{n,i}] + \lambda \mathbb{E} [Y_{n,i+1}],
\end{align}
from which \eqref{decomposed_MDP} follows.
\end{proof}

For any $s\geq 0$, define the $\sigma$-fields $\mathcal{F}^{s}_{n,t} = \sigma(O_{n, s+r}-O_{n, s}: r\in[0, t])$ and  the right-continuous filtration $\mathcal{F}_{n,t} ^{s+} = \cap_{r>t}\mathcal{F}_{n,r} ^{s}$. Then, $\{\mathcal{F}^{s+}_{n,t} ,t\geq0\}$ is the {filtration}  of the {time-shifted OU process} $\{O_{n, s+t}-O_{n, s},t\in[0,\infty)\}$. Define ${\mathfrak{M}}_{n,s}$ as the set of integrable stopping times of $\{O_{n, s+t}-O_{n, s},t\in[0,\infty)\}$, i.e.,
\begin{align}
\mathfrak{M}_{n, s} = \{\tau \geq 0:  \{\tau\leq t\} \in \mathcal{F}^{s+}_{n,t}, \mathbb{E}\left[\tau\right]<\infty\}.
\end{align}
By using  a sufficient statistic of \eqref{relaxed_MDP_1},  
we can  obtain
\begin{lemma}\label{thm_solution_form}
An optimal solution  $(Z_{n,0}, Z_{n,1},\ldots)$ 
to 
\eqref{relaxed_MDP_1} 
satisfies 
\begin{align}\label{eq_opt_stopping}
\inf_{Z_{n,i}\in \mathfrak{M}_{Y_{n,i}}} \mathbb{E}\!\bigg[& \int_{Y_{n,i}}^{Y_{n,i}+Z_{n,i}} (w_n \varepsilon_n ^2 (s) -\beta_n)ds \nonumber\\
& + \gamma_n  O_{n, Y_{n,i}+Z_{n,i}}^2\bigg|  O_{n, Y_{n,i}},Y_{n,i} \bigg],
\end{align}
where  $\beta_n \geq 0$ and $\gamma_n \geq0$ are defined in \eqref{beta} and \eqref{gamma}, respectively.
\end{lemma}

\begin{proof}
Because the $Y_{n,i}$'s are \emph{i.i.d.}, \eqref{first} is determined by the control decision $Z_{n,i}$ and the information $(Y_{n,i}, O_{n, Y_{n,i}})$. Hence, $(Y_{n,i}, O_{n, Y_{n,i}})$ is a \emph{sufficient statistic} for determining $Z_{n,i}$ in \eqref{relaxed_MDP_1}. Therefore, there exists an optimal policy $(Z_{n,0},Z_{n,1},\ldots)$ to \eqref{relaxed_MDP_1}, in which $Z_{n,i}$ is determined based on only  $(Y_{n,i}, O_{n, Y_{n,i}})$. 
By this, \eqref{relaxed_MDP_1} is decomposed into a sequence of per-sample MDPs, given by \eqref{eq_opt_stopping}. 
This completes the proof. 
\end{proof}


Next, we solve \eqref{eq_opt_stopping} by using free-boundary method for optimal stopping problems. Let 
consider an OU process $V_{n,t}$ with initial state $V_{n,0}=v_n$ and $\mu_n =0$. Define the $\sigma$-fields $\mathcal{F}^{V}_{n,t} = \sigma(V_{n, s}: s\in[0, t])$, $\mathcal{F}^{V+}_{n,t} = \cap_{r>t}\mathcal{F}_{n,r}^{V}$, and the {filtration} $\{\mathcal{F}^{V+}_{n,t},t\geq0\}$ associated to $\{V_{n,t},t\geq0\}$. Define ${\mathfrak{M}}_V$ as the set of integrable stopping times of $\{V_{n,t},t\in[0,\infty)\}$, i.e.,
\begin{align}
\mathfrak{M}_{V} = \{\tau \geq 0:  \{\tau\leq t\} \in \mathcal{F}^{V+}_{n,t}, \mathbb{E}\left[\tau\right]<\infty\}.
\end{align}

Our  goal is to solve the following optimal stopping problem for any given initial state $v_n \in \mathbb{R}$ and for any $\beta_n > 0$
\begin{align}\label{eq_stop_problem}
\!\sup_{\tau \in \mathfrak{M}_{V} } \mathbb{E}_{v_n} \left[-w_n \gamma_n  V_{n, \tau}^2 - \int_{0}^{\tau} \!\!\!(w_n V_{n, s}^2-\beta_n)ds\right],
\end{align}
where $\mathbb{E}_{v_n} [\cdot]$ is the conditional expectation for given initial state $V_{n,0} =v_n$, 
where the supremum is taken over all stopping times $\tau$ of $V_{n,t}$, and $\gamma_n$ is defined in \eqref{gamma}. In this subsection, we focus on the case that $\beta_n$ in \eqref{eq_stop_problem} satisfies $\frac{\sigma_n ^2}{2 \theta_n} \mathbb{E} [1 - e^{-2 \theta_n Y_{n,i}}] \leq {\beta_n} < {\infty}$. 

In order to solve \eqref{eq_stop_problem} for $\theta_n < 0$, we first find a candidate solution to \eqref{eq_stop_problem} by solving a free boundary problem; then we prove that the free boundary solution is indeed the value function of \eqref{eq_stop_problem}:

The general optimal stopping theory in Chapter I of \cite{Peskir2006} tells us that the following guess of the stopping time should be optimal for Problem \eqref{eq_stop_problem}: 
\begin{align}\label{eq_optimal_stopping123}
\tau_{*} = \inf\{t \geq 0: |V_{n, t}|\geq v_* \},
\end{align}
where $v_*\geq 0$ is the optimal stopping threshold to be found.
Observe that in this guess, the continuation region $(-v_*,v_*)$ is assumed symmetric around zero. This is because the  OU process is symmetric, i.e., the process $\{-V_{n,t},t\geq 0\}$ is also an OU process started at $-{V_{n,0}} = -v_n$. Similarly, we can also argue that the value function of problem \eqref{eq_stop_problem} should be even. According to \cite[Chapter 8]{Peskir2006}, and \cite[Chapter 10]{Bernt2000},  the value function 
and the optimal stopping threshold $v_*$ should satisfy the following free boundary  problem:
\begin{align}\label{eq_free1}
&\frac{\sigma_n ^2}{2} H'' _n (v_n) -\theta_n v_n H' _n (v_n) 
=  w_n v_n ^2 - \beta_n,~ ~~ v_n \in(-v_*,v_*),\\
&H_n(\pm v_*) =  -w_n \gamma_n  v_*^2, \label{eq_free2}\\
&H' _n(\pm v_*) = \mp 2 w_n \gamma_n  v_*.\label{eq_free3}
\end{align}
In this sequel, we solve \eqref{eq_free1} to find $H(v)$.

We need to use the following indefinite integrals to solve \eqref{eq_free1} that  can be obtained by \cite[Sec. 15.3.1, (Eq. 36)] {JEFFREY1995}, \cite[Sec. 3.478 (Eq. 3), 8.250 (Eq. 1,4)] {Math_table_book}. Let $\theta_n = -\rho_n$. 
\begin{align}
&\int\!\!\!\frac{2}{\sigma_n^2} w_n {e^{-{\frac{\theta_n}{\sigma_n^2}} v_n ^2} v_n ^2} dv \nonumber\\
=& \frac{{\sqrt{\pi}} w_n {\sigma_n}}{2{\theta_n}^{\frac{3}{2}}}\text{erf}\left(\frac{\sqrt{\theta_n}}{\sigma_n} v_n \right) - \frac{w_n v_n }{\theta_n} e^{-{\frac{\theta_n} {{\sigma_n}^2}} v_n ^2} + C_1,\nonumber\\
 =& \frac{{\sqrt{\pi}} w_n {\sigma_n}}{-2j \rho_n \sqrt{\rho_n}}\text{erf}\left(\frac{j\sqrt{\rho_n}}{\sigma_n} v_n \right) + \frac{w_n v_n }{\rho_n} e^{{\frac{\rho_n} {{\sigma_n}^2}} v_n ^2} + C_1,\\
 & \int \frac{2}{\sigma_n^2} {\beta_n} {e^{-{\frac{\theta_n}{\sigma_n^2}} v_n ^2}} dv \nonumber\\
 =& \frac{{\sqrt{\pi}}{\beta_n}}{{\sigma_n}{\sqrt{\theta_n}}} \text{erf}\left(\frac{\sqrt{\theta_n}}{\sigma_n} v_n \right) +C_2, \nonumber\\
 =& \frac{{\sqrt{\pi}}{\beta_n}}{j {\sigma_n}{\sqrt{\rho_n}}} \text{erf}\left(\frac{j \sqrt{\rho_n}}{\sigma_n} v_n \right) +C_2, \label{eq_app_eq_W3}
\end{align}
\begin{align}
& {\int \text{erf}{\bigg(\frac{\sqrt{\theta_n}}{\sigma_n} v_n \bigg)} e^{\frac{\theta_n}{{\sigma_n}^2}v_n ^2}  dv} \nonumber\\
=& \frac{\sigma_n}{\sqrt{\theta_n} \sqrt{\pi}} \frac{\theta_n}{{\sigma_n}^2} {v_n ^2} {}_2F_2\left(1,1;\frac{3}{2},2;\frac{\theta_n}{\sigma_n^2}v_n ^2\right) +C, \nonumber\\
=& -\frac{\sigma_n}{j \sqrt{\rho_n} \sqrt{\pi}} \frac{\rho_n}{{\sigma_n}^2} {v_n ^2} {}_2F_2\left(1,1;\frac{3}{2},2;-\frac{\rho_n}{\sigma_n^2}v_n ^2\right) +C, 
\end{align}
\begin{align}
& \int \bigg(\frac{{\sqrt{\pi}} w_n {\sigma_n}}{2{\theta_n}^{\frac{3}{2}}} - \frac{{\sqrt{\pi}}{\beta_n}}{{\sigma_n}{\sqrt{\theta_n}}}\bigg) \text{erf}{\bigg(\frac{\sqrt{\theta_n}}{\sigma_n} v_n \bigg)} e^{\frac{\theta_n}{{\sigma_n}^2}v_n ^2}  dv \nonumber\\
= & 
\left({\frac{w_n}{2{\theta_n}}} - {\frac{\beta_n}{{\sigma_n}^2}}\right) 
{v_n ^2}~ {}_2F_2\left(1,1;\frac{3}{2},2;\frac{\theta_n}{\sigma_n^2}v_n ^2\right) +C_4, \nonumber\\
= & 
\left({-\frac{w_n}{2{\rho_n}}} - {\frac{\beta_n}{{\sigma_n}^2}}\right) 
{v_n ^2}~ {}_2F_2\left(1,1;\frac{3}{2},2;-\frac{\rho_n}{\sigma_n^2}v_n ^2\right) +C_4,
\end{align}
\begin{align}
\int C_3 e^{\frac{\theta_n}{{\sigma_n}^2}v_n ^2} dv =  & C_5\text{erfi}\left(\frac{\sqrt{\theta_n}}{\sigma_n} v_n \right) +C_6,\nonumber\\
=  & C_5\text{erfi}\left(-j\frac{\sqrt{\rho_n}}{\sigma_n} v_n \right) +C_6,
\end{align}
\begin{align}
- \int \frac{w_n v_n}{\theta_n} dv = w_n \frac{v_n ^2}{2\rho_n} + C_7,
\end{align}
where $\text{erf}(\cdot)$ and $\text{erfi}(\cdot)$ are the error function and imaginary error functions, respectively.
Hence, $H_n(v_n)$  is given  by  
\begin{align} \label{eq_W}
H_n(v_n) =& \bigg(\!\!-\frac{w_n}{2 \rho_n} - \frac{\beta_n}{\sigma_n ^2}\bigg) v_n ^2 {}_2F_2\left(1,1;\frac{3}{2},2;-\frac{\rho_n}{\sigma_n^2}v_n ^2\right) +  \nonumber\\
& C_1 \text{erfi} \bigg(j\frac{\sqrt{\rho_n}}{\sigma_n} v_n \bigg)\!+\! w_n \frac{v_n ^2}{2 \rho_n} + C_2, {\thinspace} v_n \in (-v_*, v_*),  
\end{align}
where $C_1$ and $C_2$ are constants to be found for satisfying \eqref{eq_free2}-\eqref{eq_free3}, and {erfi}$(x)$ is the imaginary error function, i.e.,
\begin{align}\label{eq_erfi}
\text{erfi}(x) = \frac{2}{\sqrt \pi} \int_0^x e^{t^2} dt.
\end{align}
Because $H_n(v_n)$ should be even but {erfi}$(x)$ is odd, we should choose $C_1=0$. Further, in order to satisfy the boundary condition \eqref{eq_free2}, $C_2$ is chosen as 
\begin{align} 
& C_2 = \nonumber\\
& \!-\!\!\frac{1}{2\rho_n}\!\EE \!\left(e^{2\rho_n Y_{n, i}} \!\right) \! v_*^2\! \! +\! \left(\!{\frac{1}{2{\rho_n}}}\! + \!{\frac{\beta_n}{{\sigma_n^2}}}\!\right)\!\! {}_2F_2\big(1,1;\frac{3}{2},2;-\frac{\rho_n}{\sigma_n^2}v_*^2\big){v_*^2}, \!\!\!
\end{align}
where we have used \eqref{gamma}.
With this, the expression of $H_n(v_n)$ is obtained in the continuation region ($-v_*$, $v_*$). In the stopping region $|v_n|\geq v_*$, the stopping time in \eqref{eq_optimal_stopping123} is simply $\tau_*=0$, because $|V_{n,0}| = |v_n|\geq v_*$. Hence, if $|v_n|\geq v_*$, the objective value  achieved by the sampling time \eqref{eq_optimal_stopping123} is
\begin{align} \label{val_new}
\!\! \mathbb{E}_{v_n} \left[-\gamma_n w_n  v_n ^2 -\! \int_{0}^{0} \!\!\!(w_n V_{n,s}^2-\beta_n)ds\right] \!= \!-\gamma_n w_n  v_n ^2 . \!\!\!\! 
\end{align} 
Combining \eqref{eq_W}-\eqref{val_new}, we obtain a candidate of the value function for \eqref{eq_stop_problem}:
\begin{align}\label{eq_W1}
& H_n(v_n)=\nonumber\\
& \left\{\!\! \begin{array}{l l}\!\!w_n \frac{v_n ^2}{2\rho_n} \!\!-\!\! \left({\frac{w_n}{2{\rho_n}}} + {\frac{\beta_n}{{\sigma_n}^2}}\right) 
\!{}_2F_2\big(1,1;\frac{3}{2},2;\!\!&\!\!\!\!-\frac{\rho_n}{\sigma_n^2}v_n ^2\big) {v_n ^2}+ C_2, \\
& \text{ if }~ |v_n|<v_*,\\
\!\!-\gamma_n w_n  v_n ^2, & \text{ if }~ |v_n|\geq v_*. \end{array}\right. \!\!\!\!\!\!
\end{align}

Next, we find a candidate value of the optimal stopping threshold $v_*$. 
By taking the gradient of $H_n (v_n)$, we get
\begin{align}\label{H_gradient}
& H’ _n (v_n) = \nonumber\\
& w_n \frac{v_n}{\rho_n} \!\!-\!\! \left(\!{\frac{w_n \sigma_n}{j \rho_n \sqrt{\rho_n}}} \!+\! {\frac{2\beta_n}{j {\sigma_n}{\sqrt{\rho_n}}}}\!\right) \!\!F\!\!\left(j \frac{\sqrt \rho_n}{\sigma_n}v_n\right)\!, v_n\!\in\!(-v_*,v_*), \!\!\!
\end{align}
where 
\begin{align}
F(x) = e^{x^2} \int_0^{x} e^{-t^2} dt.  
\end{align}
The boundary condition \eqref{eq_free3} implies that $v_*$ is the root of 
\begin{align}\label{eq_threshold1}
 w_n \frac{v_n}{\rho_n} \!\!-\!\! \left({\frac{w_n \sigma_n}{{j \rho_n \sqrt{\rho_n}}}} + {\frac{2\beta_n}{{j\sigma_n}{\sqrt{\rho_n}}}}\right) \!F\!\left(j \frac{\sqrt \rho_n}{\sigma_n}v_n \right) \!=\! - 2\gamma_n w_n v_n.
\end{align}
Substituting \eqref{gamma} into \eqref{eq_threshold1}, yields that $v_*$ is the root of 
\begin{align}\label{eq_threshold11}
& - \left({w_n \frac{\sigma_n^2}{2 \rho_n} + {\beta_n}}\right) Q{\bigg(j \frac{\sqrt{\rho_n}}{\sigma_n} v_n \bigg)} = -w_n \frac{\sigma_n^2}{2 \rho_n} \mathbb{E} [e^{2 \rho_n Y_{n,i}}],
 \end{align}
where $Q(\cdot)$ is defined in \eqref{G}. By using \eqref{KG} in \eqref{eq_threshold11}, we get that
\begin{align} \label{eq_threshold12}
& - \left({w_n \frac{\sigma_n^2}{2 \rho_n} + {\beta_n}}\right) K{\bigg(\frac{\sqrt{\rho_n}}{\sigma_n} v_n \bigg)} = -w_n \frac{\sigma_n^2}{2 \rho_n} \mathbb{E} [e^{2 \rho_n Y_{n,i}}],
\end{align}
where $K(\cdot)$ is defined in \eqref{K}.
Rearranging \eqref{eq_threshold12}, we obtain the threshold as follows
\begin{align} \label{K_threshold}
v_n (\beta_n) = \frac{\sigma_n}{\sqrt{\rho_n}} K^{-1} \bigg(\frac{w_n \frac{\sigma_n ^2}{2 \rho_n} \mathbb{E} [e^{2 \rho_n Y_{n,i}}]}{w_n \frac{\sigma_n ^2}{2 \rho_n} + \beta_n}\bigg)
\end{align}
Substituting $\rho_n = -\theta_n$ in \eqref{K_threshold},
we get \eqref{threhsold} for $\theta_n < 0$.

In addition, when $\theta_n \to 0$, \eqref{eq_threshold11} can be expressed as 
\begin{align}\label{Eq_v}
 \left({\sigma_n ^2} - \frac{2{\theta_n}{\beta_n}}{w_n}\right) Q\left(\frac{\sqrt \theta_n}{\sigma_n}v_n\right) = \sigma_n ^2\EE \left[e^{-2\theta_n Y_{n,i}}\right].
\end{align}
The error function  $\text{erf}(x)$
has a Maclaurin series representation, given by
\begin{align}
&  \text{erf}(x) = {\frac{2}{\sqrt{\pi}}} \left[x - \frac{x^3}{3} + o(x^3)\right]. 
\end{align}
Hence, the Maclaurin series representation of $Q(x)$ in \eqref{G} is
\begin{align}
& Q(x) =  1 + \frac{2x^2}{3} + o(x^2). 
\end{align}
Let $x = \frac{\sqrt{\theta_n}}{\sigma_n} v_n$, we get
\begin{align}
& Q\left(\frac{\sqrt{\theta_n}}{\sigma_n}v_n\right) = 1+ \frac{2}{3} \frac{\theta_n}{\sigma_n ^2} v_n ^2 + o(\theta_n). 
\end{align}
In addition, 
\begin{align}
\EE \left[e^{-2\theta_n Y_{n,i}}\right]  = 1-2\theta_n \EE[Y_{n,i}]+o(\theta_n).
\end{align}
Hence, \eqref{Eq_v} can be expressed as
\begin{align}\label{Eq_taylor}
& \left({\sigma_n ^2}{}-\frac{2\beta_n\theta_n}{\sigma_n ^2 w_n}\right)\! \left[1+ \frac{2}{3} \frac{\theta_n}{\sigma_n ^2} v_n ^2 + o(\theta_n)\right] \nonumber\\
=& ~~\sigma_n ^2 (1-2\theta_n \EE[Y_{n,i}]+o(\theta_n)).\!\!
\end{align}
Expanding \eqref{Eq_taylor}, yields
\begin{align}\label{Eq_taylor1}
\sigma_n ^2 2\theta_n \EE[Y_{n,i}] - \frac{2{\beta_n} \theta_n}{\sigma_n ^2 w_n}+ \frac{2}{3} \frac{\theta_n} {\sigma_n ^2} {v_n ^2} + o(\theta_n) =0.
\end{align}
Divided by $\theta_n$ and let $\theta_n \to 0$ on both sides of \eqref{Eq_taylor1},  yields
\begin{align}\label{taylor}
v_n ^2- \frac{1}{w_n} 3(\beta_n- w_n \sigma_n ^2 \EE[Y_{n,i}])  =0. 
\end{align}
Equation \eqref{taylor} has two roots $ \!-(\!1\!/\!\sqrt{w_n}) \!\sqrt{\!3(\!\beta_n \!-\! w_n \sigma_n ^2 \EE[Y_{n,i}])}$, and $(\!1\!/\!\sqrt{w_n}) \!\sqrt{\!3(\!\beta_n \!-\! w_n \sigma_n ^2 \EE[Y_{n,i}])}$. 

If $v_*= -(1/\sqrt{w_n}) \sqrt{3(\beta_n - w_n \sigma_n ^2 \EE[Y_{n,i}])}$, the free boundary problem in \eqref{eq_free1}-\eqref{eq_free3} are invalid. Hence, the root of \eqref{taylor} is $v_*= (1/\sqrt{w_n})
\sqrt{3(\beta_n- w_n \sigma_n ^2 \EE[Y_{n,i}])}$, from which \eqref{threhsold} follows for $\theta_n =0$.

\subsubsection{Verification of the Optimality of the Candidate Solution}
Next, we use It\^{o}'s formula to verify the above candidate solution is indeed optimal, as stated in the following theorem: 

\begin{theorem}\label{thm_optimal_stopping}
If $\frac{\sigma_n ^2}{2 \theta_n} \mathbb{E} [1- e^{-2 \theta_n Y_{n,i}}] \leq {\beta_n} < {\infty}$, then for all $v_n \in \mathbb{R}$, $H_n(v_n)$ in \eqref{eq_W1} is the value function of the optimal stopping problem \eqref{eq_stop_problem}. In addition,
the optimal stopping time for solving \eqref{eq_stop_problem} is $\tau_*$ in \eqref{eq_optimal_stopping123}, where $v_*=v_n(\beta_n)$ is given by \eqref{threhsold}. 
\end{theorem} 

In order to prove Theorem \ref{thm_optimal_stopping}, we need to establish the following properties of $H_n(v_n)$ in \eqref{eq_W1}, for the case that $\frac{\sigma_n ^2}{2 \theta_n} \mathbb{E} [1- e^{-2 \theta_n Y_{n,i}}] \leq {\beta_n} < {\infty}$ is satisfied in \eqref{eq_stop_problem}:
\begin{lemma}\label{lem_stop1} \cite{Ornee_TON}
$H_n(v_n) \!\!= \!\!{\mathbb{E}_{v_n}} {[-\gamma_n  V_{n,\tau_*}^2 - \int_{0}^{\tau_*} (V_{n,s}^2-\beta_n)ds]}.$
\end{lemma}


\begin{lemma}\label{lem_stop2} \cite{Ornee_TON}
$H_n(v_n) \geq {-\gamma_n}v_n ^2$ for all $v_n \in {\mathbb{R}}$.
\end{lemma}

A function $f(v)$ is said to be \emph{excessive} for the process $V_{n,t}$ if 
\begin{align}
{\mathbb{E}_{v}} f(V_{n,t}) \leq f(v), \forall t \geq 0, v\in{\mathbb{R}}.
\end{align}
By using It\^{o}'s formula in stochastic calculus, we can obtain
\begin{lemma}\label{lem_stop3} \cite{Ornee_TON}
The function $H_n(v_n)$ is excessive for the process $V_{n,t}$.
\end{lemma}

Now, we are ready to prove Theorem \ref{thm_optimal_stopping}.

\begin{proof}[Proof of Theorem \ref{thm_optimal_stopping}] In Lemmas \ref{lem_stop1}-\ref{lem_stop3}, we have shown that $H_n(v_n) = {\mathbb{E}_{v_n}} {\left[-\gamma_n  V_{n, \tau_*}^2 - \int_{0}^{\tau_*} (V_{n,s}^2-\beta_n)ds\right]}$, $H_n(v_n) \geq -\gamma_n v_n ^2$, and $H_n(v_n)$ is an excessive function. 
Moreover, from Lemma \ref{lem_stop1}, we know that $\mathbb{E}_{v_n} [\tau_* ] <\infty$ holds for all $v_n \in \mathbb{R}$. Hence, $\mathbb{P}_{v_n} (\tau_* < {\infty}) = 1$ for all $v_n \in \mathbb{R}$. 
These conditions and Theorem 1.11 in \cite[Section 1.2]{Peskir2006} imply that $\tau_*$ is an optimal stopping time of \eqref{eq_stop_problem}. This completes the proof.
\end{proof}
Because \eqref{eq_opt_stopping} is a special case of \eqref{eq_stop_problem}, we can get from Theorem \ref{thm_optimal_stopping} that 

\begin{corollary}\label{coro_stop}
If $\frac{\sigma_n ^2}{2 \theta_n} \mathbb{E} [1- e^{-2 \theta_n Y_{n,i}}] \leq {\beta_n} < {\infty}$, then a solution to \eqref{eq_opt_stopping} is $(Z_{n,1}(\beta_n), $ $Z_{n,2} (\beta_n), \ldots)$, where
\begin{align}\label{eq_optimal_stopping1234}
Z_{n,i} (\beta_n) = \inf\{t \geq 0: |O_{n, Y_{n,i}+t}|\geq v_n(\beta_n) \},
\end{align}
and $v_n(\beta_n)$ is defined  in  \eqref{threhsold}.
\end{corollary}

This concludes the proof.


\section{Proof of Theorem \ref{indexability}} \label{proof_of_indexability}

If $\gamma =0$, i.e., no sample from source $n$ is currently in service, from Proposition \ref{opt_sampler_theorem}, we get that the optimal sampling policy of \eqref{per_arm_problem} is a threshold policy which is given by \eqref{eq_opt_solution}. Given Proposition \ref{opt_sampler_theorem}, for an instantaneous estimation error $|\varepsilon_n (t) = \varepsilon|$, it is optimal not to schedule source $n$ if
\begin{align} \label{threshold_policy}
|\varepsilon| < v_n (\bar m_{n} (\lambda)),
\end{align}
where
\begin{align} \label{mse_prev}
\bar m_n (\lambda) =
\frac{\mathbb{E} \bigg[\int_{D_{n,i} (\bar m_n (\lambda))}^{D_{n,i+1} (\bar m_n (\lambda))} w_n \varepsilon_n ^2 (s) ds\bigg] + \lambda \mathbb{E} [Y_{n,i+1}]}{\mathbb{E} [D_{n,i+1} (\bar m_n (\lambda)) - D_{n,i} (\bar m_n (\lambda))]},
\end{align}
and $\bar m_n (\lambda)$ is the optimal objective value of \eqref{per_arm_problem}. We use $\bar m_{n, \text{opt}}$ as the optimal objective value in \eqref{per_arm_problem}. For convenience of the proof and to illustrate the dependency of the activation cost $\lambda$, we express it as a function of $\lambda$ in this proof.
The numerator in \eqref{mse_prev} represents the expected penalty of source $n$ strating from $i$-th delivery time to $(i+1)$-th delivery time and the denominator represents the expected time from $i$-th delivery time to the end of $(i+1)$-th delivery time. In order to prove Theorem \ref{indexability}, we need to introduce the following Lemma.

\begin{lemma} \label{monotonicity_1}
$\bar m_n (\lambda)$ is a continuous and strictly increasing function of $\lambda$.
\end{lemma}

\begin{proof}
The $(i+1)$-th delivery time from source $n$ is given by
\begin{align} \label{delivery}
D_{n,i+1} (\bar m_n (\lambda)) = S_{n, i+1} (\bar m_n (\lambda)) + Y_{n, i+1},
\end{align}
and for \eqref{eq_opt_solution}, the $(i+1)$-th sampling time is 
\begin{align} \label{sampling}
& S_{n, i+1} (\bar m_n (\lambda)) = \nonumber\\
& \inf\{t \geq D_{n, i} (\bar m_n (\lambda)) : |\varepsilon_n (t)| \geq v_n (\bar m_n (\lambda))\}.
\end{align}
Let the waiting time after the delivery of the $i$-th sample is 
\begin{align} \label{wait_time}
& Z_{n,i} (\bar m_n (\lambda)) = \nonumber\\
& \inf\{z \geq 0: |\varepsilon_n (D_{n,i} (\bar m_n (\lambda))+ z)| \geq v_n (\bar m_n (\lambda))\},
\end{align}
which represents the minimum time $z$ upto which it needs to wait after the delivery of the $i$-th sample before generating the $(i+1)$-th sample.
Hence, by using \eqref{delivery}, \eqref{sampling}, and \eqref{wait_time} the sampling time $S_{n, i} (\bar m_n (\lambda))$ and the delivery time $D_{n, i} (\bar m_n (\lambda))$ can also be expressed as
\begin{align}
& S_{n, i} (\bar m_n (\lambda)) = \sum_{j=0}^{i-1} Y_{n,j} + Z_{n,j} (\bar m_n (\lambda)) \label{app1},\\
& D_{n, i} (\bar m_n (\lambda)) = \sum_{j=0}^{i-1} Y_{n,j} + Z_{n,j} (\bar m_n (\lambda)) + Y_{n,i}. \label{app2}
\end{align}
By substituting \eqref{app1} and \eqref{app2} into \eqref{mse_prev}, we get that
\begin{align} \label{app3}
& \bar m_n (\lambda) = \nonumber\\
& \frac{\mathbb{E} \bigg[\int_{Y_{n,i}}^{Y_{n,i} + Z_{n,i} (\bar m_n (\lambda)) + Y_{n, i+1}} w_n \varepsilon_n ^2 (s) ds\bigg] + \lambda \mathbb{E} [Y_{n,i+1}]}{\mathbb{E} [Y_{n,i+1} + Z_{n,i} (\bar m_n (\lambda))]}.
\end{align}
The optimal objective value $ \bar m_n (\lambda)$ in \eqref{app3} is exactly equal to the root of the following equation:
\begin{align} \label{fixed_pt}
f(\beta_n) + \lambda \mathbb{E} [Y_{n,i+1}] =0,
\end{align}
where
\begin{align}
f(\beta_n) = & \mathbb{E} \bigg[\int_{Y_{n,i}}^{Y_{n,i} + Z_{n,i} (\beta_n) + Y_{n, i+1}} w_n \varepsilon_n ^2 (s) ds\bigg] \nonumber\\
& - \beta_n \mathbb{E} [Z_{n,i} (\beta_n) + Y_{n,i+1}].
\end{align}
Because $f(\beta_n)$ is a concave, continuous, and strictly decreasing function of $\beta_n$ \cite[Lemma 2]{Ornee_TON}, from \eqref{fixed_pt}, it is evident that the root of \eqref{fixed_pt} is unique and continuous in $\lambda$. Hence, $\bar m_n (\lambda)$ is unique and continuous in $\lambda$. From \eqref{fixed_pt}, we get that
\begin{align} \label{app4}
f(\beta_n) = -\lambda \mathbb{E} [Y_{n, i+1}].
\end{align}
For any $0 \leq \lambda_1 \leq \lambda_2$ and $\beta_n = \bar m_n (\lambda)$, from \eqref{app4}, we have
\begin{align}
& f(\bar m_n (\lambda_1)) = -\lambda_1 \mathbb{E} [Y_{n,i+1}], \\
& f(\bar m_n (\lambda_2)) = -\lambda_2 \mathbb{E} [Y_{n,i+1}].
\end{align}
As $f(\beta_n)$ is a continuous and strictly decreasing function of $\beta_n$, for any non-negative $\lambda_2 > \lambda_1$ implies
$\bar m_n (\lambda_1) < \bar m_n (\lambda_2)$. Therefore, $\bar m_n (\lambda)$ is continuous and strictly increasing function of $\lambda$.
\end{proof}

The next task is to show the properties of the threshold $v_n (\bar m_n (\lambda))$ in \eqref{threshold_policy} by using the \eqref{threhsold} for the three cases of Gauss-Markov processes. In that sequel, we need to use the following lemma.

\begin{lemma} \label{lemma_threshold}
The threshold $v_n (\bar m_n (\lambda))$ is continuous and strictly increasing in $\lambda$ irrespective of the signal structure.
\end{lemma}

\begin{proof}
For $\theta_n > 0$, $v_n (\bar m_n (\lambda))$ is as follows
\begin{align}
v_n (\bar m_n (\lambda)) = \frac{\sigma_n}{\sqrt{\theta_n}} Q^{-1} \bigg(\frac{w_n \frac{\sigma_n ^2}{2 \theta_n} \mathbb{E} [e^{-2 \theta_n Y_{n,i}}]}{w_n \frac{\sigma_n ^2}{2 \theta_n} - \bar m_n (\lambda)}\bigg).
\end{align}
The derivative of $v_n (\bar m_n (\lambda))$ is given by
\begin{align} \label{v'}
v'_n (\bar m_n (\lambda)) = \frac{\sigma_n}{\sqrt{\theta_n}} \bigg\{Q^{-1} \bigg(\frac{w_n \frac{\sigma_n ^2}{2 \theta_n} \mathbb{E} [e^{-2 \theta_n Y_{n,i}}]}{w_n \frac{\sigma_n ^2}{2 \theta_n} - \bar m_n (\lambda)}\bigg)\bigg\}'.
\end{align} 
Let 
\begin{align}
Q^{-1} \bigg(\frac{w_n \frac{\sigma_n ^2}{2 \theta_n} \mathbb{E} [e^{-2 \theta_n Y_{n,i}}]}{w_n \frac{\sigma_n ^2}{2 \theta_n} - \bar m_n (\lambda)}\bigg) = y.
\end{align}
By using the property of derivative of an inverse function \cite{strang1991calculus}, $v'_n (\bar m_n (\lambda))$ in \eqref{v'} can be expressed as
\begin{align}
v'_n (\bar m_n (\lambda)) = \frac{\sigma_n}{\sqrt{\theta_n}} \frac{1}{Q’ (y)} \frac{w_n \frac{\sigma_n ^2}{2 \theta_n} \mathbb{E} [e^{-2 \theta_n Y_{n,i}}]}{(w_n \frac{\sigma_n ^2}{2 \theta_n} - \bar m_n (\lambda))^2},
\end{align}
where $Q’(x)$ is as follows
\begin{align}
Q’(x) = -\frac{\sqrt{\pi}}{2} \frac{e^{x^2}}{x^2} \text{erf} (x) + \sqrt{\pi} e^{x^2} \text{erf} (x) + \frac{1}{x} > 0,  
\end{align}
for all $x > 0$. Hence, by using Lemma \ref{monotonicity_1} and the fact that $v'_n (\bar m_n (\lambda)) > 0$, it is proved that $v_n (\bar m_n (\lambda))$ is a strictly increasing function of $\lambda$.

In addition, for $\theta_n < 0$, $v_n (\bar m_n (\lambda))$ can be expressed as
\begin{align}
v_n (\bar m_n (\lambda)) = \frac{\sigma_n}{\sqrt{-\theta_n}} K^{-1} \bigg(\frac{w_n \frac{\sigma_n ^2}{2 \theta_n} \mathbb{E} [e^{-2 \theta_n Y_{n,i}}]}{w_n \frac{\sigma_n ^2}{2 \theta_n} - \bar m_n (\lambda)}\bigg).
\end{align}
The derivative of $v_n (\bar m_n (\lambda))$ is then given by
\begin{align} \label{v'}
v'_n (\bar m_n (\lambda)) = \frac{\sigma_n}{\sqrt{-\theta_n}} \bigg\{K^{-1} \bigg(\frac{w_n \frac{\sigma_n ^2}{2 \theta_n} \mathbb{E} [e^{-2 \theta_n Y_{n,i}}]}{w_n \frac{\sigma_n ^2}{2 \theta_n} - \bar m_n (\lambda)}\bigg)\bigg\}'.
\end{align} 
Let 
\begin{align}
K^{-1} \bigg(\frac{w_n \frac{\sigma_n ^2}{2 \theta_n} \mathbb{E} [e^{-2 \theta_n Y_{n,i}}]}{w_n \frac{\sigma_n ^2}{2 \theta_n} - \bar m_n (\lambda)}\bigg) = p.
\end{align}
Utilizing the property of the derivative of an inverse function, $v'_n (\bar m_n (\lambda))$ in \eqref{v'} can be expressed as
\begin{align}
v'_n (\bar m_n (\lambda)) = \frac{\sigma_n}{\sqrt{-\theta_n}} \frac{1}{K' (p)} \frac{w_n \frac{\sigma_n ^2}{2 \theta_n} \mathbb{E} [e^{-2 \theta_n Y_{n,i}}]}{(w_n \frac{\sigma_n ^2}{2 \theta_n} - \bar m_n (\lambda))^2},
\end{align}
where $K'(x)$ is as follows
\begin{align}
K'(x) = -\frac{\sqrt{\pi}}{2} \frac{e^{-x^2}}{x^2} \text{erfi} (x) - \sqrt{\pi} e^{-x^2} \text{erfi} (x) + \frac{1}{x} < 0,  
\end{align}
for all $x > 0$. Hence, from Lemma \ref{monotonicity_1} and the fact that $v'_n (\bar m_n (\lambda)) > 0$, it is proved that $v_n (\bar m_n (\lambda))$ is a strictly increasing function of $\lambda$. By using similar proof arguments, the result can be proven for $\theta_n =0$. Hence, combining the results for all of the three cases of $\theta_n$, Lemma \ref{lemma_threshold} is proven.
\end{proof}



From \eqref{index_set}, the set $\Psi_n (\lambda)$ is
\begin{align} \label{index_set_app}
\Psi_n (\lambda) = \{(\varepsilon, \gamma) : \gamma > 0 {\thinspace} {\text{or}} {\thinspace} |\varepsilon| < v_n (\bar m_n (\lambda))\}. 
\end{align}
Moreover, for a given $\varepsilon$, if $\gamma =0$ and $\varepsilon \in \Psi_n (\lambda_1)$, then
\begin{align}
|\varepsilon| < v_n (\bar m_n (\lambda_1)).
\end{align}
Because Lemma \ref{lemma_threshold} implies that $v_n (\bar m_n (\lambda))$ is continuous and strictly increasing in $\lambda$, we get that for $\gamma=0$, $(\varepsilon, \gamma) \in \Psi_n (\lambda_2)$ for any $\lambda_1 < \lambda_2$. Hence, $\Psi_n (\lambda_1) \subseteq \Psi_n (\lambda_2)$. Thus from the definition of indexability, the bandit $n$ is indexable for all $n$. This completes the proof.

\section{Proof of Theorem \ref{Whittle index}} \label{proof_of_Whittle_index}

Substituting \eqref{index_set} into the definition of Whittle index in \eqref{Whittle_index_general_definition}, we obtain that
\begin{align} \label{Whitt_2}
\alpha_n (\varepsilon, \gamma) = \inf_{\lambda} \{\lambda \in \mathbb{R} : \gamma > 0 {\thinspace} \text{or} {\thinspace} |\varepsilon| < v_n (\lambda)\}.
\end{align}
First, we consider the case when $\gamma =0$. By using Lemma \ref{lemma_threshold}, \eqref{Whitt_2} implies that the Whittle index $\alpha_n (\varepsilon, \gamma)$ is unique and it satisfies the following at $\lambda = \alpha_n (\varepsilon, \gamma)$:
\begin{align} \label{Whitt_3}
|\varepsilon| = v_n (\bar m_n (\alpha_n (\varepsilon, \gamma))),
\end{align}
where $v_n (\cdot)$ is defined in \eqref{threhsold}. First, consider the case of stable OU process (i.e., $\theta_n > 0$). Substituting \eqref{threhsold} for $\theta_n > 0$ into \eqref{Whitt_3}, we get that
\begin{align}
& |\varepsilon| = \frac{{\sigma_n}}{\sqrt \theta_n} Q^{-1} \frac{w_n \frac{\sigma_n^2}{2 \theta_n} \mathbb{E} [e^{-2 \theta_n Y_{n,i}}]}{w_n \frac{\sigma_n^2}{2 \theta_n} - {\bar m}_n (\alpha_n (\varepsilon, \gamma))},
\end{align}
which implies
\begin{align}
& Q \left(\frac{\sqrt{\theta_n}}{\sigma_n} |\varepsilon|\right) = \frac{w_n \frac{\sigma_n^2}{2 \theta_n} \mathbb{E} [e^{-2 \theta_n Y_{n,i}}]}{w_n \frac{\sigma_n^2}{2 \theta_n} - \bar m_n (\alpha_n (\varepsilon, \gamma))}. \label{W3}
\end{align}
After some rearrangements, \eqref{W3} becomes
\begin{align} \label{Whitt_4}
\bar m_n (\alpha_n (\varepsilon, \gamma)) = \frac{w_n \frac{\sigma_n^2}{2 \theta_n} \left(Q\left(\frac{\sqrt{\theta_n}}{\sigma_n} |\varepsilon|\right) - \mathbb{E} [e^{-2 \theta_n Y_{n,i}}]\right)}{Q\left(\frac{\sqrt{\theta_n}}{\sigma_n} |\varepsilon|\right)}. 
\end{align}

The optimal objective value $\bar m_n (\alpha_n (\varepsilon))$ to problem \eqref{per_arm_problem}  is defined by
\begin{align} \label{mse_Whitt}
& \bar m_n (\alpha_n (\varepsilon, \gamma))= \nonumber\\
& \frac{\mathbb{E} \!\bigg[\!\!\int_{D_{n,i} (\bar m_n (\alpha_n (\varepsilon, \gamma)))}^{D_{n,i+1} (\bar m_n (\alpha_n (\varepsilon, \gamma)))} \!\!w_n \varepsilon_n ^2 (s) ds\!\bigg] \!\!+\!\! \alpha_n\! (\!\varepsilon, \!\gamma) \mathbb{E} [Y_{n,i+1}]}{\mathbb{E} [D_{n,i+1} (\bar m_n (\alpha_n (\varepsilon, \gamma))) - D_{n,i} (\bar m_n (\alpha_n (\varepsilon, \gamma)))]}.
\end{align}
Substituting \eqref{mse_Whitt} into \eqref{Whitt_4} implies
\begin{align}
&  \frac{\mathbb{E} \!\bigg[\!\!\int_{D_{n,i} (\bar m_n (\alpha_n (\varepsilon, \gamma)))}^{D_{n,i+1} (\bar m_n (\alpha_n (\varepsilon, \gamma)))} \!\!w_n \varepsilon_n ^2 (s) ds\!\bigg] \!\!+\!\! \alpha_n\! (\!\varepsilon, \!\gamma) \mathbb{E} [Y_{n,i+1}]}{\mathbb{E} [D_{n,i+1} (\bar m_n (\alpha_n (\varepsilon, \gamma))) - D_{n,i} (\bar m_n (\alpha_n (\varepsilon, \gamma)))]} \nonumber\\ 
&= \frac{w_n \frac{\sigma_n^2}{2 \theta_n} \left(Q\left(\frac{\sqrt{\theta_n}}{\sigma_n} |\varepsilon|\right) - \mathbb{E} [e^{-2 \theta_n Y_{n,i}}]\right)}{Q\left(\frac{\sqrt{\theta_n}}{\sigma_n} |\varepsilon|\right)},
\end{align}
which yields
\begin{align} \label{proof_Whitt_1}
& \mathbb{E} \!\bigg[\!\!\int_{D_{n,i} (\bar m_n (\alpha_n (\varepsilon, \gamma)))}^{D_{n,i+1} (\bar m_n (\alpha_n (\varepsilon, \gamma)))} \!\!w_n \varepsilon_n ^2 (s) ds\!\bigg] \!\!+\!\! \alpha_n\! (\!\varepsilon, \!\gamma) \mathbb{E} [Y_{n,i+1}] = \nonumber\\
& \mathbb{E} [D_{n,i+1} (\bar m_n (\alpha_n (\varepsilon, \gamma))) - D_{n,i} (\bar m_n (\alpha_n (\varepsilon, \gamma)))] \nonumber\\
& \!\frac{w_n \!\frac{\sigma_n^2}{2 \theta_n}\!\! \!\left(\!Q\!\!\left(\!\frac{\sqrt{\theta_n}}{\sigma_n} \!|\varepsilon|\!\right)\!\! -\!\! \mathbb{E} [e^{\!-2 \theta_n \!Y_{n,i}}\!]\!\right)}{Q\left(\frac{\sqrt{\theta_n}}{\sigma_n} |\varepsilon|\right)}.
\end{align}
After rearranging \eqref{proof_Whitt_1}, we get that
\begin{align} \label{stable_index}
& \alpha_n (\varepsilon, \gamma) = \nonumber\\
& \frac{1}{\mathbb{E} [Y_{n,i}]} \bigg\{\mathbb{E} [D_{n,i+1} (\bar m_n (\alpha_n (\varepsilon, \gamma))) - D_{n,i} (\bar m_n (\alpha_n (\varepsilon, \gamma)))] \nonumber\\
& ~~~~~~~~~~~~\frac{w_n \frac{\sigma_n^2}{2 \theta_n} \left(Q\left(\frac{\sqrt{\theta_n}}{\sigma_n} |\varepsilon|\right) - \mathbb{E} [e^{-2 \theta_n Y_{n,i}}]\right)}{Q\left(\frac{\sqrt{\theta_n}}{\sigma_n} |\varepsilon|\right)} \nonumber\\
& ~~~~~~~~~~~~- \mathbb{E} \!\bigg[\!\!\int_{D_{n,i} (\bar m_n (\alpha_n (\varepsilon, \gamma)))}^{D_{n,i+1} (\bar m_n (\alpha_n (\varepsilon, \gamma)))} \!\!\!\!\!\!\!\!\!\!\!\!w_n \varepsilon_n ^2 (s) ds\!\bigg] \bigg\},
\end{align}
where \eqref{stable_index} holds because $Y_{n,i}$'s are \emph{i.i.d.}. In Theorem \ref{Whittle index}, we have used $D_{n,i} (\varepsilon)$ to represent $D_{n,i} (\varepsilon, \bar m_n (\alpha_n (\varepsilon, \gamma)))$ because $D_{n,i}$ is dependent on $m_n (\alpha_n (\varepsilon, \gamma))$ through the state $\varepsilon$. 

In addition, for unstable OU process, when $\theta_n < 0$, by using the similar proof arguments, we can prove \eqref{Whitt_index_unstable} for $\theta_n < 0$.

Subsequently, when $\theta_n =0$, by using \eqref{threhsold}, \eqref{Whittle_index_general_definition}, and Lemma \ref{lemma_threshold}, at $\lambda = \alpha_n (\varepsilon, \gamma)$, we get that
\begin{align}
& |\varepsilon| = \frac{1}{\sqrt{w_n}} \sqrt{3 (\bar m_n (\alpha_n (\varepsilon, \gamma)) - w_n \sigma_n ^2 \mathbb{E} [Y_{n,i}])}, 
\end{align}
which implies
\begin{align} \label{proof_Whitt_2}
& w_n \varepsilon^2 = 3 (\bar m_n (\alpha_n (\varepsilon, \gamma)) - w_n \sigma_n ^2 \mathbb{E} [Y_{n,i}]).
\end{align}
After some rearrangements, \eqref{proof_Whitt_2} becomes
\begin{align}
& w_n \bigg(\frac{\varepsilon^2}{3} + \sigma_n ^2 \mathbb{E} [Y_{n,i}]\bigg) = \bar m_n (\alpha_n (\varepsilon, \gamma)). \label{corollary_eq1}
\end{align}
Substituting \eqref{mse_Whitt} into \eqref{corollary_eq1}, we obtain that
\begin{align}
&  w_n \bigg(\frac{\varepsilon^2}{3} + \sigma_n ^2 \mathbb{E} [Y_{n,i}]\bigg) = \nonumber\\
& \frac{\mathbb{E} \!\bigg[\!\!\int_{D_{n,i} (\bar m_n (\alpha_n (\varepsilon, \gamma)))}^{D_{n,i+1} (\bar m_n (\alpha_n (\varepsilon, \gamma)))} \!\!w_n \varepsilon_n ^2 (s) ds\!\bigg] \!\!+\!\! \alpha_n\! (\!\varepsilon, \!\gamma) \mathbb{E} [Y_{n,i+1}]}{\mathbb{E} [D_{n,i+1} (\bar m_n (\alpha_n (\varepsilon, \gamma))) - D_{n,i} (\bar m_n (\alpha_n (\varepsilon, \gamma)))]}, \label{corollary_eq2}
\end{align}
which yields
\begin{align}
& w_n \mathbb{E} \!\bigg[\!\!\int_{D_{n,i} (\bar m_n (\alpha_n (\varepsilon, \gamma)))}^{D_{n,i+1} (\bar m_n (\alpha_n (\varepsilon, \gamma)))} \!\!w_n \varepsilon_n ^2 (s) ds\!\bigg] \!\!+\!\! \alpha_n\! (\!\varepsilon, \!\gamma) \mathbb{E} [Y_{n,i+1}] \nonumber\\
=& w_n \mathbb{E} [D_{n,i+1} (\bar m_n (\alpha_n (\varepsilon, \gamma))) - D_{n,i} (\bar m_n (\alpha_n (\varepsilon, \gamma)))] \nonumber\\
&\bigg(\frac{\varepsilon^2}{3} + \sigma_n ^2 \mathbb{E} [Y_{n,i}]\bigg),
\end{align}
from which \eqref{wiener_index} follows for $\theta_n =0$. 

Next, we consider $\gamma > 0$. From Definition \ref{def_1} and Definition \ref{def_2}, the possible infimum cost $\lambda$  for which to activate and not to activate are equally desirable is $-\infty$, from which \eqref{Whitt_infty} follows. This concludes the proof.


\section{Proof of Lemma \ref{lemma1}} \label{proof_lemma1}
In order to prove Lemma \ref{lemma1}, we need to consider the following 
two cases:  

\emph{Case 1:} If $|\varepsilon_n (D_{n,i})| = |O_{n, D_{n,i}-S_{n,i}} | = |O_{n, Y_{n,i}}| \geq |\varepsilon|$, then $S_{n,i+1} = D_{n,i}$. Hence, 
\begin{align} \label{eqn2}
\mathbb{E} [D_{n,i+1} (\varepsilon) - D_{n,i} (\varepsilon)] = & \mathbb{E} [S_{n,i+1} (\varepsilon) + Y_{n,i+1} - S_{n, i+1} (\varepsilon)], \nonumber\\
=& \mathbb{E} [Y_{n,i+1}].
\end{align}
Using the fact that the $Y_{n,i}$'s are independent of the OU process, we can obtain
\begin{align}\label{eq_expectation_9}
& \mathbb{E} \bigg[\int_{D_{n,i} (\varepsilon)}^{D_{n,i+1} (\varepsilon)} \varepsilon_n ^2 (s) ds \bigg| O_{n, Y_{n,i}}, |O_{n, Y_{n,i}} | \geq |\varepsilon|\bigg] \nonumber\\
=& \mathbb{E} \bigg[\int_{Y_{n,i}}^{Y_{n,i}+ Y_{n,i+1}} O_{n,s}^2 ds \bigg| O_{n, Y_{n,i}}, |O_{n, Y_{n,i}} | \geq |\varepsilon|\bigg] \nonumber\\
=& \mathbb{E} \bigg[\int_{0}^{Y_{n,i}+ Y_{n,i+1}} O_{n,s}^2 ds \bigg| O_{n, Y_{n,i}}, |O_{n, Y_{n,i}} | \geq |\varepsilon|\bigg] \nonumber\\
& -  \mathbb{E} \bigg[\int_{0}^{Y_{n,i}} O_{n,s}^2 ds \bigg| O_{n, Y_{n,i}}, |O_{n, Y_{n,i}} | \geq |\varepsilon|\bigg].
\end{align}
By invoking Lemma \ref{lem_stop}, we get that
\begin{align} \label{lemma1_new_1}
& \mathbb{E} \bigg[\int_{0}^{Y_{n,i}+ Y_{n,i+1}} O_{n,s}^2 ds \bigg|O_{n, Y_{n,i}}, |O_{n, Y_{n,i}} | \geq |\varepsilon|\bigg] \nonumber\\
= & R_{n,2} (O_{n, Y_{n,i} + Y_{n, i+1}}),
\end{align}
\begin{align} \label{lemma1_new_2}
\mathbb{E} \bigg[\int_{0}^{Y_{n,i}} O_{n,s}^2 ds \bigg| O_{n, Y_{n,i}}, |O_{n, Y_{n,i}} | \geq |\varepsilon|\bigg] =  R_{n,2} (O_{n, Y_{n,i}}).
\end{align}
Substituting \eqref{lemma1_new_1} and \eqref{lemma1_new_2} into \eqref{eq_expectation_9}, it becomes
\begin{align}
& \mathbb{E} \bigg[\int_{D_{n,i} (\varepsilon)}^{D_{n,i+1} (\varepsilon)} \varepsilon_n ^2 (s) ds \bigg| O_{n, Y_{n,i}}, |O_{n, Y_{n,i}} | \geq |\varepsilon|\bigg] \nonumber\\ 
=& R_{n,2} (O_{n, Y_{n,i} + Y_{n, i+1}}) - R_{n,2} (O_{n, Y_{n,i}}) \nonumber\\
=& R_{n,2} (|\varepsilon| + O_{n, Y_{n, i+1}}) - R_{n,2} (O_{n, Y_{n,i}}), \label{lemma1_new_3}
\end{align}
where \eqref{lemma1_new_3} holds because at $t=D_{n,i} (\varepsilon)$, the estimation error $O_{n, Y_{n,i}}$ reaches the threshold $|\varepsilon|$.

\emph{Case 2:} If $|\varepsilon_n (D_{n,i})|= |O_{n, Y_{n,i}} | < |\varepsilon|$, then,  almost surely,
\begin{align}
|\varepsilon_n (S_{n, i+1}) | = |\varepsilon|.
\end{align}
Then,
\begin{align} \label{proof_lemma1_1}
& \mathbb{E} [D_{n, i+1} (\varepsilon) - D_{n,i} (\varepsilon)] \nonumber\\
= & \mathbb{E} [D_{n, i+1} (\varepsilon) - S_{n, i+1} (\varepsilon) + S_{n, i+1} (\varepsilon) - S_{n,i} (\varepsilon) \nonumber\\
& ~~~+ S_{n,i} (\varepsilon) -D_{n,i} (\varepsilon)]. 
\end{align}
Because $D_{n, i+1} (\varepsilon) = S_{n, i+1} (\varepsilon) + Y_{n, i+1}$, by invoking Lemma \ref{lem_stop}, we can obtain the remaining expectations in \eqref{proof_lemma1_1} which are given by
\begin{align}\label{eq_expectation_2}
&\mathbb{E}\left[S_{n,i+1} (\varepsilon)- S_{n,i} (\varepsilon)  \Big|O_{n, Y_{n,i}}, |O_{n, Y_{n,i}} |< |\varepsilon|\right] = R_{n,1} (\varepsilon),\\
&\mathbb{E}\left[D_{n,i} (\varepsilon) \!-\! S_{n,i} (\varepsilon) \Big| O_{n, Y_{n,i}}, |O_{n, Y_{n,i}} |\!<\! |\varepsilon|\right] \!\!=\!\! R_{n,1} (O_{n, Y_{n,i}}).\label{eq_expectation_3}
\end{align}
Using \eqref{eq_expectation_2} and \eqref{eq_expectation_3}, we get that
\begin{align}\label{eq_expectation_7}
& \mathbb{E}\left[D_{n, i+1} (\varepsilon) - D_{n,i} (\varepsilon) \Big|O_{n, Y_{n,i}}, |O_{n, Y_{n,i}} |< |\varepsilon|\right] \nonumber\\
=& \mathbb{E} [Y_{n,i+1}] + R_{n,1} (\varepsilon) - R_{n,1} (O_{n, Y_{n,i}}).
\end{align}
In addition, 
\begin{align} \label{eqn1}
& \mathbb{E} \bigg[\int_{D_{n,i} (\varepsilon)}^{D_{n,i+1} (\varepsilon)} \varepsilon_n ^2 (s) ds \bigg| O_{n, Y_{n,i}}, |O_{n, Y_{n,i}} | < |\varepsilon|\bigg] \nonumber\\
=& \mathbb{E} \bigg[\int_{Y_{n,i}}^{Y_{n,i} + Z_{n,i} (\varepsilon) + Y_{n,i+1}} O_{n,s} ^2 ds\bigg] \nonumber\\
=& \mathbb{E} \bigg[\int_{0}^{Y_{n,i} + Z_{n,i} (\varepsilon) + Y_{n, i+1}} O_{n,s} ^2 ds\bigg] 
- \mathbb{E} \bigg[\int_{0}^{Y_{n,i} } O_{n,s} ^2 ds\bigg].
\end{align}
By invoking Lemma \ref{lem_stop} again, we can obtain
\begin{align}\label{eq_expectation_4}
&\mathbb{E}\left[\int_{0}^{Y_{n,i} + Z_{n,i} (\varepsilon) + Y_{n, i+1}} O_{n,s}^2 ds \bigg|O_{n, Y_{n,i}}, |O_{n, Y_{n,i}} |< |\varepsilon|\right] \nonumber\\
=& R_{n,2} (O_{n, Y_{n,i} + Z_{n,i} (\varepsilon) + Y_{n, i+1}}),\\
&\mathbb{E}\left[\!\int_{0}^{Y_{n,i}} \!\!\!O_{n,s}^2 ds \bigg| O_{n, Y_{n,i}}, |O_{n, Y_{n,i}} |< |\varepsilon|\right] = R_{n,2} (O_{n, Y_{n,i}}).\label{eq_expectation_5}
\end{align}
By using \eqref{eq_expectation_4} and \eqref{eq_expectation_5} in \eqref{eqn1}, we have
\begin{align}
& \mathbb{E} \bigg[\int_{D_{n,i} (\varepsilon)}^{D_{n,i+1} (\varepsilon)} \varepsilon_n ^2 (s) ds \bigg| O_{n, Y_{n,i}}, |O_{n, Y_{n,i}} | < |\varepsilon|\bigg] \nonumber\\
=& R_{n,2} (O_{n, Y_{n,i} + Z_{n,i} (\varepsilon) + Y_{n, i+1}}) - R_{n,2} (O_{n, Y_{n,i}}) \nonumber\\
=& R_{n,2} (|O_{n, Y_{n,i}}| + O_{n, Y_{n, i+1}}) - R_{n,2} (O_{n, Y_{n,i}}),\label{lemma1_new_4}
\end{align}
where \eqref{lemma1_new_4} holds because at $t=D_{n,i} (\varepsilon)$, the estimation error $O_{n, Y_{n,i}}$ is below the threshold $|\varepsilon|$.

By combining \eqref{eqn2} and \eqref{eq_expectation_7} of the two cases, yields
\begin{align}\label{eq_expectation_8}
& \mathbb{E}\left[D_{n,i+1} (\varepsilon) - D_{n,i} (\varepsilon) \Big|O_{n, Y_{n,i}}\right] \nonumber\\
= & \max\{R_{n,1} (|\varepsilon|) - R_{n,1} (O_{n, Y_{n,i}}),0\} + \mathbb{E} [Y_{n,i+1}]. 
\end{align}
By taking the expectation over $O_{n, Y_{n,i}}$ in \eqref{eq_expectation_8} gives
\begin{align}
&\mathbb{E} \bigg[\mathbb{E}\left[D_{n,i+1} (\varepsilon) - D_{n,i} (\varepsilon) \Big|O_{n, Y_{n,i}}\right]\bigg] \nonumber\\
= & \mathbb{E} [\max\{R_{n,1} (|\varepsilon|) - R_{n,1} (O_{n, Y_{n,i}}),0\} 
+ Y_{n,i+1}] \nonumber\\
= & \mathbb{E} [\max\{R_{n,1} (|\varepsilon|) - R_{n,1} (O_{n, Y_{n,i}}),0\} 
+ R_{n,1} (O_{n, Y_{n,i}})], \label{proof_lemma1_2}
\end{align}
where \eqref{proof_lemma1_2} follows from the fact that $Y_{n,i}$'s are \emph{i.i.d.} and \eqref{eq_stop11} in Lemma \ref{lem_stop}. Because $R_{n,1} (\cdot)$ is an even function, form \eqref{proof_lemma1_2} we get that
\begin{align} \label{eq_expectation_last}
& \mathbb{E}\left[D_{n,i+1} (\varepsilon) - D_{n,i} (\varepsilon) \big| O_{n, Y_{n,i}}\right] \nonumber\\
= & R_{n,1} (\max\{|\varepsilon|, |O_{n, Y_{n,i}}|\}).
\end{align}
Similarly, by combining \eqref{lemma1_new_3} and \eqref{lemma1_new_4}  of the two cases, yields
\begin{align}\label{eq_expectation_11}
&\mathbb{E} \bigg[\int_{D_{n,i} (\varepsilon)}^{D_{n,i+1} (\varepsilon)} \varepsilon_n ^2 (s) ds\bigg|O_{n, Y_{n,i}}\bigg] \nonumber\\
= & R_{n,2} (\max\{|\varepsilon|, |O_{n, Y_{n,i}}|\} + O_{n, Y_{n, i+1}}) - R_{n,2} (O_{n, Y_{n,i}}).
\end{align}
Finally, by taking the expectation over $O_{n, Y_{n,i}}$ in \eqref{eq_expectation_last} and \eqref{eq_expectation_11} and using the fact that $R_{n,1} (\cdot)$ and $R_{n,2}(\cdot)$ are even functions, Lemma \ref{lemma1} is proven.

\section{Proof of Lemma \ref{dummy_lemma}} \label{proof_dummy}

Because $\lambda$ represents the cost to activate an arm, it is optimal in \eqref{problem_dummy} to activate a \emph{dummy bandit} only when $\lambda < 0$. Conversely, when $\lambda \geq 0$, it is optimal not to activate the dummy bandit. Hence, from Definition \ref{def_1}, the \emph{dummy bandits} are always indexable. In addition, from Definition \ref{def_2} and the fact that the \emph{dummy bandits} are activated only when $\lambda < 0$, we get $\alpha_0 (\varepsilon, \gamma) =0$. 

\section{Proof of Theorem \ref{single_theorem_whittle}} \label{proof_of_single_whittle}

In order to prove Theorem \ref{single_theorem_whittle}, we first show that \eqref{single_solution} and \eqref{single_solution_whittle} are equivalent to each other. For single source, the source weight $w_1 = 1$ and the transmission cost $\lambda = 0$. 

We first show the proof for stable OU process, i.e., for $\theta_1 > 0$. When $\varepsilon = v_1 (\beta_1)$ in \eqref{single_threshold}, we have
\begin{align}
Q\left(\frac{\sqrt{\theta_1}}{\sigma_1} v_1 (\beta_1)\right) \!\!= &Q\left(\frac{\sqrt{\theta_1}}{\sigma_1} \frac{\sigma_1}{\sqrt{\theta_1}} Q^{-1} \bigg(\frac{\frac{\sigma_1 ^2}{2 \theta_1} \mathbb{E} [e^{-2 \theta_1 Y_{1,i}}]}{\frac{\sigma_1 ^2}{2 \theta_1} - \beta_1}\bigg)\right) \nonumber\\
=& \frac{\sigma_1 ^2 \mathbb{E} [e^{-2 \theta_1 Y_{1,i}}]}{\sigma_1 ^2 - 2 \theta_1 \beta_1}. \label{proof_thm4_1}
\end{align}
Substituting \eqref{proof_thm4_1} into \eqref{Whittle_index} when $\gamma=0$ for single source results, 
\begin{align}
& \alpha_1 (v_1 (\beta_1), 0)\! = \nonumber\\
& \!\!\frac{1}{\mathbb{E} [Y_{1,i}]}\!\bigg\{\!\mathbb{E} [D_{1, i+1} (\varepsilon)\! -\! D_{1,i} (\varepsilon)] \frac{\sigma_1^2}{2 \theta_1}\! \bigg(\!1 \!-\! \frac{\mathbb{E} [e^{-2 \theta_1 Y_{1,i}}]}{\frac{\sigma_1 ^2 \mathbb{E} [e^{-2 \theta_1 Y_{1,i}}]}{\sigma_1 ^2 - 2 \theta_1 \beta_1}}  \bigg) \nonumber\\
&~~~~~~~~~- \mathbb{E} \bigg[\int_{D_{1,i} (\varepsilon)}^{D_{1, i+1} (\varepsilon)} \varepsilon_{1} ^2 (s) ds\bigg]\bigg\},
\end{align}
which becomes
\begin{align}
& \alpha_1 (v_1 (\beta_1), 0) = \nonumber\\
& \frac{1}{\mathbb{E} [Y_{1,i}]} \!\bigg\{\!\mathbb{E} [D_{1, i+1} (\varepsilon) \!-\! D_{1,i} (\varepsilon)] \beta_1 \!-\! \mathbb{E} \bigg[\!\!\int_{D_{1,i} (\varepsilon)}^{D_{1, i+1} (\varepsilon)} \!\!\!\varepsilon_{1} ^2 (s) ds\!\bigg]\!\bigg\}. \label{proof_thm4_2}
\end{align}
The parameter $\beta_1$ in \eqref{proof_thm4_2} can be found from \eqref{single_beta} and \eqref{single_obj}, which is exactly equal to the optimal objective value ${\bar m}_{1, \text{opt}}$. Hence, substituting $\beta_1$ into \eqref{proof_thm4_2} yields
\begin{align}
\alpha_1 (v_1 (\beta_1), 0) =0.
\end{align}
If $\varepsilon > v_1 (\beta_1)$, as $Q(x)$ is a strictly increasing function in $[x, \infty)$, we have
\begin{align}
Q\left(\frac{\sqrt{\theta_1}}{\sigma_1} \varepsilon\right) > Q\left(\frac{\sqrt{\theta_1}}{\sigma_1} v_1 (\beta_1)\right),
\end{align}
which yields
\begin{align}
\bigg(1 - \frac{\mathbb{E} [e^{-2 \theta_1 Y_{1,i}}]}{Q\left(\frac{\sqrt{\theta_1}}{\sigma_1} \varepsilon\right)}  \bigg) > \bigg(1 - \frac{\mathbb{E} [e^{-2 \theta_1 Y_{1,i}}]}{Q\left(\frac{\sqrt{\theta_1}}{\sigma_1} v_1 (\beta_1)\right)}  \bigg).
\end{align}
From the above arguments, it is proved that $\alpha_1 (\varepsilon, 0) > \alpha_1 (v_1 (\beta_1), 0) =0$ for $\varepsilon > v_1 (\beta_1)$. 
Similarly, as $Q(x)$ is an even function, for $\varepsilon < v_1 (\beta_1)$, it holds that $\alpha_1 (\varepsilon, 0) < \alpha_1 (v_1 (\beta_1), 0) =0$. 

By using the similar proof arguments we can show that for all $\theta_1$, $\alpha_1 (v_1 (\beta_1), 0) = 0$, $\alpha_1 (\varepsilon, 0) > 0$ for $\varepsilon > v_1 (\beta_1)$, and $\alpha_1 (\varepsilon, 0) < 0$ for $\varepsilon < v_1 (\beta_1)$. Hence, the two statements in \eqref{single_solution} and \eqref{single_solution_whittle} are equivalent to each other. This result also illustrate in Fig. \ref{fig_discussion} from which it is evident that $\alpha_1 (\varepsilon, 0)$ is an even function. We prove the optimality of Proposition \ref{opt_sampler_theorem} for any number of sources in Appendix \ref{threshold_proof}. Hence, Theorem \ref{single_theorem_whittle} is also optimal. This completes the proof.

\section{Proof of Theorem \ref{aoi_indexability}} \label{proof_of_aoi_indexability}

 
If $\gamma=0$, i.e., no sample from source is currently in service, from Proposition \ref{opt_sampler_theorem_age}, we get that for an AoI $\Delta_n (t) = \delta$, it is optimal not to schedule source $n$ if
\begin{align} \label{threshold_policy_aoi}
\mathbb{E} [p(\delta + Y_{n, i+1})] < \bar m_{n, \text{age}} (\lambda),
\end{align}
where
\begin{align} \label{mse_prev_aoi}
\bar m_{n, \text{age}} (\lambda) \!\!=\!\!
\frac{\mathbb{E} \!\bigg[\!\!\int_{D_{n,i} \!(\bar m_{n, \text{age}} (\lambda)\!)}^{D_{n,i+1} \!(\bar m_{n, \text{age}} (\lambda)\!)} \!\!\!w_n p_n \!(\!\Delta_n (s)\!) ds\!\bigg] \!+\! \lambda \mathbb{E} [Y_{n,i+1}]}{\mathbb{E} [D_{n,i+1} (\bar m_{n, \text{age}} (\lambda)) - D_{n,i} (\bar m_{n, \text{age}} (\lambda))]},
\end{align}
and $\bar m_{n, \text{age}} (\lambda)$ is the optimal objective value to \eqref{per_arm_problem_age}. We use $\bar m_{n, \text{age-opt}}$ as the optimal objective value in \eqref{per_arm_problem_age}. For convenience of the proof and to illustrate the dependency of the activation cost $\lambda$, we express it as a function of $\lambda$ in the rest of the proofs.

According to \eqref{aoi_penalty}, \eqref{aoi_penalty_1}, and Lemma \ref{monotonicity_1}, $\bar m_{n, \text{age}} (\lambda)$ is a continuous and strictly increasing function of $\lambda$.

By utilizing \eqref{def_3}, 
if $\gamma=0$, for a given $\delta$, if $(\delta, \gamma) \in \Psi_{n, \text{age}} (\lambda_1)$, then
\begin{align}
\mathbb{E} [p_n (\delta + Y_{n,i+1})] < \bar m_{n, \text{age}} (\lambda_1).
\end{align}
By using the fact that $\bar m_{n, \text{age}} (\lambda)$ is continuous and strictly increasing in $\lambda$, we get that $(\delta, \gamma) \in \Psi_{n, \text{age}} (\lambda_2)$ for any $\lambda_1 < \lambda_2$. Hence, $\Psi_{n, \text{age}} (\lambda_1) \subseteq \Psi_{n, \text{age}} (\lambda_2)$. Thus from the definition of indexability, the bandit $n$ is indexable for all $n$. This concludes the proof.

\section{Proof of Theorem \ref{theorem3}} \label{proof of age-based index}
When $\gamma =0$, the Whittle index $\alpha_{n, \text{age}} (\delta, \gamma)$ of bandit $n$ at state $(\delta, \gamma)$ is given by
\begin{align} \label{Whitt_1_aoi}
\alpha_{n, \text{age}} (\delta, \gamma) = \inf_{\lambda} \{\lambda \in \mathbb{R} : (\delta, \gamma) \in \Psi_{n, \text{age}} (\lambda)\}.
\end{align}
By utilizing \eqref{def_3} in \eqref{Whitt_1_aoi}, we obtain that
\begin{align} \label{Whitt_2_aoi}
& \alpha_{n, \text{age}} (\delta, \gamma) = \nonumber\\
& \inf_{\lambda} \{\lambda \in \mathbb{R} : \gamma>0 {\thinspace} {\text{or}} {\thinspace} \mathbb{E} [p_n (\delta + Y_{n,i+1})] < \bar m_{n, \text{age}} (\lambda)\}.
\end{align}
At $\lambda = \alpha_{n, \text{age}} (\delta, \gamma)$, we have
\begin{align} \label{proof_theorem6_1}
& w_n \mathbb{E} [p_n (\delta + Y_{n, i+1})] = \bar m_{n, \text{age}} (\alpha_n (\delta, \gamma)).
\end{align}
Substituting \eqref{mse_prev_aoi} into \eqref{proof_theorem6_1}, we get that
\begin{align}
& w_n \mathbb{E} [p_n (\delta + Y_{n, i+1})] = \nonumber\\
& \frac{\mathbb{E} \bigg[\int_{D_{n,i} (\bar m_{n, \text{age}} (\lambda))}^{D_{n,i+1} (\bar m_{n, \text{age}} (\lambda))} w_n O_{n,s}^2 ds\bigg] + \alpha_{n, \text{age}} (\delta, \gamma) \mathbb{E} [Y_{n,i+1}]}{\mathbb{E} [D_{n,i+1} (\bar m_{n, \text{age}} (\lambda)) - D_{n,i} (\bar m_{n, \text{age}} (\lambda))]},
\end{align}
which yields
\begin{align}
& w_n \mathbb{E} [p_n (\delta + Y_{n, i+1})] \mathbb{E} [D_{n,i+1} (\bar m_{n, \text{age}} (\lambda)) - D_{n,i} (\bar m_{n, \text{age}})] \nonumber\\
=& \mathbb{E} \bigg[\!\!\int_{D_{n,i} (\bar m_{n, \text{age}} (\lambda))}^{D_{n,i+1} (\bar m_{n, \text{age}} (\lambda))} \!\!\!\!\!\!\! w_n p_n (\Delta_n (s)) ds\bigg] \!\!+\! \alpha_{n, \text{age}} (\delta, \gamma) \mathbb{E} [Y_{n,i+1}],
\end{align}
from which \eqref{age-based_index} follows because the $Y_{n,i}$'s are \emph{i.i.d.}. 

When $\gamma > 0$, from \eqref{def_3} and \eqref{Whitt_2_aoi}, \eqref{Whitt_infty_age} yields. This completes the proof.

\section{Proof of Lemma \ref{aoi_lemma}} \label{proof_aoi_lemma}

In order to prove Lemma \ref{aoi_lemma}, we need to consider the following two cases:

\emph{Case 1:} If $\mathbb{E} \left[p_n (\Delta_n (D_{n,i} + Y_{n, i+1}) \right] = \mathbb{E} \left[p_n (\Delta_n (D_{n,i}) + Y_{n, i+1}) \right] = \mathbb{E} \left[p_n (Y_{n,i} + Y_{n, i+1}) \right] \geq \mathbb{E} \left[p_n (\delta + Y_{n, i+1}) \right]$, then, $S_{n,i+1} = D_{n,i}$. Hence,
\begin{align} 
&\mathbb{E} [D_{n,i+1} (\delta) - D_{n,i} (\delta)\big| \delta, Y_{n,i}] \nonumber\\
= & \mathbb{E} [S_{n,i+1} (\delta) + Y_{n,i+1} - S_{n, i+1} (\delta)\big| \delta, Y_{n,i}], \nonumber\\
=& \mathbb{E} [Y_{n,i+1}] = \mathbb{E} [Y_{n,i}], \label{proof_aoi_lemma1}
\end{align}
where \eqref{proof_aoi_lemma1} holds because the $Y_{n,i}$'s are \emph{i.i.d.}. In addition, 
\begin{align}\label{eq_expectation_aoi}
& \mathbb{E} \bigg[\!\!\int_{D_{n,i} (\delta)}^{D_{n,i+1} (\delta)} \!\!\!p_n (s) ds \bigg| \delta, Y_{n,i}\bigg] \nonumber\\
=& \mathbb{E} \bigg[\int_{Y_{n,i}}^{Y_{n,i}+ Y_{n,i+1}} p_n (s) ds \bigg| \delta, Y_{n,i}\bigg] \nonumber\\
=& \mathbb{E} \bigg[\!\!\int_{0}^{Y_{n,i}+ Y_{n,i+1}} \!\!\!\!\!p_n (s) ds \bigg| \delta, Y_{n,i}\bigg] \!-\! \mathbb{E} \bigg[\!\!\int_{0}^{Y_{n,i}} \!\!p_n (s) ds \bigg| \delta, Y_{n,i}\bigg] \nonumber\\
=& R_{n,3} (Y_{n, i} + Y_{n, i+1}) - R_{n,3} (Y_{n,i}).
\end{align}

\emph{Case 2:} If $\mathbb{E} \left[p_n (\Delta_n (D_{n,i} + Y_{n, i+1}) \right] =  \mathbb{E} \left[p_n (Y_{n,i} + Y_{n, i+1}) \right] < \mathbb{E} \left[p_n (\delta + Y_{n, i+1}) \right]$, then, almost surely,
\begin{align}
\mathbb{E} \left[p_n (\Delta_n (S_{n,i+1} + Y_{n, i+1}) \right] = \mathbb{E} \left[p_n (\delta + Y_{n, i+1}) \right].
\end{align}
Then,
\begin{align} \label{proof_lemmaaoi_1}
\mathbb{E} [D_{n, i+1} \!(\delta)\! \!-\! D_{n,i} \!(\delta)\!\big| \delta, Y_{n,i}]
\!=\! \mathbb{E} [Z_{n,i} (\delta) \!+\! Y_{n, i+1} \big| \delta, Y_{n,i}] \!=\! \delta. 
\end{align}
In addition, 
\begin{align} \label{proof_lemma_aoi_2}
& \mathbb{E} \bigg[\int_{D_{n,i} (\delta)}^{D_{n,i+1} (\delta)} p_n (s) ds \bigg| \delta, Y_{n,i}\bigg] \nonumber\\
=& \mathbb{E} \bigg[\int_{Y_{n,i}}^{Y_{n,i} + Z_{n,i} (\delta) + Y_{n,i+1}} p_n (s) ds\bigg] \nonumber\\
=& \mathbb{E} \bigg[\int_{0}^{Y_{n,i} + Z_{n,i} (\delta) + Y_{n,i+1}} p_n (s) ds\bigg] 
- \mathbb{E} \bigg[\int_{0}^{Y_{n,i} } p_n (s) ds\bigg] \nonumber\\
=& R_{n,3} (\delta + Y_{n, i+1}) - R_{n,3} (Y_{n,i}).
\end{align}
By combining \eqref{proof_aoi_lemma1} and \eqref{proof_lemmaaoi_1} of the two cases, yields
\begin{align}\label{proof_aoi_lemma12}
\mathbb{E}\left[D_{n,i+1} (\delta) - D_{n,i} (\delta) \Big|\delta, Y_{n,i} \right]
= & \max\{\delta, Y_{n,i}\}. 
\end{align}
Similarly, by combining \eqref{eq_expectation_aoi} and \eqref{proof_lemma_aoi_2}  of the two cases, yields
\begin{align} \label{proof_aoi_lemma13}
& \mathbb{E} \bigg[\int_{D_{n,i} (\delta)}^{D_{n,i+1} (\delta)} p_n (s) ds\bigg|Y_{n,i}\bigg] \nonumber\\
= & R_{n,3} (\max\{\delta, Y_{n,i}\} + Y_{n, i+1}) - R_{n,3} (Y_{n,i})
\end{align}
Finally, by taking the expectation over $Y_{n,i}$ in \eqref{proof_aoi_lemma12} and \eqref{proof_aoi_lemma13}, Lemma \ref{aoi_lemma} is proven.

\end{document}